\documentclass[reqno]{amsart}
\usepackage[latin,english]{babel}
\usepackage[utf8]{inputenc}
\usepackage[T1]{fontenc}
\usepackage{cancel}
\usepackage{eufrak}
\usepackage{geometry}
\usepackage{amsmath}
\usepackage{amssymb}
\usepackage{amsbsy}
\usepackage{amsthm}
\usepackage[showonlyrefs]{mathtools}
\usepackage{mathabx, mathrsfs, dsfont}
\usepackage{cases}
\usepackage{braket}
\usepackage[toc,page]{appendix}
\usepackage{pdfsync}
\usepackage{hyperref}
\usepackage{graphicx}
\graphicspath{{images/}}
\usepackage{psfrag}
\usepackage{epstopdf}
\usepackage{todonotes}
\usepackage{marginnote}
\usepackage{graphicx}
\usepackage{fancyhdr}
\usepackage{math}
\usepackage{cite}
\usepackage{color}
\usepackage{comment}
\usepackage{subcaption}

\newcommand{\bkappa}{{\boldsymbol{\kappa}}}
\newcommand{\btau}{{\boldsymbol{\tau}}}

\newtheorem{lemma}{Lemma}[section]
\newtheorem{theorem}[lemma]{Theorem}
\newtheorem{proposition}[lemma]{Proposition}

\newtheorem{example}[lemma]{Example}
\newtheorem{definition}[lemma]{Definition}

\numberwithin{equation}{section}

\title{Correlation functions for a chain of short range oscillators}
\author{T. Grava}
\address{SISSA, via Bonomea 265, 34136 Trieste, Italy and School of Mathematics, University of Bristol, UK}
\email{grava@sissa.it}
\author{T. Kriecherbauer}
\address{Department of Mathematics, Universitat Bayreuth, Germany}
\email{thomas.kriecherbauer@uni-bayreuth.de}
\author{G. Mazzuca}
\address{SISSA, via Bonomea 265, 34136 Trieste, Italy }
\email{guido.mazzuca@sissa.it} 
\author{K. D. T.-R. McLaughlin}
\address{Department of Mathematics, Colorado State University, 1874 campus delivery, Fort Collins, CO 80523}
\email{kenmcl@rams.colostate.edu}

\date{\today}

\date{\today}

\begin{document}
	\maketitle
	\begin{abstract}
	We consider a system of harmonic oscillators with short range interactions and we study their correlation  functions when the initial data is sampled with respect to the Gibbs measure.
	Such correlation functions display rapid oscillations that travel through the chain.
	We show  that   the correlation functions always have %generically 
	two  %distinguished 
	fastest peaks which move in  opposite directions and decay at rate $t^{-\frac{1}{3}}$  for position and momentum correlations and as $t^{-\frac{2}{3}}$ for energy correlations.  The 
	shape of these peaks is asymptotically described  by  the Airy function. Furthermore, the correlation functions have some non generic peaks with lower decay rates. In particular, there are 
	peaks which decay at rate 
	$t^{-\frac{1}{4}}$ for position and momentum correlators and with rate $t^{-\frac{1}{2}}$  for energy correlators. The shape of these peaks is described by the Pearcey integral. Crucial for our analysis is an appropriate generalisation of spacings, i.e.~differences of the positions of neighbouring particles, that are used as spatial variables in the case of nearest neighbour interactions. Using the theory of circulant matrices we are able to introduce a quantity that retains both localisation and analytic viability. This also allows us to define and analyse some additional quantities used for nearest neighbour chains.
	Finally, we study numerically  the evolution of the correlation functions 
	after adding nonlinear perturbations
 to our model. Within the time range of our numerical simulations the asymptotic description of the linear case  seems to persist for small nonlinear perturbations while stronger nonlinearities change shape and decay rates of the peaks significantly.
	\end{abstract}
	
	\section{Introduction}
	%\todo{add that at rest each particle is at equal distance.}

	In this manuscript we consider a system of $N=2M+1$ particles interacting with a short range harmonic potential with
	 Hamiltonian of the form
	\begin{equation}
	\label{H0}
H=\sum_{j=0}^{N-1}\frac{p_j^2}{2}+\sum_{s=1}^m\frac{\kappa_s}{2}\sum_{j=0}^{N-1}(q_j-q_{j+s})^2\,,
\end{equation}
where $1\leq m\ll N$, $\kappa_1>0$, $\kappa_m>0$, and $\kappa_s\geq 0$ for $1 < s < m$. In order to make sense of \eqref{H0} we need to introduce boundary conditions. Throughout this paper we consider periodic boundary conditions. By that we mean that the indices $j$ are taken from $\Z / N\Z$ and therefore
%We consider periodic boundary conditions
\[
q_{N+j}=q_j,\quad p_{N+j}=p_j
\]
holds for all $j$.
%At rest the particles are placed at equal distance $L/N$. 
%Alternatively (see e.g.\cite{Spohn2014})  the boundary condition $q_{N+1}=q_1+L$ can be  considered.
%The periodic boundary condition is  recovered by  change  of coordinates $q_j\to q_j-\frac{L}{N}(j-1)$.
The  Hamiltonian \eqref{H0} can be rewritten in the form 
			\begin{equation}
	\label{eq:A}
	H(\bp,\bq) := \frac{1}{2} \langle\bp, \bp\rangle +\frac{1}{2}\langle \bq,A \bq\rangle,
	\end{equation}
	where $\bp=(p_0,\dots,p_{N-1}) $, $\bq=(q_0,\dots,q_{N-1}),$  $\langle\,.,\,.\rangle$ denotes the standard scalar product in $\R^N$ and
	where $A\in\mbox{Mat}(N,\R)$ is a positive semidefinite  symmetric   circulant matrix generated by the vector  $\ba=(a_0,\dots,a_{N-1})$ namely $A_{kj}=a_{(j-k)\mbox{\tiny {mod $N$}}}$ or
	\begin{equation}
	\label{A}
			A = {\begin{bmatrix}
				a_{0}&  a_{1}&\dots & a_{N-2}& a_{N-1}
				\\
				a_{N-1}& a_{0}& a_{1}&& a_{N-2}
				\\
				\vdots & a_{N-1}& a_{0}&\ddots &\vdots 
				\\
				a_{2}&&\ddots &\ddots & a_{1}
				\\
				a_{1}& a_{2}&\dots & a_{N-1}& a_{0}
				\\
				\end{bmatrix}}\, ,
		\end{equation}
		where 
\begin{equation}
\label{v_a}
\begin{split}
&a_0=2\sum_{s=1}^m\kappa_s,\quad  a_s=a_{N-s}=-\kappa_s,\quad \mbox{ for $s=1,\dots,m$ and $a_s=0$ otherwise.}
\end{split}
\end{equation}
Due to the condition $\kappa_1 >0$ we have $\langle \bq,A \bq\rangle =0$ iff all spacings $q_{j+1}-q_j$ vanish. Therefore the kernel of $A$ is one-dimensional with the constant vector $(1,\ldots, 1)^\intercal$ providing a basis. This also implies that the lattice at rest has zero spacings everywhere. Observe, however, that one may introduce an arbitrary spacing $\Delta$ for the lattice at rest by the canonical transformation $Q_j = q_j + j\Delta$, $P_j = p_j$ which does not change the dynamics. The periodicity condition for the positions $Q_j$ then reads $Q_{N+j}=Q_j + L$ with $L=N \Delta$ (see e.g.\cite[Sec.~2]{Spohn2014}). 

%We require that the kernel of $A$ is  one-dimensional, so that the lattice does not decouple into independent sublattices. This is always achieved if $\kappa_1\neq 0$ \todo{Sure??}.

The harmonic oscillator with  only  nearest neighbour interactions  is recovered by choosing 
\[
a_0=2\kappa_1,\quad a_1=a_{N-1}=-\kappa_1,
\]
and the remaining coefficients  are set to zero.

The equations of motion  for the Hamiltonian $H$  take the form
\[
\dfrac{d^2}{dt^2}q_j=\sum_{s=1}^m \kappa_s (q_{j+s}-2q_j+q_{j-s}),\quad j\in\Z / N\Z.
\]
The integration is obtained  by studying the dynamics in Fourier space  (see e.g. \cite{Lukkarinen}). 
 In this paper we study correlations between momentum, position and  local versions of energy. Following the standard procedure in the case of nearest neighbour interactions we replace the vector of position $\bq$ by a new variable $\br$ so that the Hamiltonian takes the form
 \[
H=\frac{1}{2}\langle \bp,\bp\rangle+\frac{1}{2}\langle\br,\br\rangle.
\] 
Such a change of variables may be achieved by any linear transformation
\begin{equation}
\label{r}
\br=T\bq,
\end{equation}  
with an $N\times N$ matrix $T$ that satisfies 
\begin{equation}
	\label{eq:T0}
		A = T^\intercal T,
	\end{equation}
where $ T^\intercal$ denotes the transpose of $T$. In the case of nearest neighbour interactions one may choose $r_j = \sqrt{\kappa_1} (q_{j+1}-q_j)$ corresponding to a circulant matrix $T$ generated by the vector $\boldsymbol{\tau}=\sqrt{\kappa_1}(-1, 1,0,\dots,0)$. We show in Proposition~\ref{prop:matrix_root} below that short range interactions given by matrices $A$ of the form~\eqref{A},~\eqref{v_a} also admit such a {\em localized square root}. More precisely, there exists a circulant $N\times N$ matrix $T$
 of the form
	\begin{equation}\label{form:T}
			T ={\begin{bmatrix}
				\tau_0&  \tau_{1}&\dots&  \tau_{m}&0 &\dots&0	\\			
				0& \tau_{0}& \tau_{1}&\dots&\tau_m&0&				\\
				&\ddots & \ddots& \ddots &\ddots &&\\
%				&&&&&&&&&\\
%	&&&&&&&&&\\
%	&&&&&&&&&\\
%	&&&&&&&&&\\
%	&&&&&&&&&\\ 
                                 \tau_m&0&\ddots & \ddots& \ddots &\ddots &\\
                                 				&\ddots & \ddots& \ddots &\ddots & \ddots&\ddots\\
				\tau_2&\dots&\tau_m&0&\dots&\tau_0&\tau_1\\
				\tau_1&\tau_2&\dots&\tau_m &0&0&\tau_{0}
				\\
				\end{bmatrix}}\, .
		\end{equation}
that satisfies \eqref{eq:T0}. The crucial point here is that $T$ is not the standard (symmetric) square root of the positive semidefinite matrix $A$ but a localized version generated by some vector $\boldsymbol{\tau}$
with zero entries everywhere, except possibly in the first $m+1$ components. Hence the $j$-th component of the {\em generalized elongation} $\br$ defined through \eqref{r} depends only on the components $q_s$ with $s=j,j+1,\ldots,j+m$.
It is worth noting that $\boldsymbol{1}=(1,\ldots, 1)^\intercal$ satisfies $T \boldsymbol{1}=0$ since $\langle \boldsymbol{1}, A\boldsymbol{1}\rangle =0$. This implies
$$\sum_{s=0}^m\tau_s=0\,, \quad 
r_j= \sum_{s=1}^{m} \tau_s (q_{j+s}-q_j)\,  \quad \text{and} \quad \sum_{j=0}^{N-1}r_j= (1,\ldots, 1) T \bq = 0.$$

The local energy $e_j$ takes the form
\[
e_j=\dfrac{1}{2}p_j^2+\dfrac{1}{2}r_j^2\,.
\]
The goal  of  this manuscript is to study the  behaviour of the correlation functions for the momentum $p_j$, the generalized elongation $r_j$ and the local energy $e_j$ 
when $N\to\infty$ and $t\to\infty$. 
Due to the spatial translation invariance of the Hamiltonian $H(\bp,\bq) = H(\bp,\bq+\lambda \boldsymbol{1})$, $\lambda \in \R$, that corresponds to the conservation of total momentum, we 
%need to 
reduce the Hamiltonian system by one degree of freedom to obtain a normalizable Gibbs measure. This leads to the reduced phase space

%In order to define a convergent Gibbs measure for our system 
% introduce  the reduced phase space
		\begin{equation}
	\label{eq:phase_space}
	\cM := \left\{(\bp,\bq)\in \R^{N}\times \R^{N} \, : \, \sum_{k=0}^{N-1} p_k =0\, ;\, \sum_{k=0}^{N-1} q_k = 0 \right\}.
	\end{equation}
	We endow   $\cM$  with the Gibbs measure   at temperature $\beta^{-1}$, namely:
	\begin{equation}
	\label{eq:measure}
	\di \mu = Z_N(\beta)^{-1} \delta\left(\sum_{k=0}^{N-1} p_k\right)\delta\left(\sum_{k=0}^{N-1} q_k \right)e^{-\beta H(\bp,\bq)}\di \bp \di \bq
	\end{equation}
	where $Z_N(\beta)$ is the norming constant and $\delta(x)$ is the delta function. 

For convenience we introduce the vector
\[
\boldsymbol{u}(j,t)=(r_j(t),p_j(t),e_j(t)).
\]
We consider the  correlation functions 
	\begin{equation}
	\label{eq:correlation_functions}
	\begin{split}
&S^N_{\alpha\alpha'}(j,t) = \la u_{\alpha}(j,t)u_{\alpha'}(0,0) \ra -  \la u_{\alpha}(j,t)\ra\la u_{\alpha'}(0,0) \ra,\;\;\alpha,\alpha'=1,2,3, \\
	\end{split}
	\end{equation}
	where  the symbol $\la\,.\,\ra$ refers to averages with respect to $d\mu$ .		We calculate the limits 
	\[
	\lim_{N\to\infty}S^N_{\alpha\alpha'}(j,t)=S_{\alpha\alpha'}(j,t)\,.
	\]
	For the harmonic oscillator with nearest neighbor interactions such limits have  been calculated in   \cite{Mazur}.
%	The goal of the present manuscript is to show that one can tune the  harmonic  interactions $\kappa_1,\dots,\kappa_m$ in such a way that the decay in time of the correlation
%	functions is 
	
	In an interesting series of papers, (see e.g. \cite{Spohn2016},  and also the collection \cite{LLP2016}) several researchers have considered the evolution of space-time correlation functions, for "anharmonic chains", which are nonlinear nearest-neighbor Hamiltonian systems of oscillators. The authors consider the deterministic evolution from random initial data sampled from a Gibbs ensemble, with a large number of particles  and study the correlation functions $S_{\alpha \alpha'}^{N}$.  

In addition to intensive computational simulations \cite{KD2016}, \cite{SpohnMendl17},   Spohn and collaborators  also propose and study a nonlinear stochastic conservation law model \cite{Spohn2014},\cite{Spohn2016}. Using deep physical intuition, it has been proposed that the long-time behaviour of space-time correlation functions of the deterministic Hamiltonian evolution from random initial data is equivalent to the behaviour of correlation functions of an analogous nonlinear stochastic system of PDEs.  Studying this stochastic model, Spohn eventually arrives at an asymptotic description of the "sound peaks" of the  correlation functions in normal modes coordinates which are related to $S_{\alpha\alpha'}$ by orthogonal transformation:
\begin{eqnarray}
\label{TW}
\tilde{S}_{\alpha \alpha} \cong \left( \lambda_{s} t \right)^{-2/3} f_{\mbox{KPZ}} \left( ( \lambda_{s} t )^{-2/3} ( x - \alpha c t ) \right) \ ,
\end{eqnarray}
using the notation of [Formula (3.1)]\cite{Spohn2014}. Here $  f_{\mbox{KPZ}}$ is a universal function that first emerges in the Kardar-Parisi-Zhang equation and it is  related to the Tracy-Widom distribution, \cite{TW}, (for a review see \cite{CorwinReviewPaper} and also \cite{Kriech}).
%
%In \cite{SpohnMendl15} the authors proposed that the time-integrated currents should exhibit fluctuations about their mean described by the Tracy-Widom distribution, again based on the use of the nonlinear stochastic pde system as a model for the deterministic evolution from random initial data.  As one example, they consider the quantity
%\begin{eqnarray}
%\Phi(x,t)=\int_{0}^{t} \frak{j}(x,t') dt' - \int_{0}^{x} \frak{u}(x',0)dx' \ ,
%\end{eqnarray}
%where $\frak{u}(x,t)$ arises as a sort of continuum limit of a particle system obeying a discrete analogue of a system of  conservation laws taking the form $\partial_{t} \frak{u}(x,t) + \partial_{x} \frak{j}(x,t) = 0$, in which $\frak{j}(x,t)$ is a local current density for $\frak{u}(x,t)$.  The authors suggest a dual interpretation of $\Phi(x,t)$ as the height function from a KPZ equation, and thus arrive at the proposal that 
%\begin{eqnarray}
%\Phi(x,t) \simeq a_{0} t + \left( \Gamma t \right)^{1/3} \xi_{TW} \ ,
%\end{eqnarray}
%where $a_{0}$ and $\Gamma$ are model-dependent parameters, and $\xi_{TW}$ is a random amplitude with Tracy-Widom distribution.
%
A common element to the above cited papers is the observation that such formulae should hold for non-integrable dynamics, while the correlation functions of integrable lattices of oscillators will exhibit {\textit{ballistic scaling}}, which means the correlation functions decay as $\frac{1}{t}$ for $t$ large.  For example, in \cite{KD2016} the authors present the results of simulations of the Toda lattice in 3 different asymptotic regimes (the harmonic oscillator limit, the hard-particle limit, and the full nonlinear system).  They present plots of the quantity $t S(x,t)$ as a function of the scaled spatial variable $x/t$ (here $S(x,t)$ represents any of the correlation functions).  The numerical results support the ballistic scaling conjecture in some of the asymptotic scaling regimes.
Further analysis in \cite{Spohn_Toda} gives a derivation of the ballistic scaling for the Toda lattice.
The decay of equilibrium correlation functions show similar  features as  anomalous heat transport  in one-dimensional systems \cite{Dhar},\cite{Lepri_Book}\cite{Onorato}  which leads to conjecture that the two phenomena are related  \cite{LLP2016}.

%\textcolor{red}{
In \cite{SpohnMendl15} the authors also pursue a different connection to random matrices, and in particular to the Tracy-Widom distribution.  Over the last 15 years, there has emerged a story originating in the proof that for the totally asymmetric exclusion process on a 1-D lattice (TASEP), the fluctuations of the height function are governed (in a suitable limit) by the Tracy-Widom distribution.  Separately, a partial differential equations model for these fluctuations emerged, which takes the form of a stochastic Burgers equation:
\begin{eqnarray}
\label{eq:StochBurgers}
\frac{\partial u}{\partial t} = \nu \frac{\partial^{2} u}{\partial x^{2} } - \lambda u \frac{\partial u}{\partial x} + \frac{\partial \zeta}{\partial x} \ ,
\end{eqnarray}
where $\zeta$ is a stationary spatio-temporal white noise process.  (The mean behaviour of TASEP is actually described by the simpler Euler equation $ \displaystyle 
\frac{\partial u}{\partial t} =  - \lambda u \frac{\partial u}{\partial x}$.). From these origins there have now emerged proofs, for a small collection of initial conditions, that the fluctuations of the solution to (\ref{eq:StochBurgers}) are indeed connected to the Tracy-Widom distribution (see \cite{CorwinReviewPaper} and the references contained therein).  In \cite{SpohnMendl15}, the authors considered continuum limits of anharmonic lattices with random initial data, in which there are underlying conservation laws describing the mean behaviour that are the analogue of the Euler equation associated to (\ref{eq:StochBurgers}).  By analogy with the connection between TASEP and (\ref{eq:StochBurgers}), they proposed that the time-integrated currents are the analogue of the height function, and should exhibit fluctuations about their mean described by the Tracy-Widom distribution, again based on the use of the nonlinear stochastic pde system as a model for the deterministic evolution from random initial data.  As one example, they consider the quantity
\begin{eqnarray}
\label{Spohn1}
\Phi(x,t)=\int_{0}^{t} \frak{j}(x,t') dt' - \int_{0}^{x} \frak{u}(x',0)dx' \ ,
\end{eqnarray}
where $\frak{u}(x,t)$ arises as a sort of continuum limit of a particle system obeying a discrete analogue of a system of  conservation laws taking the form $\partial_{t} \frak{u}(x,t) + \partial_{x} \frak{j}(x,t) = 0$, in which $\frak{j}(x,t)$ is a local current density for $\frak{u}(x,t)$.  The authors suggest a dual interpretation of $\Phi(x,t)$ as the height function from a KPZ equation, and thus arrive at the proposal that 
\begin{eqnarray}
\label{Spohn2}
\Phi(x,t) \simeq a_{0} t + \left( \Gamma t \right)^{1/3} \xi_{TW} \ ,
\end{eqnarray}
where $a_{0}$ and $\Gamma$ are model-dependent parameters, and $\xi_{TW}$ is a random amplitude with Tracy-Widom distribution.
%}

Our main result is the  analogue of the relations  \eqref{TW}  for  the harmonic oscillator with 
short range interactions and \eqref{Spohn2} for the harmonic oscillator.
	%  Since the distribution is Gaussian,
	%the energy correlation $S_{33}(j,t)$ can be be expressed as quadratic function  of the momentum and elongation correlation $S_{\alpha\alpha'}(j,t)$, $\alpha=1,2$.
	For  stating our result,  we first calculate the dispersion relation  $|\omega(k)|$  for the harmonic oscillator with short range interaction in the limit $N\to\infty$ obtaining
	\begin{equation}
	\label{dispersion0}
	f(k)=|\omega(k)|=\sqrt{2\sum_{s=1}^m\kappa_s\left(1-\cos(2\pi k s)\right)}\,,
	\end{equation}
see \eqref{eq:general_dispersion}. The points $k= 0,1$ contribute to the fastest moving peaks  of the correlation functions  that have a  velocity $\pm v_0$  where  $v_0= \sqrt{\sum_{s=1}^ms^2\kappa_s}=f'(0)/(2\pi)$.
 If $f''(k)<0$ for all $0<k\leq 1/2$ then as $t \to \infty$ the following holds uniformly in $j \in \Z$ (cf.~Theorem~\ref{th:theorem_slow}  and  Figure~\ref{Figure1}):
\begin{equation}
\begin{split}
\label{Airy0}
S_{\alpha\alpha'}(j,t)&= \frac{1}{2 \beta \lambda_0 t^{1/3}} \left[ (-1)^{\alpha+\alpha'}\mbox{Ai}\left(\dfrac{j-v_0t}{ \lambda_0 t^{1/3}}\right)+ \mbox{Ai}\left(-\dfrac{j+v_0t}{ \lambda_0 t^{1/3}}\right) \right]+\cO  \left(t^{-1/2}\right),\quad \alpha,\alpha'=1,2\\
S_{33}(j,t)&= \frac{1}{2 \beta^2 \lambda_0^2 t^{2/3}} \left[ \mbox{Ai}^2\left(\dfrac{j-v_0t}{ \lambda_0 t^{1/3}}\right)+ \mbox{Ai}^2\left(-\dfrac{j+v_0t}{ \lambda_0 t^{1/3}}\right) \right]+\cO  \left(t^{-5/6}\right)\,,
\end{split}
\end{equation}
where  $\mbox{Ai}(w) = \frac{1}{\pi} \int_0^\infty \cos(y^3/3+wy) d y$, $w \in \R$, is the Airy function, and  $\lambda_0 := \frac{1}{2}\Big(\frac{1}{v_0} \sum_{s=1}^m s^4\kappa_s \Big)^{1/3}.$
  The above formula is the linear analogue of the Tracy-Widom distribution in \eqref{TW}.
  % however it is not clear to us whether the  ballistic scaling would  hold in the almost linear regime like for example in the Toda lattice with weak couplings.
%	
%	 When $m=1$ one recover the decay of the harmonic oscillator.
%
%	\[
%	S_{\alpha\alpha'}(j,t)= \mathcal{O}\left(\frac{1}{\sqrt{t}}\right),\;\;\alpha,\alpha'=1,2\quad S_{33}(j,t)= \mathcal{O}\left(\frac{1}{t}\right)\quad \mbox{ as $t\to\infty$}.
%	\]
%		   The   correlation functions have   two fastest    peaks,  called sound mode 
%		    that travel with opposite speeds $\pm \sqrt{\kappa_1}$ where $\kappa_1$ is the intensity of the spring and have a slower decay with time:
%	\[
%	S_{\alpha,\alpha'}(\pm\sqrt{\kappa_1}t,t)=\mathcal{O}\left(t^{-\frac{1}{3}}\right),\;\;\alpha,\alpha'=1,2,\;\;\quad S_{33}(\pm\sqrt{\kappa_1}t,t)=\mathcal{O}\left(t^{-\frac{2}{3}}\right), \mbox{ as $t\to\infty$}.
%	\]
%	These two fastest and low decaying peaks are generated by  the point $k\sim 0$ in the dispersion relations $\sqrt{ 2\kappa_1\left(1-\cos(2\pi k )\right)}$.
%	When we consider the harmonic oscillator with short range interactions the dispersion relation becomes 
%		The  corresponding fastest peaks  travel with velocity $v_0= \sqrt{\sum_{s=1}^ms^2\kappa_s}$  have a decay
%
	 
	 Furthermore we can tune the spring intensities $\kappa_s$, $s=1,\dots,m$ in \eqref{dispersion0}  so that 
	 we can find an $(m-1)$-parameter family of potentials 
	such that for  $j\sim\pm  v^\ast t$, with $0\leq v^\ast<v_0$, one has 
	\[
	S_{\alpha\alpha'}(j,t)= \mathcal{O}\left(\frac{1}{t^{\frac{1}{4}}}\right),\;\;\alpha,\alpha'=1,2\,,\quad S_{33}(j,t)= \mathcal{O}\left(\frac{1}{t^{\frac{1}{2}}}\right),\quad \mbox{ as $t\to\infty$}\,.
	\]
	In this case the local behaviour of the correlation functions is described by the Pearcey integral (see Theorem~\ref{theoremP} and Figures~\ref{fig_ex1}, \ref{fig_ex2} below). 
%	When $k^\ast=\frac{1}{2}$, $v^\ast=0$ and  one has the compact leading order formula
%	\[
%	S_{33}(j,t)= 	\frac{1}{4 \beta^2\pi^2(\lambda^{\ast})^2t^{\frac{1}{2}}}\left| {\mathcal P}_\pm\left(\frac{j}{\lambda^{\ast}t^{\frac{1}{4}}}\right) \right|^2+\cO(t^{-3/4})\,
%	\]
%	where $\lambda^{\ast}=\dfrac{1}{2\pi}(|f^{(iv)}(\frac{1}{2})|/4!)^{\frac{1}{4}}>0$,  ${\mathcal P}_\pm(a)=\int_{-\infty}^{\infty}e^{i(\pm y^4+ay)}dy,$ and ${\mathcal P}_\pm$ is chosen according to the sign of $f^{(iv)}(\frac{1}{2})$. 	%	This decay coincides with that of the  the harmonic oscillator with nearest neighbourhood interactions.
	For example a potential with such behaviour  is given by a spring  interaction  of the form $\kappa_s=\dfrac{1}{s^2}$  for $s=1,\dots, m$ and $m$ even (see Example~\ref{example1} below).\\
%\color{red}
In Section~\ref{sect3.2} we  study numerically  small nonlinear perturbations  of the harmonic oscillator with short range interactions and 
our results suggest that the behaviour of the fastest peak has a transition from  the Airy asymptotic  \eqref{Airy0} to the Tracy-Widom asymptotic \eqref{TW}, depending on the strength of the nonlinearity.
Namely  the  asymptotic  behaviour in  \eqref{TW} that has been conjectured for nearest neighbour interactions seems to persist also for sufficiently strong   nonlinear perturbations of the  harmonic oscillator with short range interactions.
Remarkably, our numerical simulations indicate that the non generic decay in time of other peaks in the correlation functions persists  under small nonlinear perturbations with the same power law $t^{-1/4}$ as in the linear case, see e.g. Figures~\ref{fig_ex1_non} and \ref{fig_ex3_non}. 
%
%Finally we compute the analogue of the formula \ref{Spohn2} in  our setting  and we show that such quantity is Gaussian random variable with variance that scales at leading order like $\sqrt{t}$ as $t\to\infty$ ( see Thereom~\ref{Theorem_Gaussian} below).
\vskip 0.2cm

So as not to overlook a large body of related work, we observe that the quantities we consider here are somewhat different than those considered in the study of thermal transport, though there is of course overlap.  (We refer to the Lecture Notes \cite{Lepri_Book} for an overview of this research area and also the seminal paper \cite{RLL}.)  As mentioned above, we study the dynamical evolution of space-time correlation functions and the statistical description of random height functions, where the only randomness comes from the initial data.  By comparison, in the consideration of heat conduction and transport in low dimensions, anharmonic chains are often connected at their ends to heat reservoirs of different temperatures, and randomness is present primarily in the dynamical laws, not only in fluctuations of initial data. 
% The scaling laws that emerge in our analysis are also distinct from the scaling laws arising in the consideration of heat conduction.

\vskip 0.2cm

This manuscript is organized as follows.
In Section 2 we study the harmonic oscillator with short range interactions and we introduce  the necessary notation and the change of coordinates $\bq\to \br$ that enables us to study correlation functions.
We then  study the time decay  of the correlation functions  via steepest descent analysis and we show that the two fastest peaks travelling in opposite directions
 originate from the points $k=0$ and $k=1$ in the spectrum. Such peaks have a decay described by the Airy scaling. We then show the 
   existence of potentials such that the correlation functions have a slower   time decaying with respect to  "Airy peaks".
   In Section 3 we show that the harmonic oscillator  with short range interactions has a complete set of {\it local } integrals of motion in involution and the correlation functions
   of such integrals have the same structure as the energy-energy correlation function.
 Finally, we show that the evolution equations for the generalized position, momentum can be written in the form of  conservation laws which have a potential function.
 For the case of the harmonic oscillator with nearest neighbour interaction, we show that this function is  a Gaussian random variable
and determine the leading order behaviour of its variance as $t\to\infty$. 
This may be viewed as  the analogue of formula \eqref{Spohn2} for the linear case.
Technicalities and a description of our numerics are deferred to the Appendix.
	\section{The harmonic oscillator with short range interactions}

As it was explained in the introduction we rewrite the Hamiltonian for the harmonic oscillator with short range interactions
\[
H(\bp,\bq)=\sum_{j=0}^{N-1}\frac{p_j^2}{2}+\sum_{s=1}^m\frac{\kappa_s}{2}\sum_{j=0}^{N-1}(q_j-q_{j+s})^2 =
\sum_{j=0}^{N-1}\left(\frac{p_j^2}{2} +\frac{1}{2} \Big(\sum_{s=1}^m \tau_s (q_{j+s}-q_j)
\Big)^2\right)
\]
so that we may define a Hamiltonian density 
\[
e_j=\frac{p_j^2}{2}+\frac{1}{2} \Big(\sum_{s=1}^m \tau_s (q_{j+s}-q_j)
\Big)^2,
\]
which is local in the variables $(\bp,\bq)$ for fixed $m$.  Namely, if we let $N\to\infty$, the quantity $e_j$ involves a finite number of physical variables  $(\bp,\bq)$. 
Recall that the coefficients $\tau_s$ are the entries of the circulant localized square root $T$ of the matrix  $A$ by which we mean a solution of the equation \eqref{eq:T0} of the form \eqref{form:T}. The matrix $T$ will also play a role in constructing a complete set of integrals that have a local density in the sense that we just described for the energy. 
%Since the system is integrable, we want to construct  a set of conserved quantities 
% whose Hamiltonian densities are local in the sense just explained.
%This construction relies on the fact that the matrix $A$ defined in \eqref{A} has a "square" root.
	
	In order to state our  result  we have to introduce some notation.   First of all,  a matrix $A$ of the form \eqref{A} with $\ba\in\R^N$ is called a circulant matrix generated by the vector $\ba$.
	\begin{definition}[$m$-physical vector and half-$m$-physical vector]
	\label{def:ph}
			Fix $m \in \N$. 
			For any odd $ N > 2m$, a   vector $\tilde{\bx}\in \R^{N}$  is said to be  {\em $m$-physical} generated by $\bx=(x_0,x_1,\dots, x_m)\in \R^{m +1}$ if $x_0=-2\sum_{s=1}^mx_s$ and 				
			\begin{align*}
				\tilde{x}_0 = & x_0 \, ,\\
				\tilde{x}_1= &\tilde{x}_{N-1} = x_1<0 ,\;\;\tilde{x}_m= \tilde{x}_{N-m} = x_m<0,\\
			\tilde{x}_k= &\tilde{x}_{N-k} = x_k \leq 0, \mbox{ for  $1< k <  m$},\\ 
			 \tilde{x}_{k} =& 0, \mbox{  otherwise,} 
						\end{align*}
	while the  vector  $\tilde{\bx}\in \R^{N}$  is called  {\em half-$m$-physical} generated by $\by\in  \R^{m +1}$ if  $y_0=-\sum_{s=1}^my_s$  and	\begin{align*}
		\tilde{x}_k=  & y_k,  \mbox{ for  $0 \leq k \leq m$   } \\
		\tilde{x}_k=&0,  \mbox{ for  $m< k \leq N-1$.   }
		\end{align*}
		\end{definition}
%	
%			\begin{definition}[
%		\label{def:hph}
%		Fix $m \in \N$. 
%		For any $ N \geq m $, a   vector $\bx\in \R^{N}$  is said to be  {\em half-$m$-\th{physical}} if 
%		there exits   a  non zero vector  $\by=(y_0, y_1, \ldots, y_{m}) \in \R^{m +1}$ such that 
%		%$y_0 = -\sum_{j=1}^{m} y_j$ and 
%		
%		
%	\end{definition}
Following  the proof of a classic lemma by Fej\'er and Riesz, see e.g.~\cite[pg.~117 f]{book_Riesz_Nagy}, one can show that a circulant symmetric matrix $A$ of the form \eqref{eq:A}  generated by  a $m$-physical vector $\ba$ always has a circulant localized square root
%{\em Golub - Toepliz} \cite{Golubov}, the next proposition shows that a circulant symmetric matrix $A$ of the form \eqref{eq:A}  generated by  a $m$-\th{physical} vector $\ba$ always has a circulant localized square root 
 $T$  that is generated by a half-$m$ physical vector $\btau$.
\begin{proposition}
	\label{prop:matrix_root}
	Fix $m \in \N$. Let the circulant matrix $A$ be generated by an $m$-physical vector $\ba$, then there exist a circulant matrix $T$ generated by an half-$m$-physical vector $\btau$ such that:
	\begin{equation}
	\label{eq:T}
		A = T^\intercal T\,.
	\end{equation}
	Moreover, we can choose $\btau$ such that $\sum_{s=1}^m s \tau_s > 0$. Then one has $\sum_{s=1}^m s \tau_s=\sqrt{\sum_{s=1}^ms^2\kappa_s}$.
\end{proposition}
The proof of the proposition is contained in Appendix~\ref{Appendix_A}.

For  example,  if  we consider $m=1$, and $a_0=2\kappa_1$ and $a_1=a_{N-1}=-\kappa_1$.  The matrix $T$ is generated by the vector $\btau=(\tau_0,\tau_1)$ with $\tau_0=-\sqrt{\kappa_1}$ and $\tau_1= \sqrt{\kappa_1}$.
When $m=2$ and  $a_0=2\kappa_1+2\kappa_2$,  $a_1=a_{N-1}=-\kappa_1$,  $a_2=a_{N-2}=-\kappa_2$. The matrix $T$ is generated by the vector $\btau=(\tau_0,\tau_1,\tau_2)$
with 
\begin{align*}
&\tau_0=- \frac{\sqrt{\kappa_1}}{2}-\frac{1}{2}\sqrt{\kappa_1+4\kappa_2}, \;\;\tau_1=\sqrt{\kappa_1},\\
&\tau_2=-\frac{\sqrt{\kappa_1}}{2}+\frac{1}{2}\sqrt{\kappa_1+4\kappa_2},\end{align*}
so that the quantities $r_j$ are defined as 
\[
r_j=\tau_1(q_{j+1}-q_j)+\tau_2 (q_{j+2}-q_j)\,,\quad j\in\Z / N\Z \,.
\]
Next we integrate the equation of motions.
The Hamiltonian $H(\bp,\bq)$	represents clearly an integrable system  that can be integrated passing through Fourier transform.
Let $\cF$ be the discrete Fourier transform with entries  $\cF_{j,k}:= 
\frac{1}{\sqrt{N}} e^{- 2\im \pi j k/N}$ with $j,k=0,\dots, N-1$. It is immediate to verify that
	\begin{equation}
	\label{DHT_prop}
	\cF^{-1} = \widebar{\cF} \qquad \cF^\intercal = \cF.
	\end{equation}
	Thanks to the above  properties,  the transformation defined by
	\begin{equation}
	\label{eq:F_variables}
	(\wh p, \wh q) = (\widebar \cF p, \cF q)
	\end{equation} 
	is canonical.  Furthermore  $\widebar{\wh p}_j = \wh p_{N-j}$  and $\widebar{\wh q}_j = \wh q_{N-j}$, for $j=1,\dots,N-1$, while  $\wh p_0$ and  $\wh q_0$ are real variables. 
The matrices  $T$ and $A$ are  circulant matrices and so they are reduced to diagonal form by   $\cF$:
\[
\cF A \cF^{-1}= \cF T^\intercal T\cF^{-1}=\overline{(\cF T\cF^{-1})}^\intercal (\cF T\cF^{-1})\;.
\]
Let  $\omega_j$ denote the eigenvalues of the matrix $T$ ordered so that $\cF T\cF^{-1}=$ diag$(\omega_j)$. Then  $|\omega_j|^2$  are the (non negative) eigenvalues of the matrix $A$ and 
\begin{equation}
\label{omega}
|\omega_j|^2=\sqrt{N}( \overline{ \cF} \tilde{\ba})_j,\quad \omega_j=\sqrt{N}( \overline{\cF}\tilde{ \boldsymbol{\tau}})_j, \quad j=0,\dots, N-1,
\end{equation}
%\todo{T. We need to say that $\omega_0=0$}
where  $\tilde{\ba}$ is  the $m$-physical vector generated by  $\ba$ and  $ \tilde{\boldsymbol{\tau}}$ is the half $m$-physical vector generated by 
  $ \boldsymbol{\tau}$ according to Definition~\ref{def:ph}. It follows that 
  \begin{equation}
  \label{sym_omega}
\omega_0=0,\quad   \omega_j=\overline{\omega}_{N-j}\,,\quad j=1,\dots,N-1,
  \end{equation}
  which implies $|\omega_{j}|^2=|\omega_{N-j}|^2$, $j=1,\dots,N-1$.
 The   Hamiltonian $H$, can be written  as the sum of $N-1$ oscillators
% \todo{T. Numerology was wrong}
  	\begin{equation}
	H(\wh \bp, \wh \bq) = \frac{1}{2}\left( \sum_{j=1}^{N-1} |\wh p_j |^2 + |\omega_j|^2 |\wh q_j|^2  \right)\, =  \sum_{j=1}^{\frac{N-1}{2}} |\wh p_j |^2 + |\omega_j|^2 |\wh q_j|^2 \, .
	\end{equation} 	
% \todo{Th. Check change, has consequences for \eqref{eq:measureDFT}}
 	There are no terms involving $\wh p_0, \wh q_0$ since the conditions defining $\cM$ \eqref{eq:phase_space} imply that $\wh p_0=0$ and $\wh q_0=0$. 
	 The Hamilton equations are 
\begin{equation}
\label{H_eq}
		\begin{cases}
		\dfrac{d}{dt}\wh{q}_{j}
		% = \dfrac{\partial H}{\partial \wh p_j}
		=\overline{\wh p_{j}}\vspace{.2cm}\\
		\dfrac{d}{dt}\overline{\wh{p}}_{j} 
		%=-\dfrac{\partial H}{\partial \overline{\wh q_{j}}}
		=- |\omega_j|^2\wh q_{j}\, .\\
		\end{cases}
	\end{equation}
	Thus the general solution reads:
	\begin{equation}
	\label{evolution_pq}
	\begin{split}
			&\wh q_{j}(t) = \wh q_{j}(0 )\cos(|\omega_{j}|t) + \frac{\overline{\wh p_{j}(0)}}{|\omega_j|}\sin(|\omega_j| t)\, , \\
			&\overline{\wh p_{j}(t)} = \overline{\wh p_{j}(0)}\cos(|\omega_{j}|t) - |\omega_j|\wh q_{j}(0)\sin(|\omega_j| t) \,,\quad j=1,\dots, N-1,
	\end{split} 
	\end{equation}
	 and $\wh q_{0}(t)=0$ and $\wh p_{0}(t)=0$. 
	Inverting the Fourier transform, we recover the variables $\bq=\cF^{-1}\wh \bq$, $\bp=\cF\wh \bp$ and $\br=\cF^{-1}\wh \br$ where
		\begin{equation}
	\label{eq:rjhat}
		\wh r_j =\omega_j\wh q_j, \, j = 0, \ldots, N-1\,.
	\end{equation}
	
%
%
%The Hartly transform is defined as
% $\cH = \Re \cF - \Im \cF$ and in coordinates
% \[
%  \cH_{j,k} := \frac{1}{\sqrt{N}}\left(\cos\left(2\pi \frac{jk}{N}\right) + \sin\left(2\pi \frac{jk}{N}\right) \right) , \qquad  j,k = 0,\ldots, N-1\,.
% \]
%One   easily verifies that 
%\begin{equation}
%\label{dht.2}
%\cH^2 = \uno , \qquad \cH^\intercal  = \cH ,
%\end{equation}
%where $\cH^\intercal$ is the transpose matrix.
%
% It is well known   \cite{Gray} that Hartley transform is a canonical transformation that diagonalizes   symmetric circulant matrices  of the form \eqref{eq:A}. Indeed let 
%\begin{equation}
%\label{dht}
%	\wt \bp  := \cH \bp,  \, \quad \wt \bq  := \cH \bq,
%\end{equation}
%then on can see using \eqref{dht.2} that $\wt \bp$ and $\wh \bq$ are canonical variables
%\[
%\{\wt q_j,\wt p_k\}=\delta_{k,j},\quad \{\wt q_j,\wt q_k\}=\{\wt p_j,\wt p_k\}=0,\quad j,k=0,\dots,N-1.
%\]
%

	\subsection*{Correlation Decay}
		We now study  the decay of correlation functions for Hamiltonian systems  of the form~\eqref{eq:A}. 
	We recall the definition \eqref{eq:measure} of the Gibbs measure at temperature $\beta^{-1}$ on the reduced phase space $\cM$, namely:
	%endow the phase space with the Gibbs measure at temperature $\beta^{-1}$, namely:
	\begin{equation*}
	%\label{eq:measure}
	\di \mu = Z_N(\beta)^{-1} \delta\left(\sum_{k=0}^{N-1} p_k\right)\delta\left(\sum_{k=0}^{N-1} q_k \right)e^{-\beta H(\bp,\bq)}\di \bp \di \bq	\end{equation*}
	where $Z_N(\beta)$ is the  norming constant of the probability measure.	
	For a function $f=f(\bp,\bq)$ we define its average as
	\begin{equation*}
	 \la f \ra := \int_{\R^{2N}} f(\bp,\bq)  \, \, \di \mu .
	\end{equation*}	
We first  compute  all correlation functions \eqref{eq:correlation_functions}, then we will evaluate the limit $N\to\infty$.
		We first observe that \eqref{eq:measure} in the variables $(\wh \bp, \wh \bq) := (\widebar \cF \bp, \cF \bq)$ becomes
	
	\begin{equation}
	\label{eq:measureDFT}
	\di \mu = Z_N(\beta)^{-1} \prod_{j=1}^{\frac{N-1}{2}}e^{-\beta(|\wh p_j|^2 +|\omega_j|^2|\wh q_j|^2 )}\di \wh p_j \di \wh q_j 
	\end{equation}
	%\todo{is this the correct measure?  dimensions do not match}
	where   $\di \wh p_j \di \wh q_j =\di \Re\wh p_j \di \Im\wh p_j \di \Re\wh q_j \di \Im\wh q_j$ and we recall that  $\wh p_j = \overline{\wh p}_{N-j}$, $\wh q_j = \overline{\wh q}_{N-j}$, $\wh r_j =\omega_j\wh q_j,$  for $j=1,\dots, N-1$. 

From the evolution of  $\wh p_j$ and $\wh q_j$   in \eqref{evolution_pq}  and \eqref{eq:rjhat},  we arrive at the relations
	\begin{align}
	& \la \wh p_j(t)\overline{\wh p_k(0)} \ra = \la \overline{\wh p_k(0)}\left(\wh p_j (0)\cos(|\omega_j|t) - |\omega_j|\overline{\wh q_{j}(0)}\sin\left(|\omega_j| t\right)\right)\ra = \delta_{j,k} \frac{1}{\beta}\cos(|\omega_j|t), \\ 
%	&\la \wh p_j(t)\overline{\wh q_k(0)} \ra = \la \overline{\wh q_k(0)}\left(\wh p_j (0)\cos(|\omega_j|t) + |\omega_j|\overline{\wh q_{j}(0)}\sin\left(\omega_j t\right)\right)\ra = \delta_{j,k} \frac{1}{|\omega_j|\beta}\sin\left(|\omega_j|t\right),\\
%	&\la \wh q_j(t)\overline{\wh q_k(0)}\ra = \la\overline{\wh q_k(0)}\left(\wh q_j (0)\cos(|\omega_j|t) - \frac{\overline{\wh p_{j}(0)}}{|\omega_j|}\sin\left(|\omega_j| t\right) \right) \ra= \delta_{j,k} \frac{1}{|\omega_j|^2\beta}\cos(|\omega_j|t),\\
		& \la \wh p_j(t)\wh r_k(0) \ra = \la \omega_k \wh q_k(0)\left(\wh p_j (0)\cos(|\omega_j|t)- |\omega_j |\overline{ \hat{q}_{j}(0)}\sin\left(|\omega_j| t\right)\right)\ra = -\delta_{j,k} \frac{\omega_j}{|\omega_j|\beta}\sin\left(|\omega_j| t\right),\\
		&\la \wh r_j(t)\wh p_k(0)\ra = \la \omega_j\wh p_k(0)\left(\wh q_j (0)\cos(|\omega_j|t)  +   \frac{\overline{\wh p_{j}(0)}}{|\omega_j|}\sin\left(|\omega_j| t\right) \right) \ra= \delta_{j,k}\frac{\omega_j}{|\omega_j|\beta}\sin(|\omega_j|t)\\
	&\la \wh r_j(t)\overline{\wh r_k(0)}\ra = \la \overline{\omega}_k\omega_j\overline{\wh q_k(0)}\left(\wh q_j (0)\cos(|\omega_j|t)  +   \frac{\overline{\wh p_{j}(0)}}{|\omega_j|}\sin\left(|\omega_j| t\right) \right) \ra= \delta_{j,k} \frac{1}{\beta}\cos(|\omega_j|t).
	\end{align} 
%\todo{the evolution in the correlations is different from  2.8 \th{what do you mean?} }
				Now we are ready to compute explicitly the correlation functions in the physical variables. We show the computation for the case $S^{N}_{11}(j,t) $, and we leave to the reader the details for the other cases:  
	
%\green{	
	\begin{equation}
	\label{eq:SppExplicit1}
	\begin{split}
	S^{N}_{11}(j,t) & = \la r_{j}(t)r_0(0) \ra = \frac{1}{N}\la\sum_{k,l=1}^{N-1}\wh r_k(t)\wh r_l(0)e^{2\pi \imath \frac{jk}{N}}\ra \\ 
	%& = \frac{1}{N}\sum_{k,l=0}^{N-1} \la \wh p_k(0)\wh p_l(t) \ra e^{-2\pi \imath \frac{jk}{2N+1}}\\
	&= \frac{1}{N\beta} \sum_{l=1}^{N-1}\cos\left(|\omega_l|t\right)\cos\left(2\pi \frac{lj}{N}\right)=S^{N}_{22}(j,t)\,.
	\end{split}
	\end{equation}	
%}
	In the same way we have that:
	\begin{align}
	\label{eq:CrrExplicit}
	%&C^{(N)}_{rr}(j,t) = \frac{1}{(N)\beta}\sum_{l=0}^{N-1} \cos(|\omega_l|t)\cos\left(2\pi \frac{lj}{N} \right) , \\
	&S^{N}_{12}(j,t)=  \frac{1}{N\beta}\sum_{l=1}^{N-1}\sin(|\omega_l|t)\cos\left(2\pi  \frac{lj}{N} + \arg (\omega_l) \right)\\
	\label{eq:CprExplicit}
	&S^{N}_{21}(j,t)= - \frac{1}{N\beta}\sum_{l=1}^{N-1}\sin(|\omega_l|t)\cos\left(2\pi  \frac{lj}{N}- \arg (\omega_l) \right)\\
	\label{eq:CzeroExplicit}
	%&C^{N}_{12}(j,t)=\frac{1}{N\beta}\sum_{l=1}^{N-1}\sin(|\omega_l|t)\cos\left(2\pi  \frac{lj}{N}{\color{red} +} \arg (\omega_l) \right)\\
		&S^{N}_{31}(j,t)= S^{N}_{32}(j,t) = S^{N}_{13}(j,t)  = S^{N}_{23}(j,t) = 0 \\
	\label{eq:CeeExplicit}	
	&S^{N}_{33}(j,t)  	
	%\frac{1}{\beta^2N^2}\sum_{n\ne l =1}^{N-1} \left[\cos\left(|\omega_n|t\right)\cos\left(|\omega_l|t\right) +\sin\left(|\omega_n|t\right)\sin\left(|\omega_l|t\right)\frac{\omega_n\omega_l}{|\omega_n\omega_l|}\right] \cos\left(2\pi \frac{(l+n)j}{N}\right) + \cO\left(\frac{1}{N}\right)\\
%&
=\frac{1}{2}((S^N_{11})^2 + (S^N_{22})^2 + (S^N_{12})^2+(S^N_{21})^2)+ \frac{3(N-1)}{2N^2\beta^2}\,.
	%&S^{(N)}_{33}(j,t) = \frac{1}{\beta^2N^2}\sum_{n\ne l =1}^{N-1}\cos\left((|\omega_n| - |\omega_l|)t\right)\cos\left(2\pi \frac{lj}{N}\right)\cos\left(2\pi \frac{nj}{N}\right) + \cO\left(\frac{1}{N}\right)
	\end{align} 	%\todo{ Calculare the correlation function of the energy with $p$ and $r$}
	%We  observe that $S_{pp} = S_{rr} $ and $ \partial_{t} S_{pp}(j,t) = - S_{pr}(j,t)$. 	
	%\todo{please check $S_{33}$, maybe using numerics}
	
	The dispersion relation  given by \eqref{omega} takes the form
	\begin{equation}
	\begin{split}
		\label{dispersion}
		&\omega_\ell=-\sum_{s=1}^m\tau_s\left(1-\cos\left(2\pi \frac{s\ell}{N}\right)\right)+\im \sum_{s=1}^m\tau_s\sin\left(2\pi \frac{s\ell}{N}\right)\\
		&|\omega_\ell|^2= \sum_{s=0}^{N-1}a_se^{-2\pi\im\frac{s\ell}{N}}=  2\sum_{s=1}^m\kappa_s\left(1-\cos\left(2\pi \frac{s\ell}{N}\right)\right)\,,
		% |\left(1-\cos(2\pi x j)\right),
	\end{split}
	\end{equation} 
	where  we substitute for the $a_s$ their values as in \eqref{v_a}. 
	We are interested in obtaining the continuum limit of the above correlation functions.
	We first define  $\omega(k)$ 
	to provide continuum limits of $\omega_\ell $ and  $|\omega_\ell |^2$, namely
	%and  $|\omega(k)|$ and the continuum limit of $\omega_\ell $ and  $|\omega_\ell |^2$ respectively, namely
	\begin{equation}
	\begin{split}
		\label{eq:general_dispersion}
		\omega(k)&:=-\sum_{s=1}^m\tau_s\left(1-\cos\left(2\pi sk\right)\right)+\im \sum_{s=1}^m\tau_s\sin\left(2\pi sk\right)\\
		|\omega(k)|^2& = 2\sum_{s=1}^m\kappa_s\left(1-\cos(2\pi k s)\right)\,,
		% = \varepsilon^2 + \sum_{j=1}^m |a_j|\sin^2(\pi x j),
	\end{split}
	\end{equation} 
where the variable $\ell/N$ has been approximated with $k\in[0,1]$. One may use equation \eqref{eq:polynomial_eq} to check the consistency of the two equations of \eqref{eq:general_dispersion}. To this end observe that $\omega(k)=Q(e^{-2\pi ik})$,  $\widebar{\omega(k)}=Q(e^{2\pi ik})$, and $|\omega(k)|^2=\ell(e^{2\pi ik})$.
\begin{lemma}\label{fandtheta}
Let $\omega(k)$ be defined as in \eqref{eq:general_dispersion}, set $f(k):=|\omega(k)|$, and denote $\theta(k):=$ arg$(\omega(k))$ for $0\leq k \leq 1$, where the ambiguity in the definition of $\theta$ is settled by requiring $\theta$ to be continuous with $\theta(0) \in (-\pi, \pi]$. Then, for all $k\in[0,1]$ we have
\begin{align}
\label{omegasym}
&\omega(1-k)=\widebar{\omega(k)},\\
\label{omegasym2}
&f(1-k)=f(k),\\
\label{thetasym}
& \theta(1-k)\equiv - \theta(k)\quad (\mbox{mod} \;\;2\pi).
\end{align}
Furthemore, the functions $f$ and $\theta-\frac{\pi}{2}$ are $C^\infty$ on $[0, 1]$ and they both possess odd $C^\infty$-extensions at $k=0$ which implies in particular $\;\theta(0)=\frac{\pi}{2}\,$.
\end{lemma}
\begin{proof}
The symmetries  follow directly from the definition of $\omega$ in~\eqref{eq:general_dispersion}. From~\eqref{eq:general_dispersion} we also learn that $|\omega(k)|^2 \geq 2 \kappa_1 (1-\cos(2\pi k)) > 0$ for $k\in (0, 1)$. Thus the smoothness of $f$ and $\theta$ only needs to be investigated for $k\in\{0,1\}$. By symmetry we only need to study the case $k=0$.
The smoothness of  the function $\theta$ may be obtained from the expansion near $k=0$
\[
\cot (\theta(k)) = -k\pi \frac{\sum_{s=1}^m s^2\tau_s}{\sum_{s=1}^m s \tau_s} + \cO(k^3)
\]
together with $\sum_{s=1}^m s \tau_s > 0$ (see Proposition~\ref{prop:matrix_root}). Since $\cot (\theta(0))=0$ and $\Im \omega(k) > 0$ for small positive values of $k$ we conclude that $\theta(0)=\frac{\pi}{2}$ from the requirement  $\theta(0) \in (-\pi, \pi]$. This also implies the existence of a smooth odd extension of $\theta-\frac{\pi}{2}$ at $k=0$ because $\cot (\theta(k))$ has such an extension. For the function $f$ the claims follow from the representation
\[
f(k)=2\pi k \left(\sum_{s=1}^m s^2\kappa_s\mbox{ sinc}^2(\pi s k)\right)^{1/2}
\]
near $k=0$ where sinc$(x)=\frac{\sin(x)}{x}$ denotes the smooth and even sinus cardinalis function. 
\end{proof}

\begin{lemma}\label{lemma_limit}
In the limit $N\to\infty$ the correlation functions have the following expansion
\begin{align*}
& S^{N}_{\alpha\alpha'}(j,t)+\frac{\delta_{\alpha\alpha'}}{N\beta}= S_{\alpha\alpha'}(j,t)+ \mathcal O \left(N^{-\infty}\right),\quad \alpha, \alpha'=1,2,\\
& S^{N}_{33}(j,t)= S_{33}(j,t)+ \mathcal O \left(N^{-1}\right),
\end{align*}
where $\delta_{\alpha\alpha'}$ denotes the Kronecker delta, 
	\begin{align}
	\label{eq:expCpp}
		S_{11}(j,t) &=S_{22}(j,t) =\frac{1}{\beta}\int_{0}^1\cos\left(|\omega(k)|t\right)\cos\left(2\pi kj\right) \di k\\
				\label{eq:expCpr}
		S_{12}(j,t) &=\frac{1}{\beta}\int_{0}^1\sin\left(|\omega(k)|t\right)\cos\left(2\pi kj +\theta(k)\right) \di k,\\
				\label{eq:expCrp}
		S_{21}(j,t) &=-\frac{1}{\beta}\int_{0}^1\sin\left(|\omega(k)|t\right)\cos\left(2\pi kj -\theta(k)\right) \di k,\\
		\label{S33}
		S_{33}(j,t)& =\frac{1}{2}(S_{11}^2 + S_{22}^2 + S_{12}^2+S_{21}^2),
%		&=\th{\frac{1}{\beta^2}\int_0^1\int_0^1 \left[\cos\left(|\omega(k)|t\right)\cos\left(|\omega(s)|t\right) +\sin\left(|\omega(k)|t\right)\sin\left(|\omega(s)|t\right) e^{i(\theta(k)+\theta(s))}\right] \cos\left(2\pi(k+s)j\right) \di k \di s}
		%\color{black} S_{33}(j,t)& =\frac{1}{\beta^2}\int_0^1\int_0^1\cos\left((|\omega(k)| -|\omega(s)|)t\right)\cos\left(2\pi sj\right)\cos\left(2\pi kj\right) \di k \di s
	\end{align}
	and  $\theta(k)=\arg\omega(k)$  with $\omega(k)$ as in \eqref{eq:general_dispersion}.
	\end{lemma}
\begin{proof}
For any periodic $C^\infty$-function $g$ on the real line with period $1$, $g(k)=\sum_{n\in \Z}\hat{g}_n e^{2\pi ikn}$, one has
\[
\frac{1}{N} \sum_{\ell=0}^{N-1} g\left(\frac{\ell}{N}\right) = \sum_{m\in\Z} \hat{g}_{mN} = \int_0^1 g(k) \di k + \mathcal O \left(N^{-\infty}\right)\,.
\]
It follows from Lemma~\ref{fandtheta} that the integrands in \eqref{eq:expCpp}-\eqref{eq:expCrp} can be extended to $1$-periodic smooth functions because we have for small positive values of $k$ that
\begin{eqnarray*}
\cos\left(f(-k)t\right)\cos\left(-2\pi kj\right) &=& \cos\left(f(k)t\right)\cos\left(-2\pi kj\right) \;\;=\;\; \cos\left(f(1-k)t\right)\cos\left(2\pi(1- k)j\right)\,,\\
\sin\left(f(-k)t\right)\cos\left(-2\pi kj \pm\theta(-k)\right) &=& -\sin\left(f(k)t\right)\cos\left(-2\pi kj \pm(\pi-\theta(k))\right)\\ 
&=& \sin\left(f(1-k)t\right)\cos\left(2\pi(1-k)j \pm\theta(1-k)\right)\,.
\end{eqnarray*}
Observing in addition that the summands corresponding to $\ell=0$ are missing in \eqref{eq:SppExplicit1}-\eqref{eq:CprExplicit} the first claim is proved. Together with \eqref{eq:CeeExplicit} this also implies the second claim.
\end{proof}

	%The proof of the above lemma is  given in the Appendix~\ref{Appendix_B}. 

%	\begin{equation}
%	\label{eq:correlation_functions}
%	\begin{split}
%&C^{(N)}_{pp}(j,t) = \la p_{0}(t)p_j(0) \ra - \la p_{0}(t) \ra \la p_j(0) \ra, \\
%&C^{(N)}_{rr}(j,t) = \la r_{0}(t)r_j(0) \ra - \la r_{0}(t) \ra \la r_j(0) \ra ,\\ 
%&C^{(N)}_{pr}(j,t) = \la p_{0}(t)r_j(0) \ra - \la p_{0}(t) \ra \la r_j(0) \ra	.
%	\end{split}
%	\end{equation}

Next we analyse the leading order behaviour (as $t \to \infty$) of the limiting correlation functions $S_{\alpha\alpha'}(j,t)$ using the method of steepest descent. In order to explain the phenomena that may occur we start by discussing $S_{11}$. Denote
\begin{equation}\label{def:fxiphi}
	%f(k):= |\omega(k)|\,,  \qquad 	
	\xi:=\dfrac{j}{t}  \quad \mbox{and} \quad \phi_\pm(k, \xi) :=  f(k) \pm2 \pi \xi k\,.	
\end{equation} 
With these definitions and using the symmetry \eqref{omegasym2}
%, namely  $f(1-k)=f(k)$ 
we may write
\begin{equation}\label{eq:heuristic.1}
	S_{11}(j,t)= \frac{1}{2\beta}\Re \int_{0}^{1}\left( e^{it(f(k)+2\pi \xi k)} +e^{it(f(k)-2\pi \xi k)}\right) d k
	= \frac{1}{\beta}\Re \int_{0}^{1} e^{it\phi_- (k, \xi)}  d k\,.
\end{equation} 
The leading order behaviour ($t \to \infty$) of such an integral is determined by the stationary phase points $k_0 \in [0, 1]$, i.e.~by the solutions of the equation
$\frac{\partial}{\partial k} \phi_-(k_0, \xi) = 0$ which depend on the value of $\xi$.

Such stationary phase points do not need to exist. In fact, as we see in Lemma~\ref{lemma:omega} b) below, the range of $f'$ is given by some interval
$[-2\pi v_0, 2 \pi v_0]$ so that there are no stationary phase points for $| \xi | > v_0$. 
As in the proof of Lemma~\ref{lemma_limit} one can argue that the integrand $\Re e^{it\phi_- (k, j/t)}$ can be extended to a periodic smooth function of $k$ on the real line with period $1$. It then follows from integration by parts that $S_{11}(j,t)$ decays rapidly in time.
%iodicity $e^{it\phi_- (k+1,\, j/t)}=e^{it\phi_- (k, \, j/t)}$ this yields via integration by parts rapid time-decay of $S_{11}(j,t)$ 
%Using the perfor $|j/t| >v_0$. 
More precisely, for every fixed $\delta > 0$ we have
\begin{equation}\label{eq:rapid_decay}
S_{11}(j,t)= \cO  \left(t^{-\infty}\right)\quad\mbox{as $t \to \infty$, uniformly for $|j|\geq(v_0 + \delta)t$.} 
\end{equation}
This justifies the name of sound speed  for  the quantity $v_0$.

In the case $| \xi | \leq v_0$ there always exists at least one stationary phase point $k_0=k_0(\xi) \in [0, 1]$. 
Each stationary phase point may provide an additive contribution to the leading order behaviour of $\int_{0}^{1} e^{it\phi_- (k, \, j/t)}  d k$ for $j$ near $\xi t$. 
However, the order of the contribution depends on the multiplicity of the stationary phase point. 
For example, let $k_0$ be a stationary phase point of $\phi_- (\cdot, \xi)$, i.e.~$\frac{\partial}{\partial k} \phi_- (k_0, \xi) = 0$.
 Denote by $\ell$ the smallest integer bigger than $1$ for which  $\frac{\partial^\ell}{\partial k^\ell} \phi_-(k_0, \xi) \neq 0$. 
 Then $k_0$ contributes a term of order $t^{1/\ell}$ to the $t$-asymptotics of $\int_{0}^{1} e^{it\phi_- (k, \, j/t)}  d k$ for $j$ in a suitable neighbourhood of $\xi t$.

Before treating the general situation let us recall the case of nearest neighbour interactions. There we have
\[
f(k) = f_{1}(k) =\sqrt{2 \kappa_1 (1- \cos(2\pi k))} = 2 \sqrt{\kappa_1} \sin(\pi k)\,, \quad k \in [0, 1]\,.
\]
The range of $f_{1}'$ equals $[-2\pi v_0, 2 \pi v_0]$ with $v_0= \sqrt{\kappa_1}$. For every $|\xi| \leq v_0$ there exists exactly one stationary phase point $k_0(\xi)\in [0, 1]$ of $\phi_- (\cdot, \xi)$ that is determined by the relation $\cos(\pi k_0(\xi)) = \xi/v_0$.
A straight forward calculation gives
\[
\frac{\partial^2}{\partial k^2} \phi_-(k_0(\xi), \xi) = f_{1}''(k_0(\xi)) = -2 \pi^2 \sqrt{v_0^2-\xi^2} = 0 \; \Leftrightarrow \; \xi=\pm v_0\,.
\]
Moreover, we have $k_0(v_0)=0$ and $k_0(-v_0)=1$ and therefore $\frac{\partial^3}{\partial k^3} \phi_-(k_0(\pm v_0), \pm v_0) = \mp2 \pi^3 v_0 \neq 0$. This implies that in addition to \eqref{eq:rapid_decay} we have $S_{11}(j,t)= \cO (t^{-1/2})$, except for $j$ near $\pm v_0 t$ where $S_{11}(j,t)= \cO (t^{-1/3})$. In order to determine the behaviour near the least decaying peaks that travel at speeds $\pm v_0$ we expand $f_{1}$ near the stationary phase points. Let us first consider $\xi=v_0$ with $k_0=0$. Introducing $\lambda_0=\frac{1}{2\pi}|f_{1}'''(0)/2|^{1/3}=\frac{1}{2}v_0^{1/3}$ we obtain
\[
f_{1}(k)=2\pi v_0 k - \frac{1}{3} (2 \pi \lambda_0 k)^3 + \cO( k^5 )\,, \quad \mbox{as $k \to 0$.}
\]
Substituting $y = 2 \pi \lambda_0 t^{1/3} k$ leads for $k$ close to $0$ to the asymptotic expression 
\[
t \phi_- (k, \, j/t)= \frac{v_0 t-j}{\lambda_0 t^{1/3}} y - \frac{1}{3} y^3 + \cO( t^{-2/3} )\,, \quad \mbox{as $t \to \infty$.}
\]
Using the well-known representation  $\mbox{Ai}(w) = \frac{1}{\pi} \int_0^\infty \cos(y^3/3+wy) d y$, $w \in \R$, of the Airy function and performing a similar analysis around the stationary phase point $k_0=-1$ for $\xi=-v_0$ one obtains an asymptotic formula for the region not covered by \eqref{eq:rapid_decay}
\begin{equation}\label{eq:slow_decay}
S_{11}(j,t)= \frac{1}{2 \beta \lambda_0 t^{1/3}} \left[ \mbox{Ai}\left(\dfrac{j-v_0t}{ \lambda_0 t^{1/3}}\right)+ \mbox{Ai}\left(-\dfrac{j+v_0t}{ \lambda_0 t^{1/3}}\right) \right]+\cO\!\!\left(t^{-1/2}\right),\mbox{ $t \to \infty$, uniformly for $|j|<(v_0 + \delta)t$} 
\end{equation}
for $\delta >0$  (see e.g. \cite{Miller}). Observe that due to the decay of Ai$(w)$ for $w \to \pm\infty$, the Airy term is dominant  roughly in the regions described by $v_0 t - o (t) <|j|<v_0 t + o ((\ln t)^{2/3})$.
%Two remarks are in order. Firstly, the asymptotics presented in \eqref{eq:slow_decay} is consistent with the asymptotics of \eqref{eq:rapid_decay} albeit much weaker in the region where \eqref{eq:rapid_decay} applies. Secondly, due to the decay of Ai$(w)$ for $w \to \pm\infty$, the Airy term   is dominant  roughly in the regions   described by $v_0 t - o (t) <|j|<v_0 t + o ((\ln t)^{2/3})$.

From the arguments just presented it is not difficult to see that the derivation of \eqref{eq:slow_decay} only uses the following properties of  $f=f_{1}$:
\begin{equation}\label{cond:1}
f''(k) < 0 \quad \mbox{for all} \quad 0 < k \leq\frac{1}{2}\,,
\end{equation}
together with
\begin{equation}\label{cond:2}
\quad f''(0) = 0\,,\quad f'''(0) < 0\,, \quad \mbox{and} \quad f(1-k)=f(k)\quad \mbox{for all} \quad 0 \leq k < \frac{1}{2}\,.
\end{equation}

Conditions \eqref{cond:1} and \eqref{cond:2} imply that statements \eqref{eq:rapid_decay} and \eqref{eq:slow_decay} hold with $v_0=\frac{f'(0)}{2 \pi} > 0$ and $\lambda_0=\frac{1}{2\pi}|f'''(0)/2|^{1/3}$.

It follows from equation \eqref{omegasym2} and from statement a) of~Lemma~\ref{lemma:omega} below that the conditions of \eqref{cond:2} are always satisfied in our model. Condition \eqref{cond:1}, however, might fail. Indeed, it is not hard to see that there exist open regions in the $\bkappa$-space $\R_+^m$ where there always exist stationary phase points $k_0 \in (0, 1) $ of higher multiplicity, i.e.~with $f''(k_0)=0$. In this situation the value of $v:= \frac{f'(k_0)}{2 \pi}$ lies in the open interval $(-v_0, v_0)$ (cf.~Lemma~\ref{lemma:omega} b). Then the decay rate of $S_{11}(j,t)$ for $j$ near $vt$ is at most of order $t^{-1/3}$. The decay is even slower (at least of order $t^{-1/4}$) if $f'''(k_0)=0$ holds in addition. We show in Theorem \ref{theoremP} that this may happen for $\bkappa$ in some submanifold of $\R_+^m$ of codimension $1$ (see also Examples \ref{example1} and \ref{example2}). Nevertheless, if $\kappa_2$, $\ldots$, $\kappa_m$ are sufficiently small in comparison to $\kappa_1$  then condition \eqref{cond:1} is always satisfied as we show in Theorem \ref{th:theorem_slow} c).

Before stating our main results of this section, Theorems \ref{th:theorem_slow} and \ref{theoremP}, we first summarize some more properties of the function $f$. %and $\theta$ introduced in Lemma~\ref{fandtheta}.}
\begin{lemma}\label{lemma:omega}
Given $(\kappa_1,\dots,\kappa_m)$ with $\kappa_1>0$, $\kappa_m>0$, and $\kappa_j\geq 0$ for $1<j<m$. Denote $f(k)=|\omega(k)|$ for $0\leq k \leq 1$ as introduced in Lemma~\ref{fandtheta} and define $v_0:=(\sum_{s=1}^m s^2\kappa_s)^{\frac{1}{2}}$. Then the following holds:
\begin{itemize}
\item[a)] $\;f(0)=f''(0)=0\,$, $\;f'(0)=2\pi v_0\,$, and $\;f'''(0)=-\frac{2 \pi^3}{v_0}\sum_{s=1}^m s^4\kappa_s\,$.
\item[b)] $f'([0, 1]) = [-2\pi v_0,2\pi v_0]$. $f'$ attains its maximum only at $k=0$ and its minimum only at $k=1$.
\item[c)] Fix $\kappa_1 >0$. Then the map $f$ can be extended as a $C^\infty$-function of the variables $(k,\kappa_2,\ldots,\kappa_m)$ on the set $[0,1]\times[0, \infty)^{m-1}$. 
%and we have for the $k$-derivatives at $k=0$: $\;f(0)=f''(0)=0\,$, $\;f'(0)=2\pi v_0\,$, and $\;f'''(0)=-\frac{2 \pi^3}{v_0}\sum_{s=1}^m s^4\kappa_s\,$.
\end{itemize}  
\end{lemma}
\begin{proof}
Statement a) follows directly from the last formula in the proof of Lemma~\ref{fandtheta} and from the expansion sinc$^2(x) = 1 -\frac{x^2}{3} +\cO(x^4)$ for small values of $x$:
%definition of $\omega$ in~\eqref{eq:general_dispersion}. From~\eqref{eq:general_dispersion} we  learn that $|\omega(k)|^2 \geq 2 \kappa_1 (1-\cos(2\pi k)) > 0$ for $k\in (0, 1)$. Thus the smoothness of $f$ only needs to be investigated for $k\in\{0,1\}$. By the symmetry \eqref{omegasym} we only need to study the case $k=0$.
% We use the sinus cardinalis function sinc$(x)=\frac{\sin(x)}{x}$ together with the expansion sinc$^2(x) = 1 -\frac{x^2}{3} +\cO(x^4)$ to obtain
\begin{equation*}\label{rep_f}
f(k)=2\pi k \left(\sum_{s=1}^m s^2\kappa_s\mbox{ sinc}^2(\pi s k)\right)^{1/2}=\; 2 \pi v_0 k-\frac{\pi^3}{3v_0}\left(\sum_{s=1}^m s^4\kappa_s\right) k^3 + \cO(k^5)\,.
\end{equation*}
This representation also settles statement c). As we know already $f'(0)=2\pi v_0=-f'(1)$ we may establish statement b) by verifying that $|f'(k)|<2\pi v_0$ holds for all $k\in (0, 1)$. To this end we write $f= (\sum_{s=1}^m h_s^2)^{1/2}$ with $h_s(k) = 2\sqrt{\kappa_s}\sin(\pi s k)$. Using the Cauchy-Schwarz inequality we obtain for $0<k<1$ that
\[
|f'(k)| = \frac{|\sum_{s=1}^m h_s(k) h_s'(k)|}{\left(\sum_{s=1}^m h^2_s(k)\right)^{1/2}} \leq \left(\sum_{s=1}^m (h'_s)^2(k)\right)^{1/2} = 2\pi \left(\sum_{s=1}^m s^2\kappa_s \cos^2(\pi s k)\right)^{1/2} < \; 2\pi v_0 \,,
\]
where the last inequality follows from $|\cos(\pi k)|<1$ and $\kappa_1 > 0$.
\end{proof}

We are now ready to state our first main result in this section.
\begin{theorem}
		\label{th:theorem_slow}
Let $m \in \N$, fix $\delta >0$, denote $f(k)=|\omega(k)|$ as introduced in Lemma~\ref{fandtheta}, and set 
\begin{equation}
\label{v0lambda0}
v_0 :=  \sqrt{\sum_{s=1}^m s^2\kappa_s}, \quad\lambda_0 := \frac{1}{2}\left(\frac{1}{v_0} \sum_{s=1}^m s^4\kappa_s \right)^{1/3}.
\end{equation}
\begin{itemize}
\item[a)] For all $\alpha$, $\alpha'=1,2,3$ we have rapid decay as $t \to \infty$, uniformly for $|j|>(v_0 + \delta)t$, i.e.
\begin{equation*}
S_{\alpha\alpha'}(j,t)= \cO  \left(t^{-\infty}\right).
%\quad\mbox{as $t \to \infty$, uniformly for $|j|>(v_0 + \delta)t$.} 
\end{equation*}
\item[b)] If $f''(k)<0$ for all $0<k\leq 1/2$ then as $t \to \infty$ the following holds uniformly for $|j|<(v_0 + \delta)t$:
\begin{eqnarray}
\label{S11A}
S_{11}(j,t)&=& \frac{1}{2 \beta \lambda_0 t^{1/3}} \left[ \mbox{Ai}\left(\dfrac{j-v_0t}{ \lambda_0 t^{1/3}}\right)+ \mbox{Ai}\left(-\dfrac{j+v_0t}{ \lambda_0 t^{1/3}}\right) \right]+\cO  \left(t^{-1/2}\right) \;= \; S_{22}(j,t)\,,\\
\label{S12A}
S_{12}(j,t)&=& \frac{1}{2\lambda_0t^{1/3}\beta}\left(\mbox{Ai}\left(-\frac{j+v_0t}{\lambda_0 t^{1/3 } }\right)-\mbox{Ai}\left(\frac{j-v_0t}{\lambda_0 t^{1/3 } }\right)\right)+\cO(t^{-\frac{1}{2}})= \; S_{21}(j,t)\,,\\
\label{S33A}
S_{33}(j,t)&=& \frac{1}{2 \beta^2 \lambda_0^2 t^{2/3}} \left[ \mbox{Ai}^2\left(\dfrac{j-v_0t}{ \lambda_0 t^{1/3}}\right)+ \mbox{Ai}^2\left(-\dfrac{j+v_0t}{ \lambda_0 t^{1/3}}\right) \right]+\cO  \left(t^{-5/6}\right)\,.
\end{eqnarray}
\item[c)]
For every $\kappa_1>0$ there exists $\varepsilon = \varepsilon(\kappa_1) > 0$ such that for all $(\kappa_2, \ldots, \kappa_m) \in [0, \varepsilon)^{m-1}$ we have $f''(k) < 0$ for all $0<k\leq 1/2$.
\end{itemize}
\end{theorem}

\begin{proof}

The rapid decay claimed in statement a) can be argued in the same way as \eqref{eq:rapid_decay} for $S_{11}=S_{22}$. 
Due to relations \eqref{eq:CzeroExplicit} and \eqref{S33} one only needs to consider $S_{12}$ and $S_{21}$. Indeed, using Lemma~\ref{fandtheta} one may show that the imaginary parts of the integrands used in the representation of $S_{12}$ and $S_{21}$ in \eqref{S12exp} below have smooth extensions to all $k\in\R$ that are $1$-periodic. This is all that is needed because $|\frac{\partial}{\partial k} \phi_\pm (k,\,j/t)|>2\pi\delta$ by Lemma~\ref{lemma:omega} b) uniformly for $k\in[0,1/2]$ and  $|j|>(v_0 + \delta)t$.

We have already argued above that conditions \eqref{cond:1}, \eqref{cond:2} suffice to derive the first claim of statement b) with $v_0=\frac{f'(0)}{2 \pi} > 0$ and $\lambda_0=\frac{1}{2\pi}|f'''(0)/2|^{1/3}$. The expressions for $f'(0)$ and $f'''(0)$ stated in Lemma~\ref{lemma:omega} a) justify the definitions of \eqref{v0lambda0}. 

Using the symmetry relations \eqref{omegasym2} and \eqref{thetasym} we derive a representation for $S_{12}$ and $S_{21}$ that is suitable for a steepest descent analysis
\begin{equation}
\begin{split}
\label{S12exp}
S_{12}(j,t) =&\frac{1}{\beta}\int_{0}^{1/2}\Big(\sin(f(k)t - 2\pi kj-\theta(k)) + \sin(f(k)t + 2\pi kj+\theta(k))\Big) \di k \\
=& \frac{1}{\beta}\Im \int_{0}^{1/2} \Big(e^{it\phi_- (k,\,j/t)}e^{-i\theta(k)} +  e^{it\phi_+ (k,\,j/t)} e^{i\theta(k)} \Big)d k \\
S_{21}(j,k)=&- \frac{1}{\beta}\Im \int_{0}^{1/2} \Big(e^{it\phi_- (k,\,j/t)}e^{i\theta(k)} +  e^{it\phi_+ (k,\,j/t)} e^{-i\theta(k)} \Big)d k \\
\end{split}
\end{equation}
where $\phi_\pm (k,\xi) =f(k)\pm2\pi\xi k$ as in \eqref{def:fxiphi} above.
Expanding for $k$ close to zero one obtains $\phi_\pm (k,j/t)=2\pi v_0 k-\frac{1}{3}(2\pi)^3\lambda_0^3k^3\pm2\pi k\frac{j}{t} +\cO(k^5)$.
%\todo{Th still thinks it is $k^5$ rather than $k^4$}
Substituting $y = 2 \pi \lambda_0 t^{1/3} k$ leads  to the asymptotic expression 
\[
t\phi_\pm (k,j/t)=\frac{v_0t\pm j}{\lambda_0t^{\frac{1}{3}}}y-\frac{1}{3}y^3 + \cO( t^{-\frac{2}{3}})\, \quad \mbox{as $t\to \infty$.}
\]
Keeping in mind that $\theta(0)=\frac{\pi}{2}$ we obtain 
\begin{eqnarray*}
S_{12}(j,t) &=& \frac{1}{2\lambda_0t^{1/3}\beta}\left(\mbox{Ai}\left(-\frac{j+v_0t}{\lambda_0 t^{1/3 } }\right)-\mbox{Ai}\left(\frac{j-v_0t}{\lambda_0 t^{1/3 } }\right)\right)+\cO(t^{-\frac{1}{2}})=S_{21}(j,t)\,.
\end{eqnarray*}
%\color{black}
%\th{
%\begin{eqnarray*}
%S_{12}(j,t) &=&\frac{1}{\beta}\int_{0}^{1/2}\Big(\sin(f(k)t - 2\pi kj) + \sin(f(k)t + 2\pi kj)\Big)\cos(\theta(k)) \di k \\
%                 &=& \frac{1}{\beta}\Im \int_{0}^{1/2} \Big(e^{it\phi (k,\,j/t)} +  e^{it\phi_- (k,\,j/t)} \Big)\cos(\theta(k)) d k \; = \; - S_{21}(j,t)\,,
%\end{eqnarray*}
%where $\phi (k,\xi) =f(k)-2\pi\xi k$ as in \eqref{def:fxiphi} above and $\phi_- (k,\xi) := f(k)+2\pi\xi k$. Due to the hypothesis on $f''(k)$ the only stationary phase point of 
%$\phi (\cdot,\xi)$ and $\phi_- (\cdot,\xi)$ of higher multiplicity occurs at $k_0=0$ with $\xi=v_0$ and $\xi=-v_0$, respectively. However, the contribution coming from $k_0=0$ is now bounded by $\cO (t^{-2/3})$ instead of $\cO (t^{-1/3})$ because the additional factor satisfies $\cos(\theta(k)) = \cO(k)$ for k near $k_0$ due to Lemma~\ref{lemma:omega} b). This proves the second asymptotics of statement b).
%}

Regarding the expansion for $t\to\infty$  of $S_{33}(j,t)$ it follows immediately from the expression \eqref{S33} and the expansions of $S_{\alpha\alpha'}(j,t)$ with $\alpha,\alpha'=1,2$.

Statement c) follows from the continuous dependence of the derivatives $f''$ and $f'''$ on the parameters $(\kappa_2, \ldots, \kappa_m)$ (see Lemma~\ref{lemma:omega} c) and from simple facts for the case of nearest neighbour interactions $f_{1}(k)=2\sqrt{\kappa_1}\sin(\pi k)$ discussed above. Indeed,
from $f''(0)=0$ and from $f'''_{1}(0) < 0$ it follows that there exists such an $\varepsilon >0$ such that $f'''(k) < 0$ and hence also $f''(k) < 0$ for $k$ in some region $(0, \delta)$ uniformly in $(\kappa_2, \ldots, \kappa_m) \in [0, \varepsilon)^{m-1}$. As $f''_{1}(k) < -2\pi^2\sqrt{\kappa_1}\sin(\pi \delta)$ for all $k \in [\delta, 1/2]$ we may prove the claim in this region by reducing the value of $\varepsilon$ if necessary.
\end{proof}

\begin{figure}[ht]
	\centering
	%\figuretitle{Correlation $C_{11}\, -\,  C_{22}$}
	\includegraphics[scale=0.18]{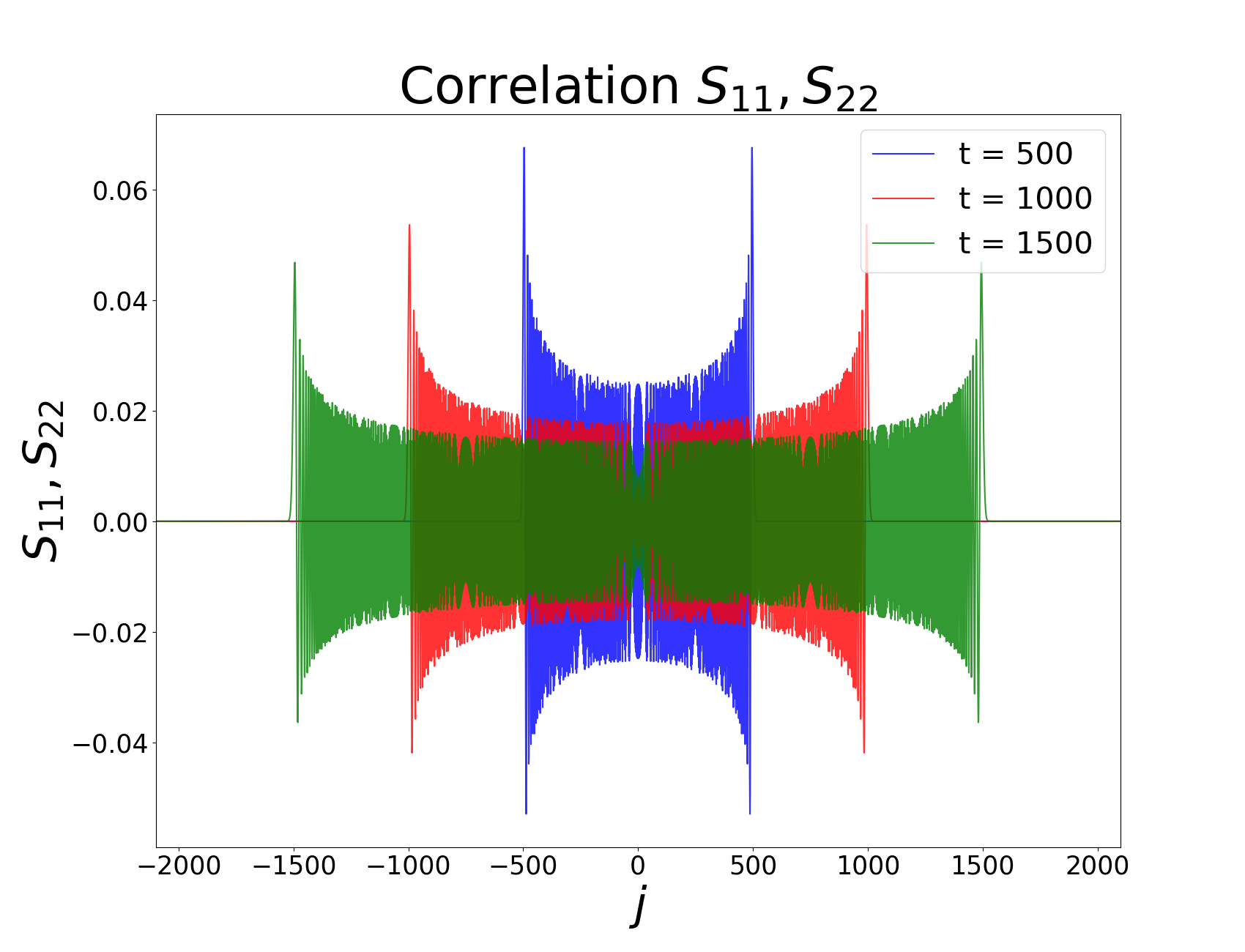}
	\includegraphics[scale=0.18]{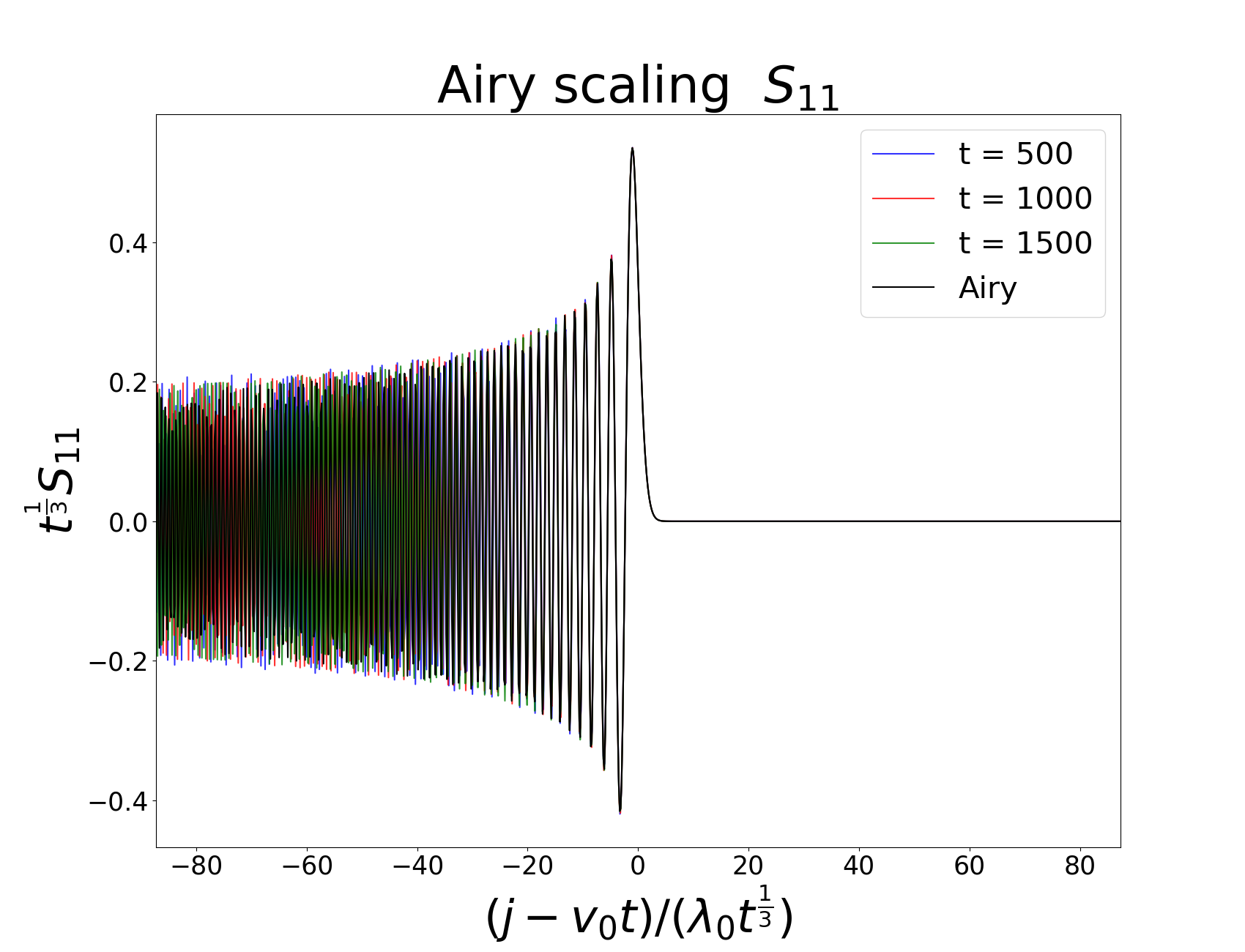}
	\includegraphics[scale=0.18]{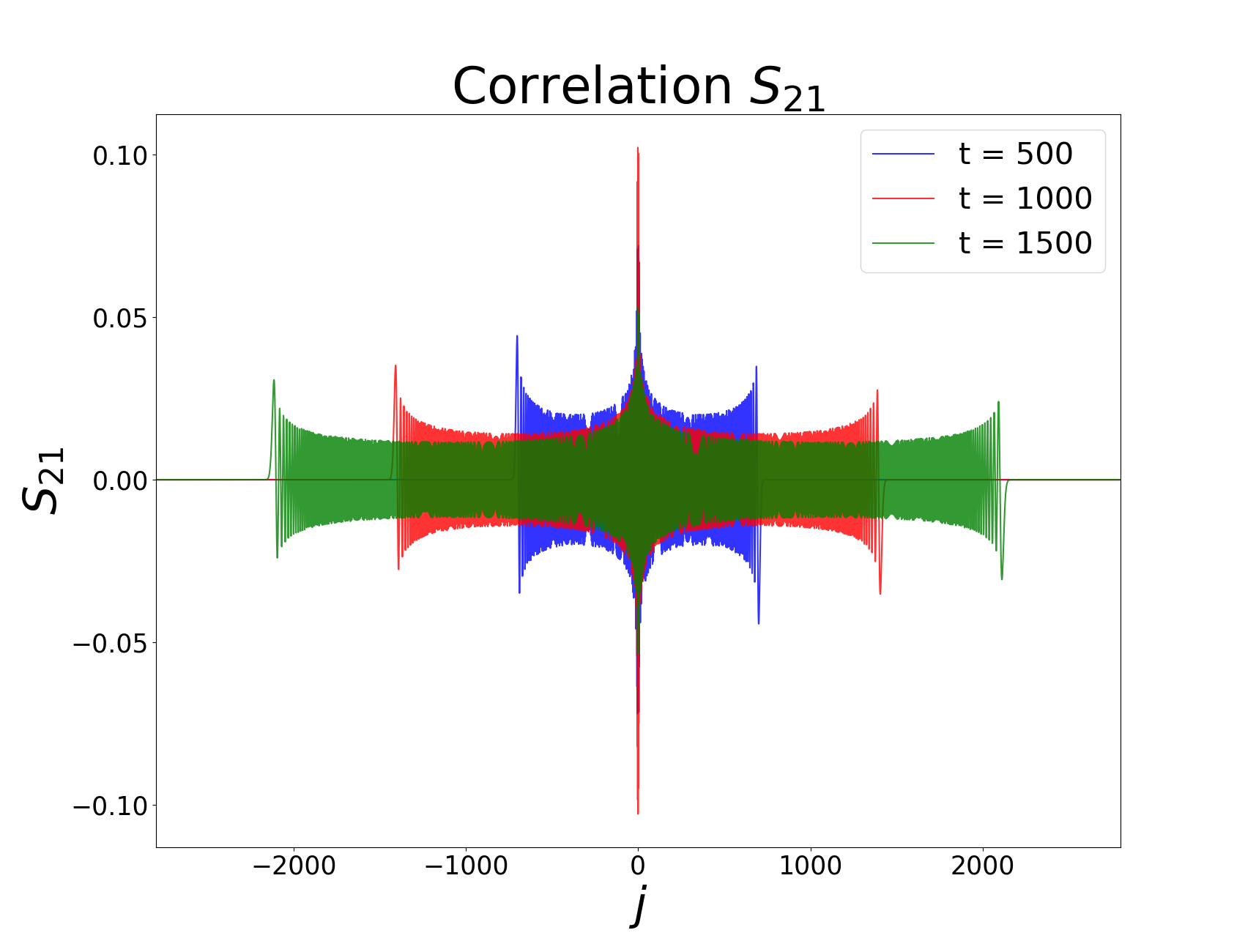}
	\includegraphics[scale=0.18]{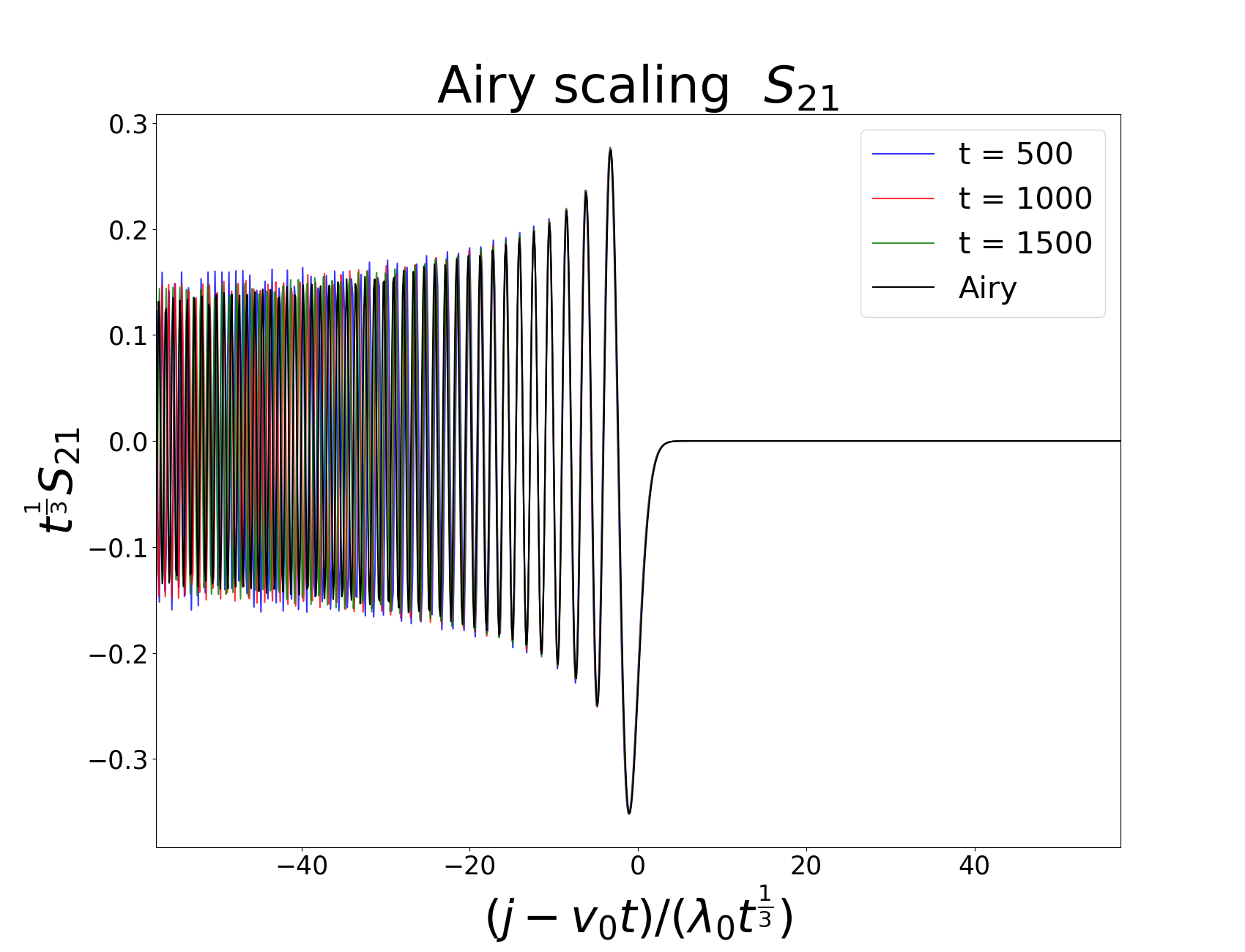}
		\includegraphics[scale=0.18]{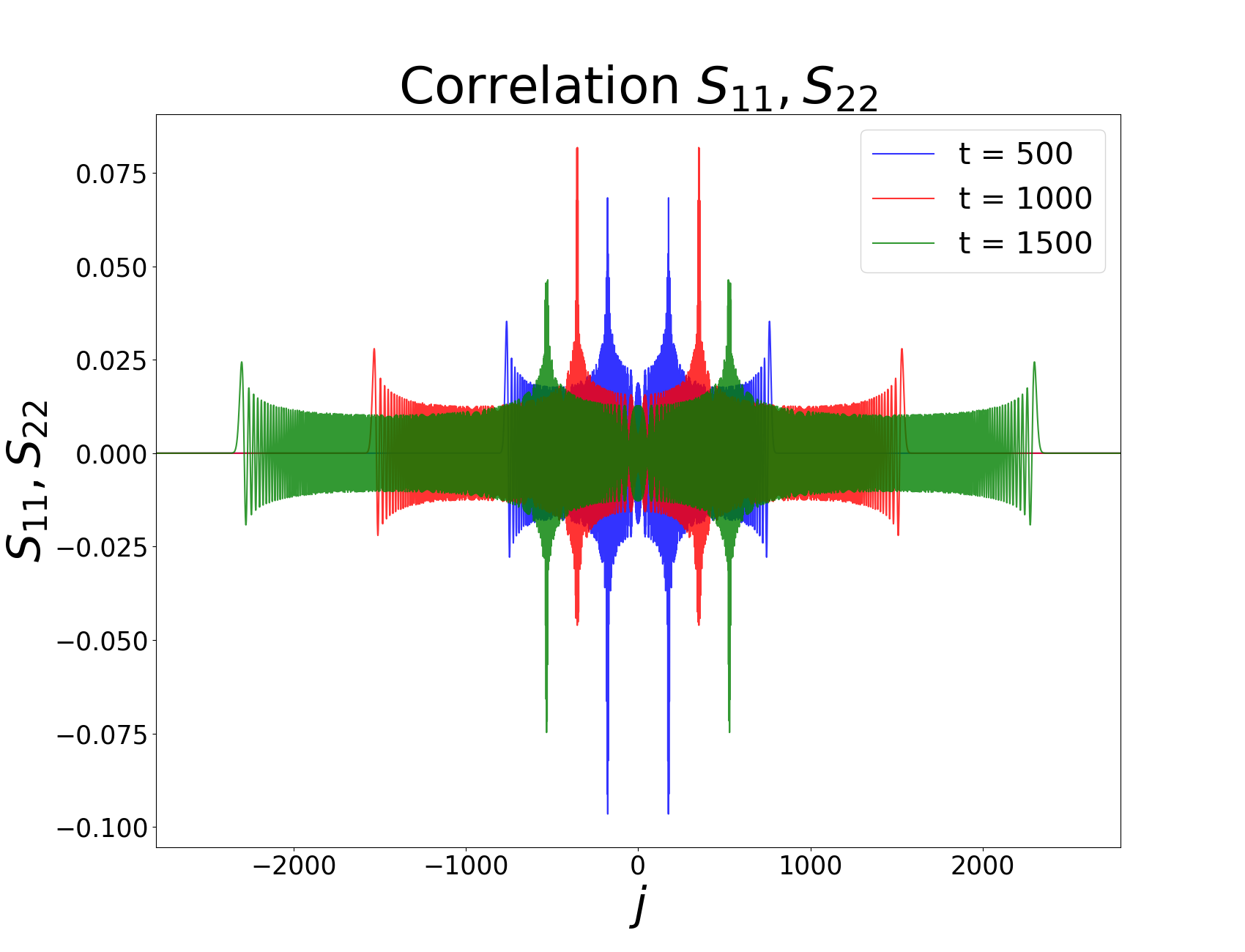}
		\includegraphics[scale=0.18]{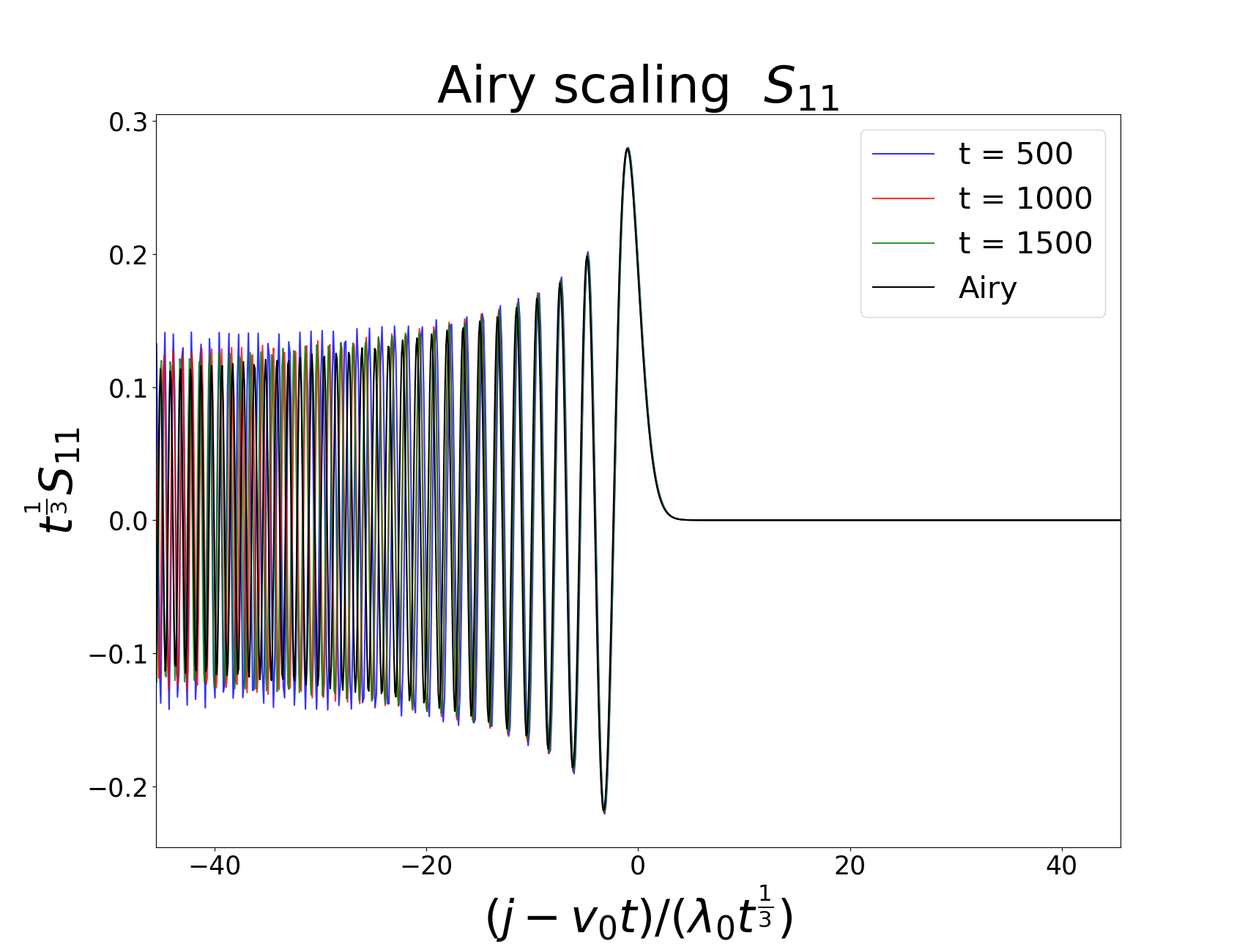}

	\caption{Correlation functions $S_{\alpha\alpha'}$ for the  harmonic oscillator with nearest neighbour interaction  with $\kappa_1=1$ (top left) and the harmonic potential with $\kappa_s=\frac{1}{s^2}$ for $s=1,2$ in Example~\ref{example1}  (center left) and the potential of Example~\ref{example2} in the bottom left.
	In the second column   the  Airy scaling \eqref{S11A}  of the corresponding fastest moving  peaks.
	 The Airy asymptotic is  perfectly matching the fastest peak and capturing several oscillations.}
\label{Figure1}
\end{figure}
 
% \color{red}
%\begin{remark}
%We observe that using the representation $\int_{-1}^1dk f(k)\sum_{j\in\mathbb{Z}}e^{2\pi i jk}=f(0)$,  and integration by parts,  one can easily obtain that 
%\[
%\sum_{j\in\mathbb{Z}}S(j,t)={\mathcal C},\quad \sum_{j\in\mathbb{Z}} jS(j,t)={\mathcal C}{\mathcal A}t
%\]
%where ${\mathcal C}=\mbox{diag}(\frac{1}{\beta},\frac{1}{\beta},\frac{1}{\beta^2})$ and ${\mathcal A}=\begin{pmatrix}0&v_0&0\\ v_0&0&0\\0&0&0\end{pmatrix}$
%where $v_0$ is defined in \eqref{v0lambda0}
%The normal mode coordinates  are defined as in \cite{Spohn2014}. They are obtained via the matrix $R$  that diagonalize ${\mathcal A}$ and such that $R{\mathcal C}R^T=1$, that is 
%\th{\todo{needs more explanation and reference}}
% \begin{equation}
% R = \begin{bmatrix}
%\sqrt{\frac{\beta}{2}}&-\sqrt{\frac{\beta}{2}}& 0 \\
%\sqrt{\frac{\beta}{2}}&\sqrt{\frac{\beta}{2}}& 0 \\
% 0 & 0 & \beta
% \end{bmatrix}\, .
% \end{equation}
%Then \th{\todo{confusion with $\tilde{S}$ introduced by Ken?}}
%\begin{equation}
%	\begin{split}
%	\tilde{S}(j,t)&:=RS(j,t)R^{-1} \\
%	&=\begin{bmatrix}
% S_{11}(j,t) - \frac{S_{21}(j,t) + S_{12}(j,t)}{2}& \frac{S_{12}(j,t) - S_{21}(j,t)}{2} &0\\
% -\frac{S_{12}(j,t) - S_{21}(j,t)}{2}& S_{11}(j,t) +\frac{S_{21}(j,t) + S_{12}(j,t)}{2}&0 \\
%0&0& S_{33}(j,t) 
%\end{bmatrix}\,.
%\end{split}
%\end{equation}
%Applying Theorem~\ref{th:theorem_slow} in the situation of statement b), one can observe that at leading order\th{\todo{$S_{33}$ is of lower order. Same in Spohn?}} the  block  matrix  $\tilde{S}(j,t)$ is diagonal   in each block as  $t\to \infty$ and $j\sim \pm v_0t$.
%\end{remark}
\color{black}
Theorem~\ref{th:theorem_slow} provides the leading order asymptotics of the limiting correlations $S_{\alpha\alpha'}(j,t)$ for $t\to \infty$ in the simple situation that the second derivative of the dispersion relation is strictly negative on the open interval $(0,1)$ (cf.~condition~\eqref{cond:1}). 
Moreover, statement c) shows that there is a set of positive measure in parameter space $\bkappa \in \R_+^m$ where this happens. For general values of $\bkappa$, however,  different phenomena may appear. In particular, there might exist stationary phase points of higher order leading to slower time-decay of the correlations (see discussion before the statement of Lemma~\ref{lemma:omega}). 
By a naive count of variables and equations one might expect that decay rates $t^{-1/(3+p)}$ occur on submanifolds of parameter space of dimension $m-p$. Theorem~\ref{theoremP} shows that this is indeed the case for $p=1$. 
Moreover, we present in this situation a formula for the leading order contribution of the corresponding stationary phase points to the asymptotics of $S_{\alpha\alpha'}(j,t)$. 
Despite being non-generic in parameter space it is interesting to note that decay rates $t^{-1/4}$ can be observed numerically (see Figures~\ref{fig_ex1} and~\ref{fig_ex2}).
There is also a second issue that may arise if condition~\eqref{cond:1} fails. Namely, for $v \in (-v_0, v_0)$ there can be several values of $k \in (0,\frac{1}{2}]$ satisfying $f'(k)\pm 2\pi v =0$ so that the contributions from all these stationary points need to be added to describe the leading order behaviour for $j$ near $vt$.
%Since the potential depends on $m$ parameters we can tune these parameters in such a way that the correlation functions have a  slower decay rate with time.
\begin{theorem}\label{theoremP}
Recall from \eqref{eq:general_dispersion} the formula for the dispersion relation$$f(k)=|\omega(k)| =\sqrt{2\sum_{s=1}^m\kappa_s\left(1-\cos(2\pi k s)\right)}\,.$$
%Let us consider the dispersion  relations \eqref{dispersion0} given  by the  function $$f(k) =\sqrt{2\sum_{s=1}^m\kappa_s\left(1-\cos(2\pi k s)\right)}.$$ 

\noindent
a) For $m\geq 3$
  there is  an  $(m-1)$-parameter family of potentials for which there exist $k^\ast=k^\ast(\bkappa)\in(0,\frac{1}{2})$ with
\begin{equation}
\label{t4_xi}
f''(k^\ast)=0,\;\;f'''(k^\ast)=0,\;\;f^{(iv)}(k^\ast)\neq 0,\;\;\mbox{ and } \;\;0 < v^\ast:=\frac{f'(k^\ast)}{2\pi}<v_0,
\end{equation}
%For $m\geq 3$
%  there is  a  $(m-2)$-parameter family of potentials such that
%\begin{equation}
%\label{t4_xi}
%f'(k^\ast)=2\pi v^\ast,\;\;f''(k^\ast)=0,\;\;f'''(k^\ast)=0,\;\;\quad k^\ast\in(0,\frac{1}{2}],\;0\leq  v^\ast<v_0,
%\end{equation}
with $v_0$ as in \eqref{v0lambda0}.
Set $\lambda^{\ast}:=\dfrac{1}{2\pi}(|f^{(iv)}(k^\ast)|/4!)^{\frac{1}{4}}>0$. 
Then for $j\to\infty$ and $t\to\infty$
in such a way that  
\[
\frac{v^\ast t - j}{\lambda^{\ast}t^{\frac{1}{4}}}
\]
is bounded, the contribution of the stationary phase point $k^\ast$ to the correlation functions is given by:
%and $v^\ast>0$. 
%Then for $|j|\to\infty$ and $t\to\infty$
%in such a way that  
%\[
%\frac{v^\ast t\pm j}{\lambda^{\ast}t^{\frac{1}{4}}}
%\]
%is bounded, the correlation functions have the expansion
\begin{align}
\label{Pearcey1}
S_{11}(j,t),S_{22}(j,t)\;:&\quad \dfrac{1}{2\beta\pi\lambda^{\ast}t^{\frac{1}{4}}}
\Re\left(e^{it\phi_-(k^\ast,j/t)}{\mathcal P}_\pm\left(\frac{v^\ast t-j}{\lambda^{\ast}t^{\frac{1}{4}}}\right)\right)+\cO(t^{-\frac{1}{2}})\,,\\
\label{Pearcey12}
S_{12}(j,t)  \;:&\quad \dfrac{1}{2\beta\pi\lambda^{\ast}t^{\frac{1}{4}}} \Im\left(e^{it\phi_-(k^\ast,j/t)-i\theta(k^\ast)}{\mathcal P}_\pm\left(\frac{v^\ast t-j}{\lambda^{\ast}t^{\frac{1}{4}}}\right)\right) + \cO(t^{-\frac{1}{2}})\,,\\
\label{Pearcey21}
S_{21}(j,t)  \;:&\quad -\dfrac{1}{2\beta\pi\lambda^{\ast}t^{\frac{1}{4}}} \Im\left(e^{it\phi_-(k^\ast,j/t)+i\theta(k^\ast)}{\mathcal P}_\pm\left(\frac{v^\ast t-j}{\lambda^{\ast}t^{\frac{1}{4}}}\right)\right) +\cO(t^{-\frac{1}{2}})\,,
\end{align}
%\begin{align}
%\label{Pearcey1}
%&S_{11}(j,t)  =S_{22}(j,t)=\dfrac{\Re\left(e^{it\phi_-(k^\ast,j/t)}{\mathcal P}_\pm\left(\frac{v^\ast t-j}{\lambda^{\ast}t^{\frac{1}{4}}}\right)\right)+
%\Re\left(e^{it\phi_+(k^\ast,j/t)}{\mathcal P}_\pm\left(\frac{v^\ast t+j}{\lambda^{\ast}t^{\frac{1}{4}}}\right)\right)}{2\beta\pi\lambda^{\ast}t^{\frac{1}{4}}}+\cO(t^{-\frac{1}{2}})\,,\\
%\label{Pearcey12}
%&S_{12}(j,t)  =\dfrac{\Im\left(e^{it\phi_-(k^\ast,j/t)-\theta(k^\ast)}{\mathcal P}_\pm\left(\frac{v^\ast t-j}{\lambda^{\ast}t^{\frac{1}{4}}}\right)\right)+
%\Im\left(e^{it\phi_+(k^\ast,j/t)+\theta(k^\ast)}{\mathcal P}_\pm\left(\frac{v^\ast t+j}{\lambda^{\ast}t^{\frac{1}{4}}}\right)\right)}{2\beta\pi\lambda^{\ast}t^{\frac{1}{4}}}+\cO(t^{-\frac{1}{2}})\,,\\
%\label{Pearcey21}
%&S_{21}(j,t)  =-\dfrac{\Im\left(e^{it\phi_-(k^\ast,j/t)+\theta(k^\ast)}{\mathcal P}_\pm\left(\frac{v^\ast t-j}{\lambda^{\ast}t^{\frac{1}{4}}}\right)\right)+
%\Im\left(e^{it\phi_+(k^\ast,j/t)-\theta(k^\ast)}{\mathcal P}_\pm\left(\frac{v^\ast t+j}{\lambda^{\ast}t^{\frac{1}{4}}}\right)\right)}{2\beta\pi\lambda^{\ast}t^{\frac{1}{4}}}+\cO(t^{-\frac{1}{2}})\,,
%\end{align}
where  $\phi_\pm (k,\xi) =f(k)\pm2\pi\xi k$, $\theta(k)=\arg \omega(k)$ as defined in Lemma~\ref{fandtheta}, ${\mathcal P}_\pm(a)$ denote the Pearcey integrals, cf.~Appendix~\ref{Appendix_D},
\begin{equation}
\label{Pearcey}
{\mathcal P}_\pm(a)=\int_{-\infty}^{\infty}e^{i(\pm y^4+ay)}dy,\;\;
%\overline{{\mathcal P}_-(a)}={\mathcal P}_+(-a), 
\;\;a\in\mathbb{R},
\end{equation}
and $ {\mathcal P}_\pm$ has to be chosen according to the sign of $f^{(iv)}(k^\ast)$.
If $j\to -\infty$ with bounded $(v^\ast t+j)/(\lambda^{\ast} t^{1/4})$ the contributions of the stationary point $k^\ast$ can be obtained from the ones presented in \eqref{Pearcey1}-\eqref{Pearcey21} by replacing $\phi_-$ by $\phi_+$, $\theta$ by $-\theta$, and $j$ in the argument of ${\mathcal P}_\pm$ by $-j$.

\noindent
b)
When $k^\ast=\frac{1}{2}$ one has $f'(1/2)=0$ and $f'''(1/2)=0$ by the symmetry \eqref{omegasym2}. 
For each $m \geq 2$ there is an $(m-1)$-parameter family of potentials so that $f''(1/2)=0$ and $f^{(iv)}(1/2)\neq 0$ holds in addition. In this case the contribution of  the stationary phase point $k^\ast=1/2$ to the correlation functions in the asymptotic regime $t\to\infty$ with bounded $j/t^{\frac{1}{4}}$ is given by ($\lambda^{\ast}$ defined as in statement a) with $k^\ast=\frac{1}{2}$)
 \begin{equation}
 \begin{split}
 \label{Pearcey0}
S_{11}(j,t), S_{22}(j,t) \;:&\quad \frac{(-1)^j}{2 \beta\pi\lambda^{\ast}t^{\frac{1}{4}}}\Re\left(e^{i t f(\frac{1}{2})}{\mathcal P}_\pm\left(\frac{j}{\lambda^{\ast}t^{\frac{1}{4}}}\right)\right)+\cO(t^{-\frac{1}{2}})\\
%\nonumber
S_{12}(j,t), S_{21}(j,t)  \;:&\quad -\mbox{sgn}\,(\!\sum\limits_{s \,odd}\tau_s)\frac{(-1)^j}{2 \beta\pi\lambda^{\ast}t^{\frac{1}{4}}}\Im\left(e^{i t f(\frac{1}{2})}{\mathcal P}_\pm\left(\frac{j}{\lambda^{\ast}t^{\frac{1}{4}}}\right)\right)+\cO(t^{-\frac{1}{2}})\\
S_{33}(j,t)  \;:&\quad \frac{1}{4 \beta^2\pi^2(\lambda^{\ast})^2t^{\frac{1}{2}}}\left| {\mathcal P}_\pm\left(\frac{j}{\lambda^{\ast}t^{\frac{1}{4}}}\right) \right|^2+\cO(t^{-\frac{3}{4}})\,.
\end{split}
\end{equation}
				\end{theorem}
				\begin{proof} We begin by proving formula \eqref{Pearcey1}  for the momentum or position correlations $S_{22}(j,t)=S_{11}(j,t)$ 
				under the assumption that we have found a $k^{\ast} \in(0,1/2)$ for which all the relations of \eqref{t4_xi} are satisfied.
				%Formulas \eqref{Pearcey12} and \eqref{Pearcey21} can be treated in a similar way.  
				From \eqref{eq:heuristic.1}  and Lemma~\ref{fandtheta}  we obtain
	\begin{align}
	\label{cor1}
	&S_{11}(j,t)=S_{22}(j,t)  =\frac{1}{\beta}\Re \int_{0}^{\frac{1}{2}}\left( e^{it(f(k)+2\pi k\frac{j}{t})} +e^{it(f(k)-2\pi k\frac{j}{t})}\right) d k.
			\end{align}
In order to compute the contribution of the stationary phase point $k^{\ast}$ to the large $t$ asymptotics of the integral in \eqref{cor1} we expand
%	By steepest descent analysis of the above integral, we observe that the statement of the theorem is achieved 
%if there is an additional stationary point $k^\ast$ of the function $f(k)$, such  that \eqref{t4_xi} is satified.
%Then
\[
f(k)=f(k^\ast)+2\pi v^\ast(k-k^\ast)+f^{(iv)}(k^\ast)(k-k^\ast)^4/4!+O((k-k^\ast)^5)\,.
\]
Introducing the change of variables 
$$y=2\pi \lambda^{\ast}(k-k^\ast)t^{\frac{1}{4}},\quad \lambda^{\ast}=\dfrac{1}{2\pi}(|f^{(iv)}(k^\ast)|/4!)^{\frac{1}{4}}$$
one obtains 
$$
tf(k)-2\pi j k=tf(k^\ast)-2\pi jk^\ast+y\frac{v^\ast t-j}{\lambda^{\ast}t^{\frac{1}{4}}}\pm y^4+\cO(t^{-\frac{1}{4}})
$$
where the $\pm$ sign is determined by the sign of $f^{(iv)}(k^\ast)$.
Then using the Pearcey integral \eqref{Pearcey}, the expansion \eqref{Pearcey1} can be derived in  a straightforward way from \eqref{cor1}.
In a similar way the expansions \eqref{Pearcey12} and \eqref{Pearcey21} 
%and  \eqref{Pearcey0}  
are  obtained by applying the above analysis to the expression \eqref{S12exp}.\\
% and \eqref{S33}  respectively. 

In the situation $k^{\ast}=1/2$ of statement b) one uses in addition that $t\phi_{\pm}(1/2, j/t)=tf(1/2) \pm j\pi$, $\omega(1/2)=-\sum_{s=1}^m\tau_s (1-\cos(\pi s))=-2\sum_{s \,odd}\tau_s$, see \eqref{eq:general_dispersion}, and consequently $e^{\pm i\theta(1/2)} = -$ sgn$(\sum_{s \,odd}\tau_s)$. The leading order contribution of the stationary phase point $k^{\ast}=1/2$ to the integral representation of, say, $S_{12}$ in \eqref{S12exp} is then given by
\[
-\mbox{sgn}\,\left(\!\sum\limits_{s \,odd}\tau_s\right)\frac{(-1)^j}{2 \beta\pi\lambda^{\ast}t^{\frac{1}{4}}}\Im\left(e^{i t f(\frac{1}{2})}\left(\int_{-\infty}^{0}e^{i(\pm y^4-wy)}dy + \int_{-\infty}^{0}e^{i(\pm y^4+wy)}dy\right)\right)
\]
with $w=\frac{j}{\lambda^{\ast}t^{\frac{1}{4}}}$. In this way and with the help of \eqref{S33} all relations of \eqref{Pearcey0} can be deduced.

We now show the existence of a codimension 1 manifold in parameter space that exhibits such higher order stationary phase points in the situation of b)
where $k^{\ast}=1/2$. As we have $f'''(1/2)=0$ by symmetry  \eqref{omegasym2} we only need to solve
%A solution of \eqref{t4_xi} can be  obtained by  choosing $k=\frac{1}{2}$,  so that   
\begin{equation}
\label{t4_1}
f''\left(\frac{1}{2}\right)=0\quad\mbox{which is equivalent to} \quad \sum_{s=1}^ms^2(-1)^{s+1}\kappa_s=0\,.
\end{equation}

The solution of the above equation is
\begin{equation}
\label{kappa_m}
\kappa_m=\frac{(-1)^m}{m^2}\sum\limits_{s=1}^{m-1}s^2(-1)^{s+1}\kappa_s.
\end{equation}
It is clear from the above relation that  for $m$ even, choosing $\kappa_1$ sufficiently big one has $\kappa_m>0$ while for $m$ odd, it is sufficient to choose $\kappa_{s+1}>\frac{s^2}{(s+1)^2}\kappa_s > 0$,   $s$ odd and  $1\leq s \leq m-2$.
Note that  in the situation of \eqref{kappa_m} $f^{(iv)}(\frac{1}{2})\neq 0$ holds iff $\sum_{s=1}^m\kappa_s s^4(-1)^{s+1}\neq 0$. This condition simply removes an $(m-2)$-dimensional plane from our manifold \eqref{kappa_m} which defines a hyperplane in the positive cone of the $m$-dimensional parameter space.
Therefore we have found an $(m-1)$-parameter  family of potentials such that the correlation functions decay as in \eqref{Pearcey0}.\\

Finally, we show for $m \geq 4$ our claim about the solution set of \eqref{t4_xi}. The case $m=3$ is treated in Example~\ref{example2}.
%Next we determine a solution of \eqref{t4_xi} with $0<v^\ast<v_0$.  The case $m=3$ is treated in example~\ref{example2}. Here we consider the case $m\geq 4$.
Our strategy is to first show that there exists a $\bkappa^{\ast}$ that satisfies $f''(1/4, \bkappa^{\ast}) = 0$, $f'''(1/4, \bkappa^{\ast}) = 0$, $f'(1/4, \bkappa^{\ast}) > 0$, and $f^{(iv)}(1/4, \bkappa^{\ast}) \neq 0$. We then invoke the Implicit Function Theorem to show the existence of the $(m-1)$-dimensional solution manifold, where the stationary phase point $k^{\ast}\sim 1/4$ may and will depend on the parameters.
%To simplify the computation, we fix the stationary point at $k^\ast=\frac{1}{4}$ 
%and we tune the  parameters $\kappa_s$   accordingly. 
The conditions $f''(\frac{1}{4}, \bkappa)=0$ and $f'''(\frac{1}{4}, \bkappa)=0$
imply
\begin{align}
\label{eq_odd}
f'''\left(\frac{1}{4}\right)=0\rightarrow&\sum\limits_{s \,odd}(-1)^{\frac{s-1}{2}}s^3\kappa_s=0,\\
\label{eq_all}
f''\left(\frac{1}{4}\right)=0\rightarrow& \left( 2\sum\limits_{s \,odd}\kappa_s + 2\sum_{s \,even}\kappa_s(1-(-1)^{\frac{s}{2}})\right)\sum_{s \,even}s^2\kappa_s(-1)^{\frac{s}{2}}-\left(\sum_{s \,odd}s\kappa_s(-1)^{\frac{s-1}{2}}\right)^2=0\,.
\end{align}
One needs to treat the case $m$ odd and even separately. Here we consider only the case $m$ even. The odd case can be treated in a similar way.
Equation  \eqref{eq_odd} gives
\[
\kappa_{m-1}=\dfrac{(-1)^{\frac{m}{2}}}{(m-1)^3}\sum\limits_{s \,odd,s=1}^{m-3}(-1)^{\frac{s-1}{2}}s^3\kappa_s.
\]
If $m=2\ell$ with $\ell$ even, a positive solution $\kappa_{m-1}$ exists, provided that $\kappa_1$ is sufficiently big. If $m=2\ell$ with $\ell$ odd then 
one needs to require $0<\kappa_{s}<\frac{ (s+2)^3}{s^3}\kappa_{s+2}$ for $s=1,5,9,\dots,m-5$.\\
The equation \eqref{eq_all}  is a linear equation in $\kappa_4$ and we solve it for $\kappa_4$ obtaining  
\[
\kappa_4=\frac{1}{32}\dfrac{\left(\sum\limits_{s \,odd,s=1}^{m-3}\kappa_s(-1)^{\frac{s-1}{2}}s(1-\frac{s^2}{(m-1)^2})\right)^2}{\sum\limits_{s \,odd,s=1}^{m-3}\kappa_s(1+\frac{s^3(-1)^{\frac{m+s-1}{2}}}{(m-1)^3}) + \sum\limits_{s \,even,s=2}^{m}\kappa_s(1-(-1)^{\frac{s}{2}})}+\frac{1}{16}\sum_{s \,even, s\neq 4, s=2}^{m}s^2\kappa_s(-1)^{\frac{s-2}{2}}\,.
\]
We observe that the first term in the above expression is always positive, while the second term is positive if we require  that $\kappa_{s}>\frac{ (s+2)^2}{s^2}\kappa_{s+2}>0$ for $s=6,10, 14,\dots,m-2$. The remaining two conditions $f'(1/4) > 0$ and $f^{(iv)}(1/4)\neq 0$ are easy to satisfy: The sign of $f'(1/4)$ agrees with the sign of $\sum_{s \,odd}s\kappa_s(-1)^{\frac{s-1}{2}}$ and can be made positive by choosing $\kappa_1$ sufficiently large. In the situation where
\eqref{eq_odd} and \eqref{eq_all} hold the fourth derivative $f^{(iv)}(1/4)$ does not vanish iff $\sum_{s \,even}s^4\kappa_s(-1)^{\frac{s}{2}} \neq 0$. This can be achieved by adjusting, for example, the value of $\kappa_2$. We have now shown that there exists $\bkappa^{\ast} \in \R_+^m$ such that the first four derivatives of $f$ have all desired properties at $k=1/4$. In order to obtain the $(m-1)$-dimensional solution manifold in parameter space, we apply the Implicit Function Theorem to $F(k, \bkappa) := (f''(k, \bkappa),  f'''(k, \bkappa))$. By a straight forward computation on sees that 
\[
\det \left[ \frac{\partial F}{\partial (k, \kappa_4)} (1/4, \bkappa^{\ast})
\right] = -  f^{(iv)}(1/4, \bkappa^{\ast}) \, \frac{\partial f''}{\partial \kappa_4}(1/4, \bkappa^{\ast}) \, \neq \, 0\,.
\]
We can therefore solve $F(k, \bkappa)=0$ near $(1/4, \bkappa^{\ast})$ by choosing $(k, \kappa_4)$ as functions of the remaining parameters $\kappa_j$ with $j\neq 4$.

\end{proof}
\begin{figure}[ht]
	\centering
	%\figuretitle{Correlation $C_{11}\, -\,  C_{22}$}
	\includegraphics[scale=0.18]{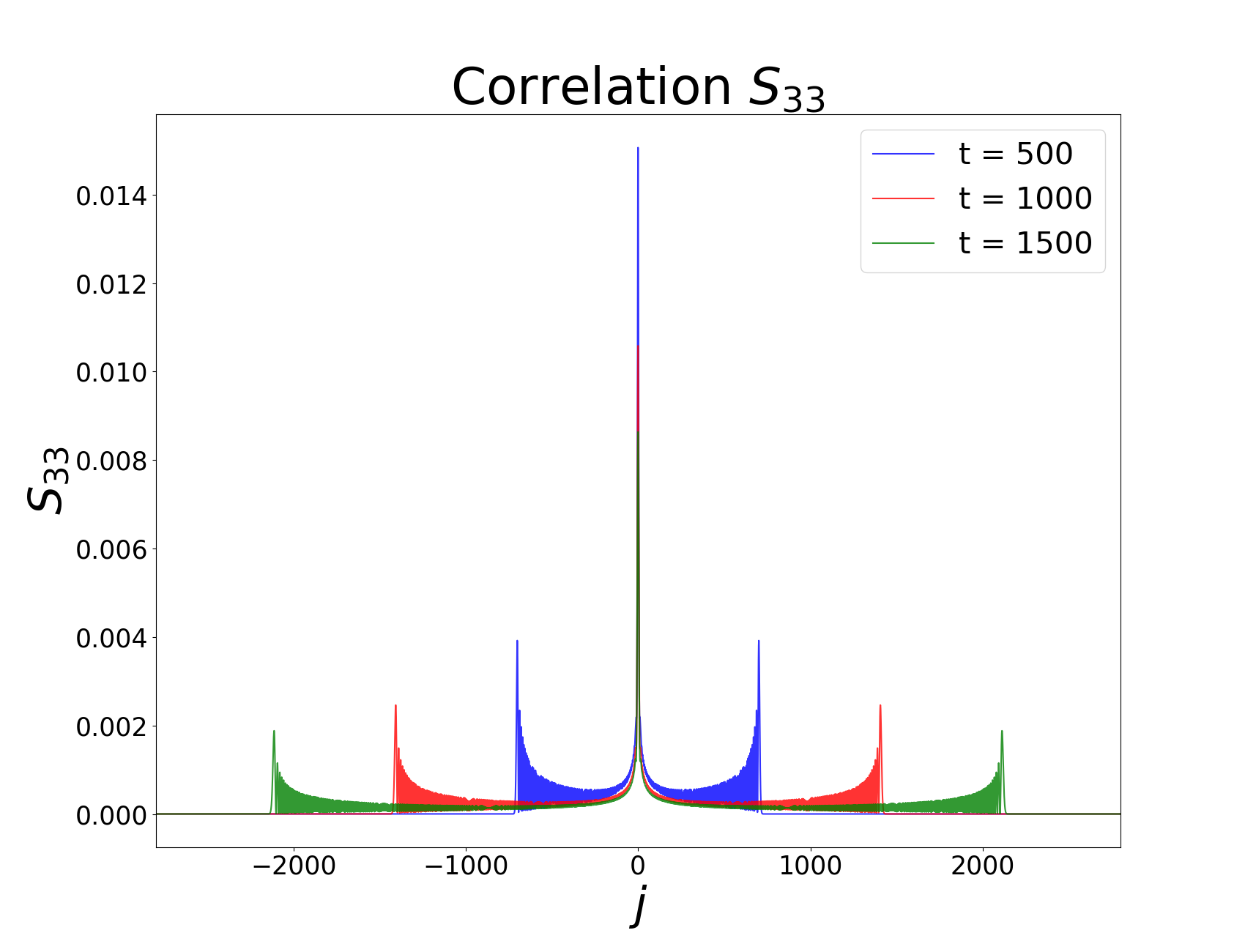}
\includegraphics[scale=0.18]{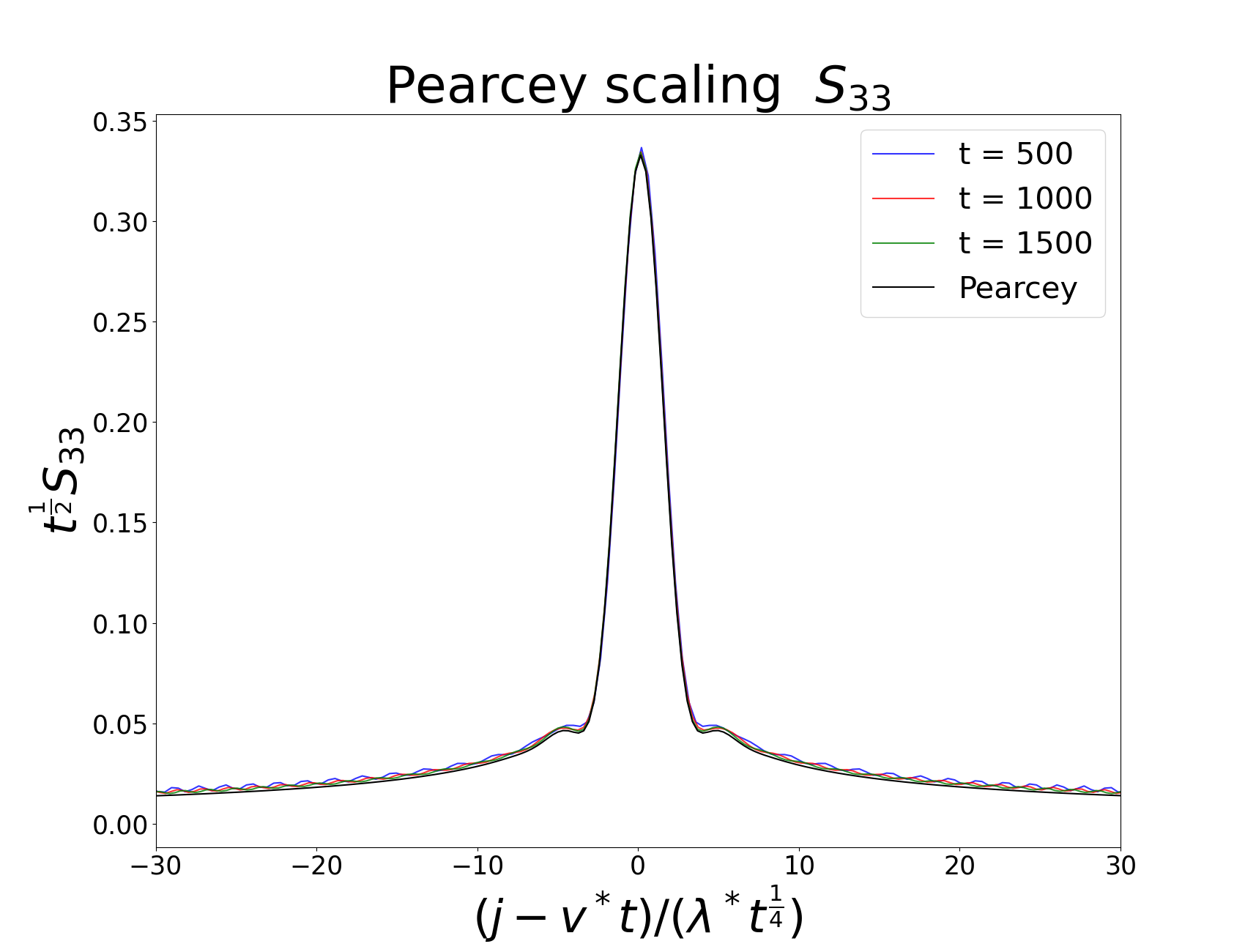}
%
%	\includegraphics[scale=0.5]{../images/Scalingharmonic/C11_Ex_0_simp_cor.png}
%	\includegraphics[scale=0.5]{../images/Scalingharmonic/C11_Ex_0_scaling_extreme_pick.png}
%\includegraphics[scale=0.5]{../images/Scalingharmonic/C33_Ex_2_simp_cor.png}
%\includegraphics[scale=0.5]{../images/Scalingharmonic/C33_Ex_0_scaling_center_pick.png}
%	i
	\caption{Correlation function $S_{33}(j,t)$ for the potential $\kappa_s=1/s^2$ for $m=2$ in Example~\ref{example1} for several values of time on the left.
	On the right one sees that the Pearcey scaling provided in \eqref{Pearcey0} matches perfectly %like $t^{\frac{1}{2}}$  
	for the central peak of $S_{33}(j,t)$. %one can see a perfect matching. 
	}
\label{fig_ex1}
\end{figure}

%\begin{remark}
%The Theorem~\ref{theorem_slow} can be generalized, with some further  effort, to the case of a decrease of the correlation functions $S_{\alpha\alpha'}(\xi t,t)$, $\alpha,\alpha'=1,2$,  like $t^{-\frac{1}{s}}$ with  $4\leq s<m-2$ for some $\xi<v_0$.
%\end{remark}

\begin{example}
\label{example1}
{\bf $m$ even}. Choosing $\kappa_s=\frac{1}{s^2}$  for $s=1,\dots, m$ one has that  conditions \eqref{t4_1} are satisfied and 
 $f^{(iv)}\left(\frac{1}{2}\right)<0$.

For $\kappa_s=\frac{1}{s^{\alpha}}$, $s=1,\dots, m-1$,   $2<\alpha<3$, and $\kappa_m$  given by \eqref{kappa_m},   there is $\alpha=\alpha(m)$ such that  $\kappa_{m}<\kappa_{m-1}$.\\
{\bf $m$ odd}.  Choosing $\kappa_s=\frac{1}{s}$,  for $s=1,\dots m-1$, one has  from \eqref{kappa_m} $\kappa_m=\frac{ m-1}{2m^2}<\kappa_{m-1}$  and $f^{(iv)}(\frac{1}{2})>0$.\\
In all these examples the correlation functions  $S_{\alpha\alpha'}(j,t)$, $\alpha,\alpha'=1,2$ decrease as $t^{-\frac{1}{4}}$ near $j=0$. 
\end{example}
%\begin{figure}[ht]
%	\centering
%	%\figuretitle{Correlation $C_{11}\, -\,  C_{22}$}
%	\includegraphics[scale=0.1]{../images/Scalingharmonic/C11_Ex_2_simp_cor.png}
%	\includegraphics[scale=0.1]{../images/Scalingharmonic/C21_Ex_2_simp_cor.png}
%\includegraphics[scale=0.1]{../images/Scalingharmonic/C33_Ex_2_simp_cor.png}
%\includegraphics[scale=0.1]{../images/Scalingharmonic/Ex_2_central_pick_C11.png}
%	\includegraphics[scale=0.1]{../images/Scalingharmonic/Ex_2_central_pick_C21.png}
%	\includegraphics[scale=0.1]{../images/Scalingharmonic/Ex_2_central_pick_C33.png}
%	\caption{Correlation function $S_{\alpha,\alpha'}$ for the potential of Example~\ref{example1}  with $\kappa_1=1$. 
%	In the second row  the plot on a $\log $  vertical scale  of the decay of the highly oscillating  central peak as a function of time is displayed. One can see the decay as $t^{-\frac{1}{4}}$ and $t^{-\frac{1}{2}}$  for the position-momentum and energy respectively.}
%\end{figure}
\begin{example}\label{example2}
We consider the case $m=3$ and we want to get a potential that  satisfies \eqref{t4_xi} with $v^\ast>0$. We chose as a critical point  of $f(k)$ the point  $k^\ast=\frac{1}{3}$ thus obtaining the equations
\[
\kappa_2=\frac{1}{8}\kappa_1,\quad \kappa_3=\frac{7}{72}\kappa_1\,.
\]
The speed of the peak is $v^\ast=\dfrac{\sqrt{2\kappa_1}}{4}$ and $f^{(iv)}(\frac{1}{3})=-\frac{68\sqrt{6}}{6}\pi^4\sqrt{\kappa_1}$.
 
 The correlation functions  $S_{\alpha\alpha'}(j,t)$, $\alpha,\alpha'=1,2$ decrease as $t^{-\frac{1}{4}}$  and $S_{33}(j,t)$ decreases like $t^{-\frac{1}{2}}$ as $t\to \infty$ and $j\sim v^\ast t$, see Figure~\ref{fig_ex2}. Note that one may obtain a $2$-parameter family of solutions of \eqref{t4_xi} by picking, for example, the particular solution related to $\kappa_1 = 1$ and by showing that the system of equations $(f'',f''')(k,\bkappa)=0$ can be solved near (1/3, 1, 1/8, 7/72) by choosing $k$ and $\kappa_3$ as functions of $\kappa_1$ and $\kappa_2$
 using the Implicit Function Theorem in the same way as at the end of the proof of Theorem~\ref{theoremP}.
 \begin{figure}[ht]
	\centering
	%\figuretitle{Example~\ref{example2}}
	\includegraphics[scale=0.18]{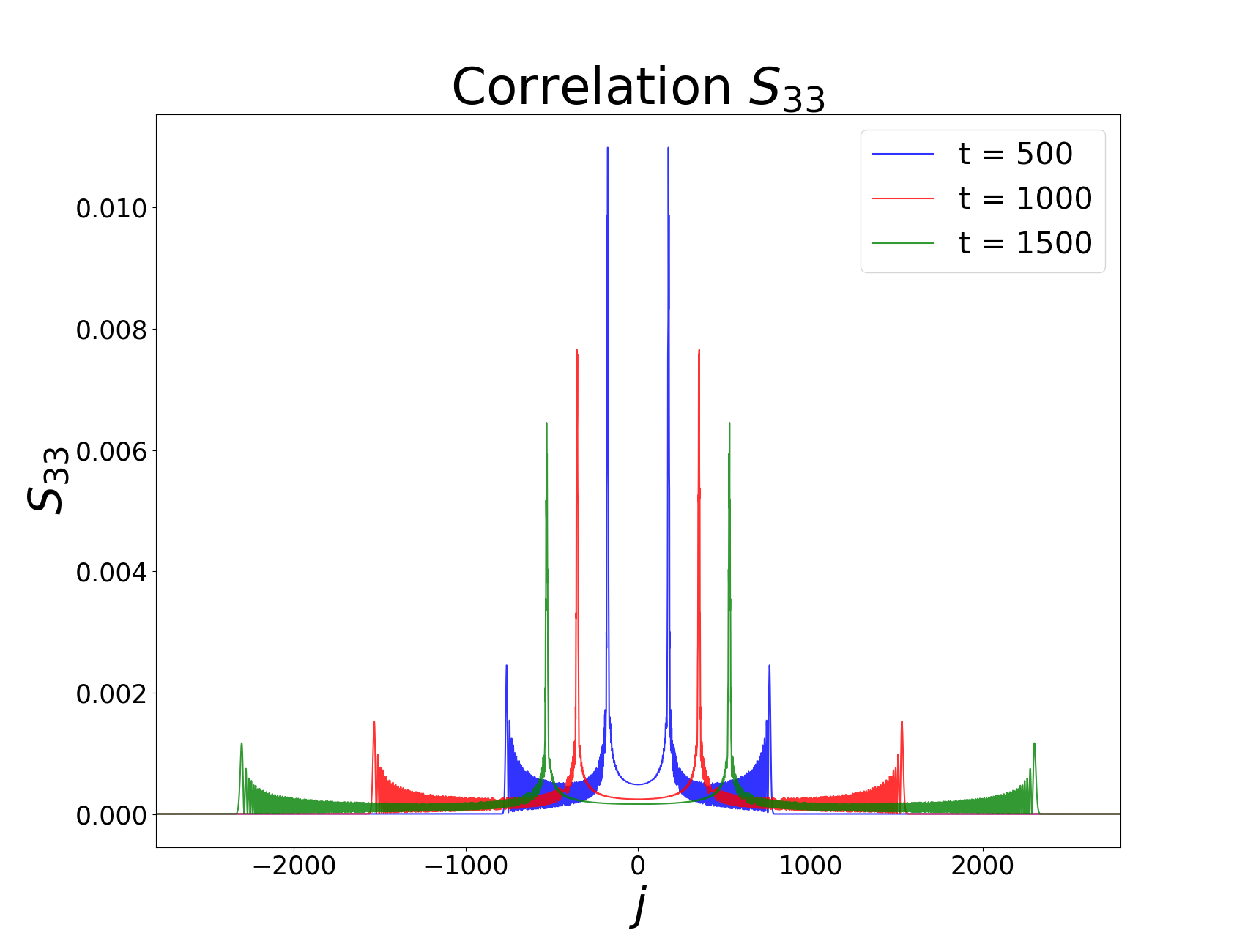}
	\includegraphics[scale=0.18]{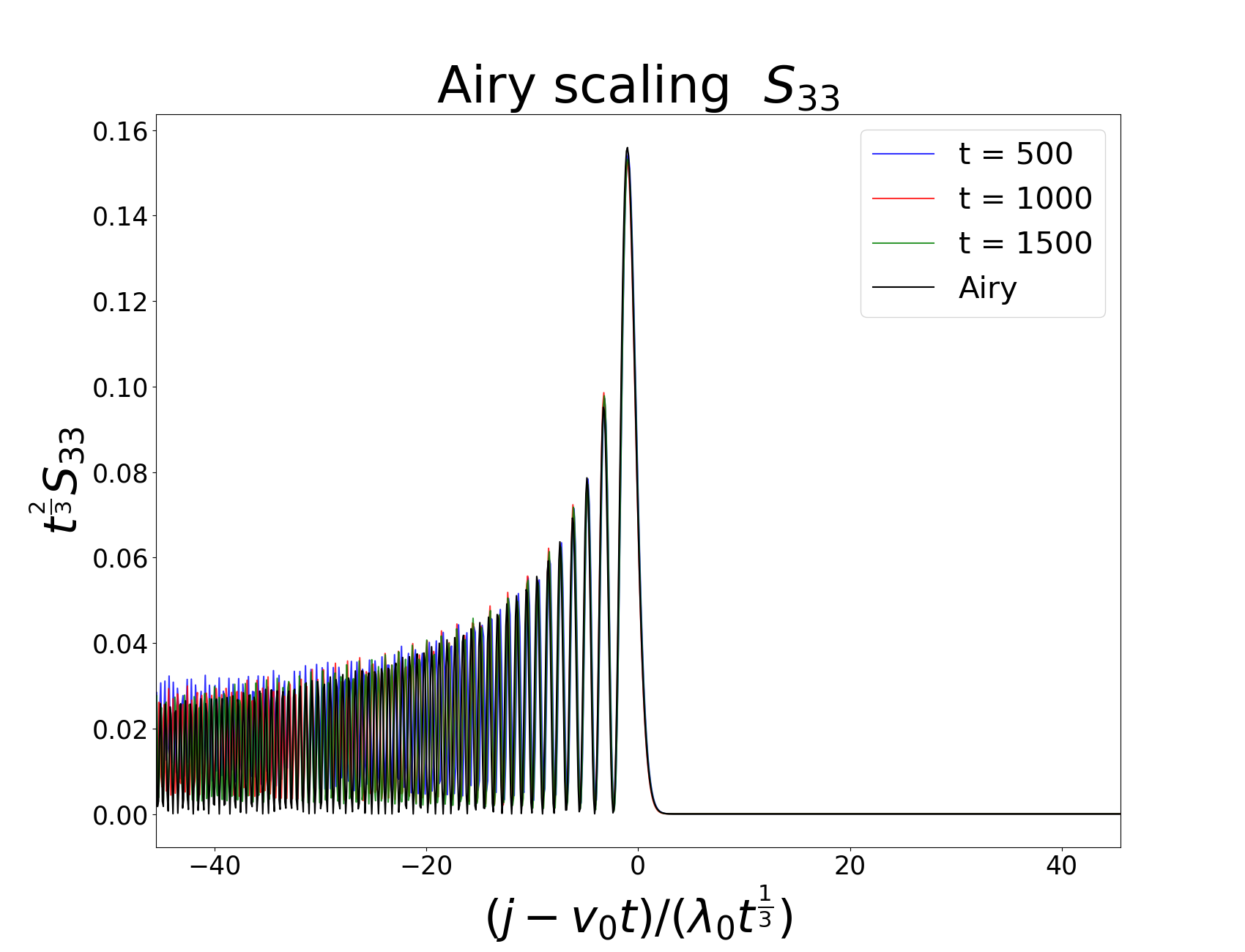}
	\includegraphics[scale=0.18]{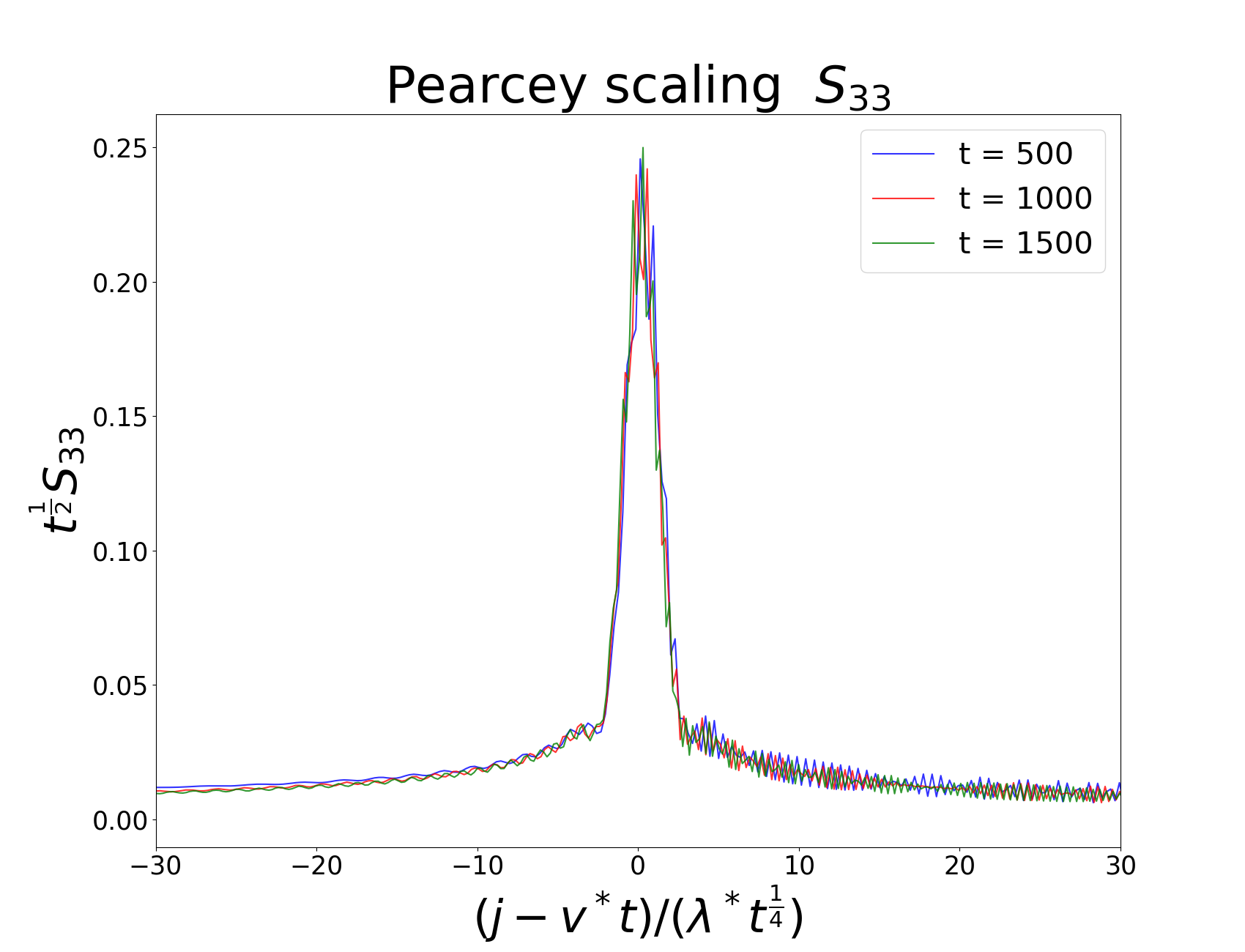}
	\includegraphics[scale=0.18]{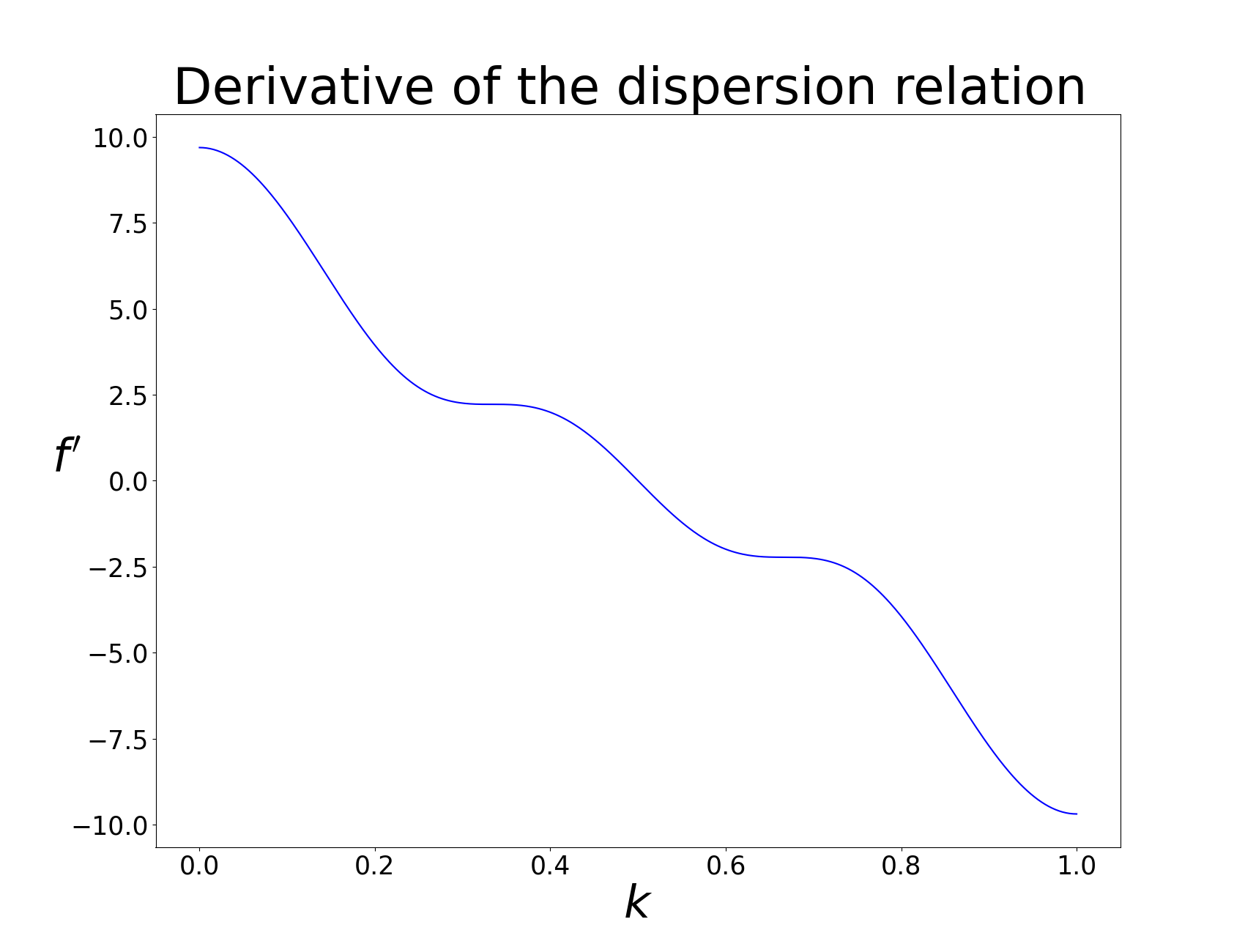}
	\caption{  Potential of Example~\ref{example2}. The top left figure displays $S_{33}(j,t)$  for several values of $t$. The scaling of $S_{33}$ according to the Airy function in Theorem~\ref{th:theorem_slow} for the fastest moving peak and the scaling of the slower moving peak according to
	the Pearcey integral are shown top right and bottom left, respectively. The corresponding critical points of the derivative of the dispersion function can be seen in the bottom right figure. }
\label{fig_ex2}
\end{figure}
 \end{example}

	\section{Complete set of integrals with local densities, currents and potentials, and some numerics for nonlinear versions}
	
	\subsection{Circulant hierarchy of integrals}
	In this section we construct a complete set of conserved quantities that have local densities.
	 The harmonic oscillator with short range interaction is clearly an integrable system. A  set of integrals of motion is given 
	  by the harmonic oscillators in each of  the  Fourier variables:  $\wh H_j=\frac{1}{2}\left( |\wh p_j| ^2 + |\omega_j|^2 |\wh q_j|^2  \right)$, $j=0,\dots \frac{N-1}{2}$.
However, when written in the physical variables $\bp$ and $ \bq$,  the quantities  
$$\wh H_j=\frac{1}{2}\sum_{k,l=0}^{N-1} \cF_{j,k} \overline{\cF_{j,l} }( p_k  p_l+|\omega_j|^2 q_k q_l)$$  depend on all components of the physical variables.
We now construct integrals of motion each having a density that involves only a limited number of components of the physical variables and this number only depends on the range~$m$ of interaction.

	For this purpose we denote by $\{\be_k\}_{k=0}^{N-1}$  the canonical  basis in $\R^N$.
	\begin{theorem}
		\label{thm:first}
		Let us consider the Hamiltonian 		
		\begin{equation}
		\label{eq:general_sys}
			H(\bp,\bq) = \frac{1}{2} \bp^\intercal \bp + \frac{1}{2}\bq^\intercal A \bq\,,
		\end{equation}
		with the symmetric circulant matrix $A$ as in \eqref{eq:A}, \eqref{A}.
	 Define the matrices $\{G_k\}_{k=1}^{M}$ to be the symmetric circulant matrix generated by  the vector  $\frac{1}{2}(\be_k + \be_{N - k})$
	  and $\{S_k\}_{k=1}^{M}$ to be the  antisymmetric  circulant matrix generated by the vector $\frac{1}{2}(\be_k - \be_{N - k})$.
	  Then  the family of Hamiltonians defined as  
			\begin{align}
			\label{eq:even_hamiltonian}
			H_k(\bp,\bq) =& \frac{1}{2} \bp^\intercal G_k \bp + \frac{1}{2}\bq^\intercal T^\intercal G_kT \bq=\frac{1}{2}\sum_{j=0}^{N-1}[p_jp_{j+k}+r_jr_{j+k}]\, ,\\ 
			\label{eq:odd_hamiltonian}
			H_{k+ \frac{N-1}{2}}(\bp,\bq) =&  \bp^\intercal T^\intercal S_k T \bq =\frac{1}{2}\sum_{j=0}^{N-1}\left[\left(\sum_{\ell=0}^m\tau_\ell p_{j+\ell}\right)(r_{j+k}-r_{j-k})\right]\,,\, \quad k=1,\dots,\frac{N-1}{2}
			\end{align} 
				together with $H_0:=H$ forms a complete family $(H_j)_{0\leq j \leq N-1}$ of integrals of motion that, moreover, is in involution.
				% 	 the set  $$\fH = \{H\} \, \bigcup_{k=1}^{N}\left(\{H^{G_k}\} \, \bigcup\, \{H^{S_k}\} \right)$$ is a complete set of integral of motion for $H$. We will call it {\em circulant hierarchy}.
	\end{theorem}   

%Before proving the theorem, let us observe that  the Hamiltonian densities    of each of the above Hamiltonians  depends only on a  finite  subset of all the  dependent variables.
%%For example
%\[
%H^{G_k}(\bp,\bq) =\sum_{j=0}^{N-1}H^{G_k}_j,\quad H^{G_k}_j=p_jp_{j+k}+r_jr_{j+k}
%\]
%\[
%H^{S_k}(\bp,\bq) =\sum_{j=0}^{N-1}H^{S_k}_j,\quad H^{S_k}_j=\left(\sum_{\ell=0}^m\tau_\ell p_{j+\ell}\right)(r_{j+k}-r_{j-k})
%\]
%\left[p_j(p_{j+1}+p_{j-1})+r_j(r_{j+1}+r_{j-1})\right]
%\
%where $\br=T\bq$ and $\br=(r_0,\dots, r_{N-1}). $ Each entry $r_j=\sum_{\ell=0}^{m-1}\tau_\ell q_{\ell+j}$ depends on $m$ variables $q_{j},\dots, q_{m+j-1}$.
%Therefore the Hamiltonian density $H^{G_k}_j$  depends on at most $m+2$  variables $q_{j-1},\dots, q_{m+j}$ and three variables $p_{j-1},p_j,p_{j+1}$.

\begin{proof}
Observe first that the Hamiltonian $H_0=H$ is included in the description of formula \eqref{eq:even_hamiltonian} as $G_0$ equals the identity matrix.
Using the symmetries $G_k^\intercal =G_k$, $0\leq k\leq(N-1)/2$, the Poisson bracket 
$\{F, G\} = \langle \nabla_{\bq} F, \nabla_{\bp} G\rangle- \langle \nabla_{\bq} G, \nabla_{\bp} F\rangle$
%$\{F, G\} = \langle \frac{\partial F}{\partial \bq}, \frac{\partial G}{\partial \bp}\rangle- \langle \frac{\partial g}{\partial \bq}, \frac{\partial F}{\partial \bp}\rangle$ 
may be evaluated in the form
\[
\begin{array}{lll}
\{H_k, H_\ell\} &= \bq^\intercal \big( T^\intercal G_{k} T G_{\ell} - T^\intercal G_{\ell} T G_{k}
\big) \bp\,,\qquad &\mbox{for $0 \leq k, \ell \leq \frac{N-1}{2}$,} \vspace{.15cm}\\
\{H_k, H_\ell\} &=  \bp^\intercal \big( T^\intercal S_k T T^\intercal S_\ell T - T^\intercal S_\ell T T^\intercal S_k T
\big) \bq\,,\qquad &\mbox{for $\frac{N+1}{2} \leq k, \ell \leq N-1$,}  \vspace{.15cm}\\
\{H_k, H_\ell\} &= \bq^\intercal T^\intercal G_{k} T T^\intercal S_\ell T \bq -  \bp^\intercal T^\intercal S_\ell T G_k \bp
\,,\qquad &\mbox{for $0 \leq k  \leq \frac{N-1}{2}$, $\frac{N+1}{2} \leq \ell \leq N-1$.}
\end{array}
\]
All these expressions vanish. To see this, it suffices to observe that multiplication is commutative for circulant matrices and, for the bottom line, that $S_\ell$ is skew symmetric: $S^\intercal_\ell = -S_\ell$. 

\end{proof}
%\color{black}
Now we introduce the local densities corresponding to the just defined integrals of motion
\[
e_j^{(k)}=
\begin{cases}
& \frac{1}{2}\left(p_jp_{j+k} +  r_jr_{j+k}\right)\, , \mbox{for $k=1, \dots,\frac{N-1}{2}$} \\ 
&\left(\sum_{l=0}^m \tau_l p_{j+l}\right) \left(r_{j+k} - r_{j-k}\right),\;\; \mbox{for $k=\frac{N+1}{2}, \dots,N$}\,.	
\end{cases}
\]	
together with their correlation functions
%and  we  compute the following correlation functions and study their decay:
\begin{equation}
\label{eq:loc_field_cor}
		S_{(k+3,n+3)}^{(N)}(j,t) := \la e^{(k)}_j(t)e^{(n)}_0(0) \ra - \la e^{(k)}_j(t)\ra\la e^{(n)}_0(0)\ra\,. \\
	\end{equation}
and limits
\begin{equation}
\label{CIM}
S_{k,n}(j,t)=\lim_{N\to\infty}S^{(N)}_{k,n}(j,t).
\end{equation}
We present explicit formulas for the limits $S_{k, n}$ in Appendix~\ref{Appendix_E} from which one can deduce that they have the same scaling behaviour as the energy-energy correlation  function $S_{33}$ when $t\to\infty$.
%one can easily check that such quantities have the same scaling behaviour as the energy-energy correlation  when $t\to\infty$.
\subsection{Currents and potentials} 
In this subsection  we write the evolution with respect to time of $r_j$, $p_j$ and $e_j$ in the form  of a (discrete)  conservation law by introducing the currents.
 Each  conservation law has a potential function that is a Gaussian random variable.
In the final part of this subsection we determine the leading order behaviour of the variance of this  Gaussian random variable as $t\to\infty$ in the case of nearest neighbour interactions.

For introducing the currents we  recall that $\br=T\bq$ with $T$ as in \eqref{form:T}. Then  one has 
\begin{equation}
\begin{split}
\label{dotrp}
&\dot{r}_j=\sum_{\ell}T_{j\,\ell} p_\ell=\sum_{\ell=1}^{{m}}\tau_\ell( p_{j+\ell}-p_j),\quad r_{j+N}=r_j\\
&\dot{p}_j=-\sum_{\ell}T_{\ell \,j}  r_\ell=\sum_{\ell=1}^{{m}}\tau_\ell(r_j- r_{j-\ell}),\quad p_{j+N}=p_j,\;\;j=0,\dots,N-1.
\end{split}
\end{equation}	
To write the above equation in the form of a discrete conservation law we introduce  the local currents
\begin{align}
\label{J_rp}
&\cJ_j^{(r)} :=\sum_{s=0}^{m-1}p_{j+1+s}\sum_{\ell=s+1}^m\tau_\ell \, ,\quad \cJ_j^{(p)} :=   \sum_{s=1}^{m}r_{j+1-s}\sum_{\ell=s}^m\tau_\ell.
\end{align}
Then  the equations  of motion  \eqref{dotrp} can be written in the form
\begin{align}
&\dot{r}_j=\cJ_{j}^{(r)}-\cJ_{j-1}^{(r)}  \\
&\dot{p}_j=\cJ_{j}^{(p)}-\cJ_{j-1}^{(p)},\quad j=0,\dots,N-1.
\end{align}
From the above equations it is clear that   the momentum $p_j$ and the generalized elongation  $r_j$ are locally
conserved. The evolution of the energy 
$e_j:=\dfrac{1}{2}p_j^2+\dfrac{1}{2}r_j^2$ at position $j$ 
takes the form 
\begin{equation}
%\label{e_j}
\label{J_j}
\dot{e}_j=\cJ_j^{(e)}-\cJ_{j-1}^{(e)}\,, \quad \cJ_j^{(e)}=\sum_{s=1}^{m}\tau_s\sum_{\ell=0}^{s-1}r_{j+1-s+\ell}p_{j+1+\ell}.
\end{equation}

We remark that all the  currents $\cJ_{j}^{(r)}$, $\cJ_{j}^{(p)}$ and $\cJ_{j}^{(e)}$ are local quantities in the variables $\bq$ and $\bp$.
We  recall the notation of the introduction  
\[
\boldsymbol{u}(j,t)=(r_j(t),p_j(t),e_j(t)),
\]
and we introduce the vector of currents  $ \boldsymbol{J}(j,t)=(\cJ_{j}^{(r)}(t), \cJ_{j}^{(p)}(t), \cJ_{j}^{(e)}(t))\,.$
The equations of motion take the compact form
\[
\dfrac{d}{dt}\boldsymbol{u}(j,t)=\boldsymbol{J}(j,t)-\boldsymbol{J}(j-1,t).
\]
We define a potential function for the above conservation law
\[
\boldsymbol{\Phi}(j,t):=\int_0^t\boldsymbol{J}(j,t')dt'+\sum_{\ell=0}^j\bu(\ell,0).
\]
Then it is straightforward to verify that $\boldsymbol{\Phi}_t(j,t)=\boldsymbol{J}(j,t)$ and $\boldsymbol{\Phi}(j,t)-\boldsymbol{\Phi}(j-1,t)=\bu(j,t)$.
The quantities  $\Phi_1(j,t)$ and $\Phi_2(j,t)$  can be expressed as sums of independent centered Gaussian random variables and are therefore also  Gaussian random variables with zero mean and variance 
$\langle (\Phi_1(j,t))^2\rangle$  and $\langle (\Phi_2(j,t))^2\rangle$,  
%while  $\Phi_3(j,t)$  is a Gaussian  random  variable with mean $\frac{j+1}{\beta}$ and variance
%$\langle (\Phi_3(j,t)-\frac{j+1}{\beta})^2\rangle$, 
where  all  the averages are taken with respect to the distribution \eqref{eq:measure}, see also \eqref{eq:measureDFT}.
We  calculate the variance for the case of   the harmonic oscillator with nearest neighbour interactions.  In this particular case
\begin{equation}
\begin{split}
\label{Phi}
\Phi_1(j, t) &=\sqrt{\kappa_{1}}\int_{0}^{t} p_{j+1}(t')dt' +\sum_{\ell=0}^{j} r_{\ell}(0) =\sqrt{\kappa_{1}}(q_{j+1}(t)-q_0(0)) \,\\
\Phi_{2}(j,t) &=\sqrt{\kappa_{1}}\int_{0}^{t}r_{j}(t')dt' + \sum_{\ell=0}^{j} p_{\ell}(0) \ . \\
%\Phi_3(j,t) &=\sqrt{\kappa_{1}}\int_{0}^{t} p_{j+1}r_{j}(t')dt' + \sum_{\ell=0}^{j} e_{\ell}(0) \ .
\end{split}
\end{equation}

After some lengthy calculations one obtains:
\begin{align}
\label{Phi1_int}
&\lim_{N \to \infty}
\langle \left(
\Phi_1(j,t) \right)^{2}\rangle= \frac{2\kappa_{1}}{\beta} \int_{0}^{1} |\omega(k)|^{-2} \left[ 1 - \cos{\left( | \omega(k)|t \right)} \cos{\left( 2 \pi (j+1) k \right)} \right] dk\\
\label{Phi2_int}
&\lim_{N \to \infty}\langle\left( \Phi_{2}(j,t) \right)^{2} \rangle= \frac{2\kappa_{1}}{\beta} \int_{0}^{1} |\omega(k)|^{-2} 
( 1 -  \cos{(|\omega(k)|t)}) \cos{\left(2\pi  (j+1)k\right)}dk + \frac{j+1}{\beta} \,.
%&\lim_{N \to \infty} \mathbb{E}\left[ \left( \Phi_{3}(j,t) -\frac{j+1}{\beta}\right)^{2} \right]=
\end{align}
Evaluating the r.h.s. of the above expressions  in the limit $t\to\infty$  we arrive to the following theorem.
\begin{theorem}
\label{Theorem_Gaussian}
In the limit $N\to\infty$ and $t\to \infty$  the quantities   $\Phi_1(j, t)$   and $\Phi_2(j, t)$ defined in \eqref{Phi} are  Gaussian random variables that have the following large $t$ behaviour:
%In the limit $N\to\infty$ and $t\to \infty$  the quantities   $\Phi_1(j, t)$   and $\Phi_2(j, t)$defined in \eqref{Phi} are  Gaussian random variables that have the following asymptotic behaviour for $t$ and $j$ large:
\begin{equation}
\lim_{N \to \infty}\Phi_1(j, t) ={\mathcal N}(0,\sigma_1^2) \qquad \mbox{and} \qquad \lim_{N \to \infty}\Phi_2(j, t) ={\mathcal N}(0,\sigma_2^2)\,.
\end{equation}
The leading order behaviour of the variances $\sigma_1^2$ and $\sigma_2^2$ agrees. In the physically interesting region $\frac{|j|}{t}\leq\sqrt{\kappa_1}$ it is given by
\begin{equation}
\sigma_1^2=\frac{ t \sqrt{\kappa_{1}}}{\beta} + \cO\big(t^{\frac{1}{3}} \big) = \sigma_2^2 \,.
\end{equation}
%\begin{align}
%&\lim_{N \to \infty}\Phi_1(j, t) ={\mathcal N}(0,\sigma_1^2)\\
%&\sigma_1^2=\begin{cases}
%&\frac{ t \sqrt{\kappa_{1}}-j}{\beta} +
%\cO \left(\frac{1}{\sqrt{t}} \right),\;\;\mbox{if $0 <  \frac{j}{t}   < (1 - \epsilon) \sqrt{\kappa_{1}} $, $\epsilon>0$,}\\
%& \frac{t^{1/3}\kappa_1^{1/6}}{2 \beta 3^{1/3} \Gamma(1/3)} + \mathcal{O}(1),\;\;\mbox{  if $j - \sqrt{\kappa_{1}} t$ is bounded ,} \end{cases}\\
%&\lim_{N \to \infty}\Phi_2(j, t) ={\mathcal N}(0,\sigma_2^2),\;\;\;\sigma_2^2=\frac{ t \sqrt{\kappa_{1}}}{\beta} + \cO\left(
%\frac{1}{\sqrt{t}} \right).
%\end{align}
%The behaviour for $\frac{j}{t}$ negative can be obtained from the above by the obvious symmetries in (\ref{Phi1_int}) and (\ref{Phi2_int}).
\end{theorem}
The proof of the above theorem relies on steepest descent analysis of the oscillatory integrals in 
%\eqref{Phi1_int}  and 
\eqref{Phi2_int}.
%respectively. 
But because the integrand is actually quite large ( $ \sim C t^{2}$) near $k=0$, we consider the following Cauchy-type integral instead,
\begin{eqnarray}
F_{0}(z) =\frac{1}{2 \pi^{2} \beta} \int_{-1/2}^{1/2} \frac{ 1 -  \cos{(|\omega(k)|t)}}{(k-z)^{2}} \cos{\left(2\pi  (j+1)k\right)}dk \ ,
\end{eqnarray}
which gives the leading order asymptotic behaviour of the integrals appearing in
%(\ref{Phi1_int}) and 
(\ref{Phi2_int}), since 
\begin{eqnarray}
\frac{ 2 \kappa_{1}}{\beta}
\int_{-1/2}^{1/2} |\omega(k)|^{-2}
( 1 -  \cos{(|\omega(k)|t)}) \cos{\left(2\pi  (j+1)k\right)}dk - F_{0}(0) \rightarrow \ 0 \ \ \mbox{ as } t, j \rightarrow \infty \ .
\end{eqnarray}
For $\frac{|j|}{t}  < (1 - \epsilon ) \sqrt{\kappa_{1}}$, $\epsilon > 0$, the analysis of $F_{0}(z)$ is quite straightforward - a standard stationary phase calculation combined with a contour deformation to permit the evaluation at $z=0$.  For $t$ and $j$ growing to $\infty$ such that $\frac{|j|}{t} \approx \sqrt{\kappa_{1}}$, the analysis is more complicated because the point of stationary phase is encroaching upon the origin, where the integrand itself is actually large as $t \to \infty$.  For this case, one must construct a local parametrix, following quite closely the analysis presented in \cite{KriechKuijMcLMiller08}, and we omit the details of this analysis. In order to analyse $\Phi_1$ observe that the difference of the integrals in relations \eqref{Phi1_int} and \eqref{Phi2_int} is given by $\int_{0}^{1} |\omega(k)|^{-2} \left[ 1 -  \cos{\left( 2 \pi (j+1) k \right)} \right] dk$ which can also be treated by a stationary phase calculation combined with a contour deformation.  

\subsection{Nonlinear Regime}\label{sect3.2}
In this section we consider a nonlinear perturbation of the harmonic oscillators with short range interactions of the form

\begin{equation}
\label{HN} H(\bp,\bq)=\sum_{j=0}^{N-1}\frac{p_j^2}{2}+\sum_{s=1}^m\kappa_s\left(\frac{1}{2}\sum_{j=0}^{N-1}(q_j-q_{j+s})^2  + \frac{\chi}{3}\sum_{j=0}^{N-1}(q_j-q_{j+s})^3 + \frac{\gamma }{4}\sum_{j=0}^{N-1}(q_j-q_{j+s})^4\right) \,.
 \end{equation}

We consider  Example~\ref{example1} and Example~\ref{example2}  with different strengths of nonlinearity  namely
\[
\mbox{  $m=2, \, \kappa_1 = 1,$   $\kappa_2 = \frac{1}{4}$,}\;\;
\begin{cases} 
	&\mbox{$\chi=0.01$ and $\gamma=0.001$}\\
	&\mbox{$\chi=0.1$ and $\gamma=0.01$ }
\end{cases}
\]
\[
\mbox{  $m=3, \, \kappa_1 = 1,$   $\kappa_2 = \frac{1}{8}$, $\kappa_2 = \frac{7}{72}$, }\;\;
\begin{cases} 
	&\mbox{$\chi=0.01$ and $\gamma=0.001$}\\
	&\mbox{$\chi=0.1$ and $\gamma=0.01$ }
\end{cases}\, .
\]
We numerically compute and study the correlatios functions for these systems sampling the initial conditions according to the Gibbs measures of just their harmonic part at temperature $\beta^{-1} = 1$.

\begin{figure}[ht]
	\centering
	%\figuretitle{Example~\ref{example2}}
	\includegraphics[scale=0.16]{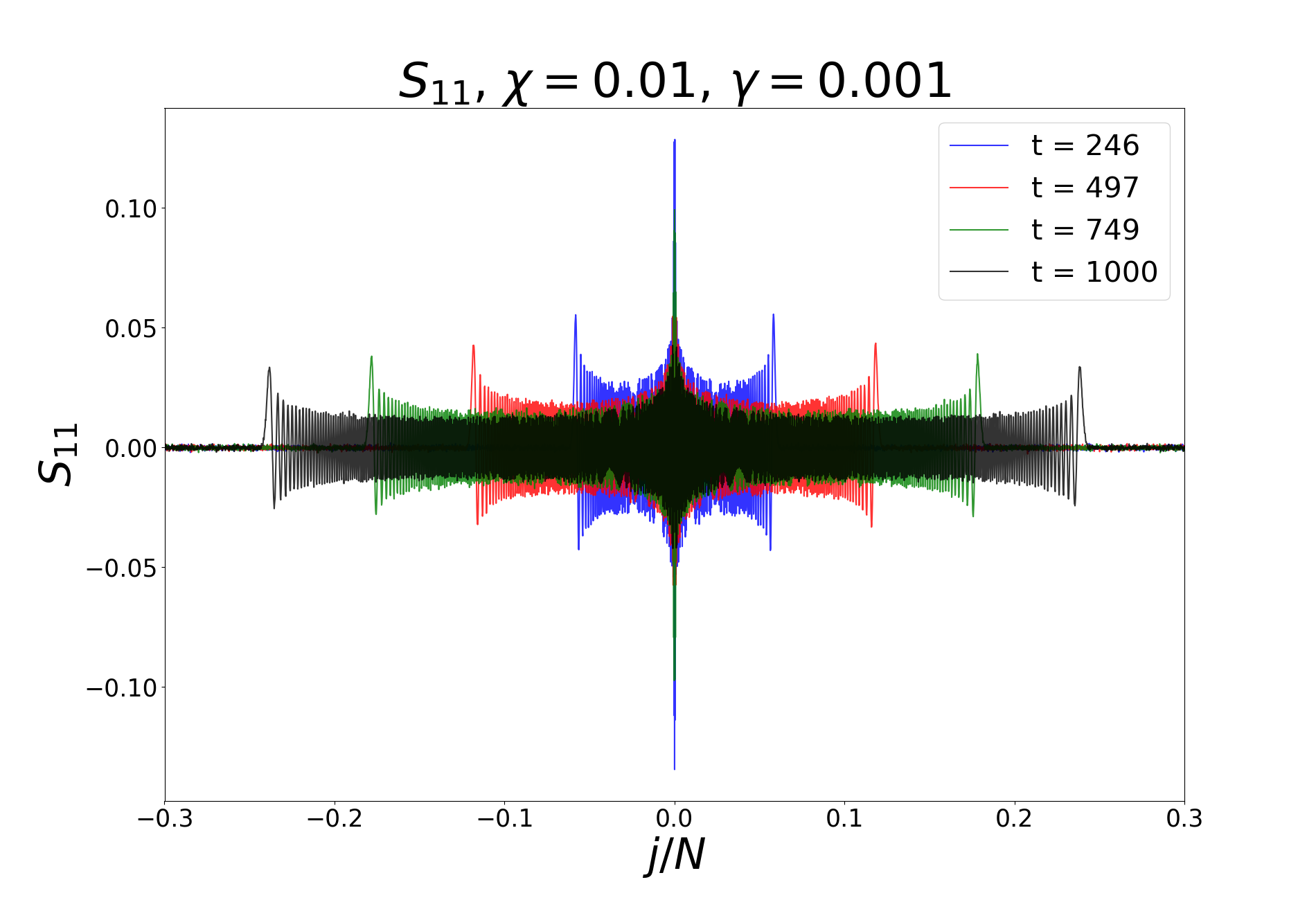}
	\includegraphics[scale=0.16]{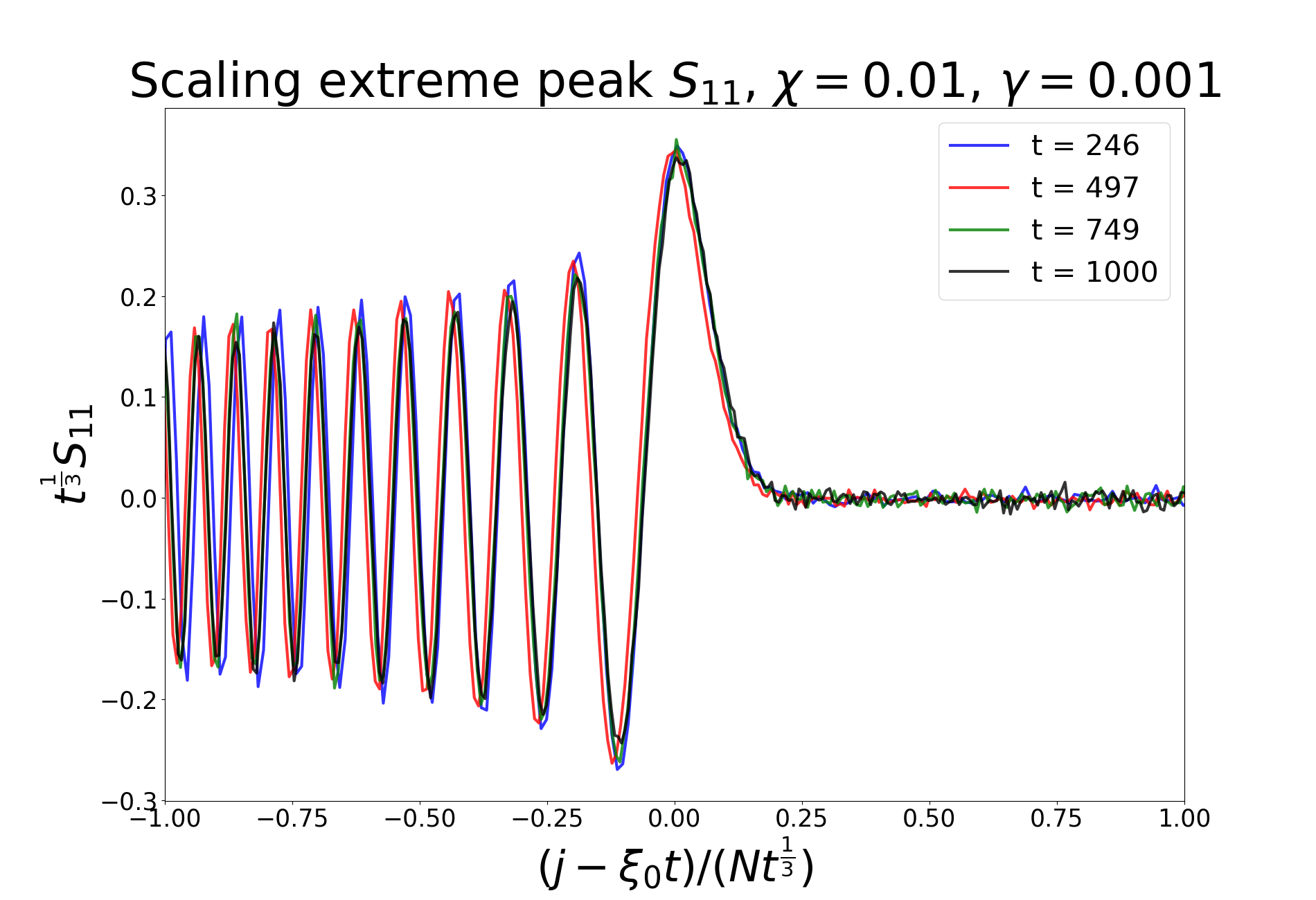}
	\includegraphics[scale=0.16]{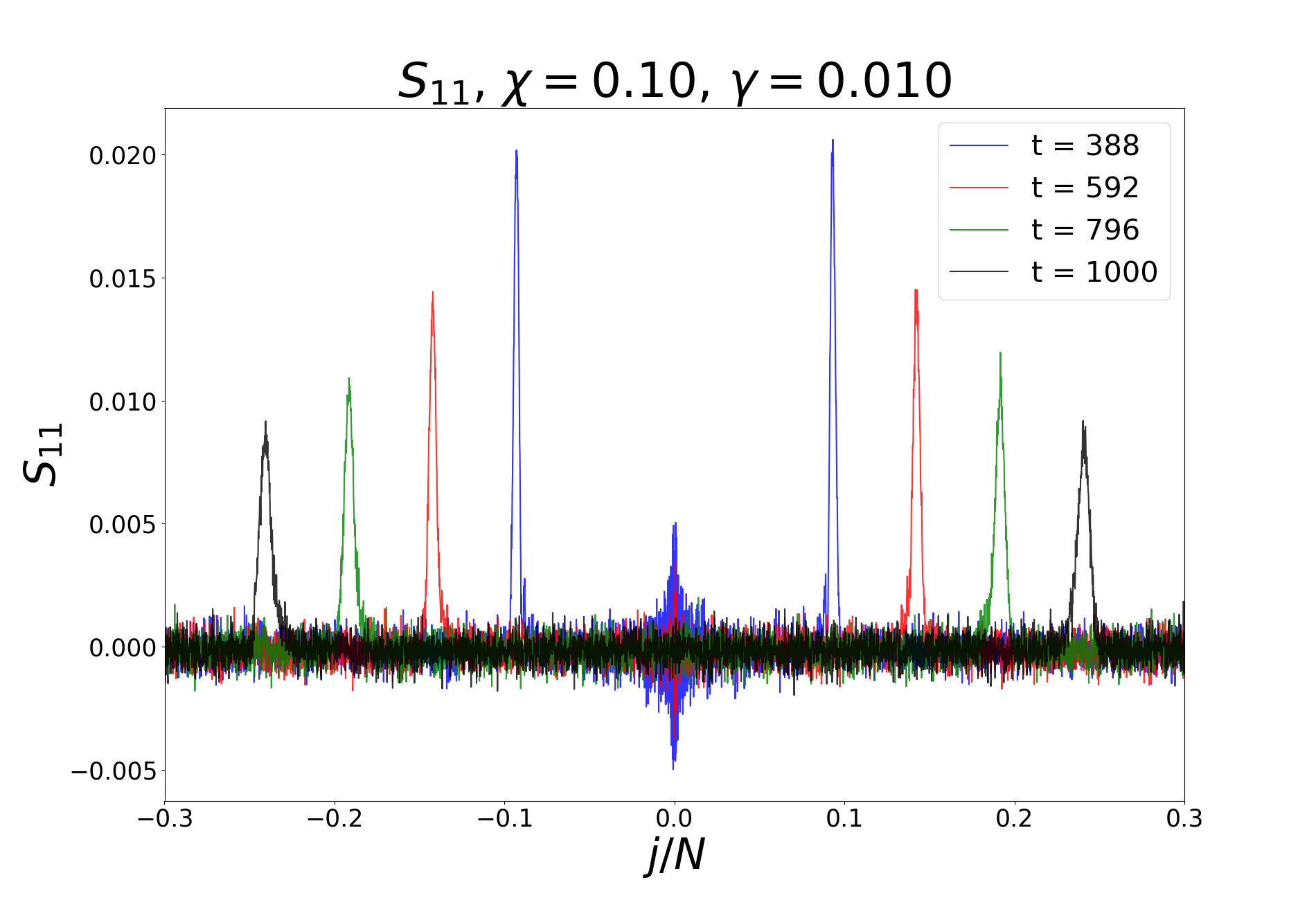}
	\includegraphics[scale=0.16]{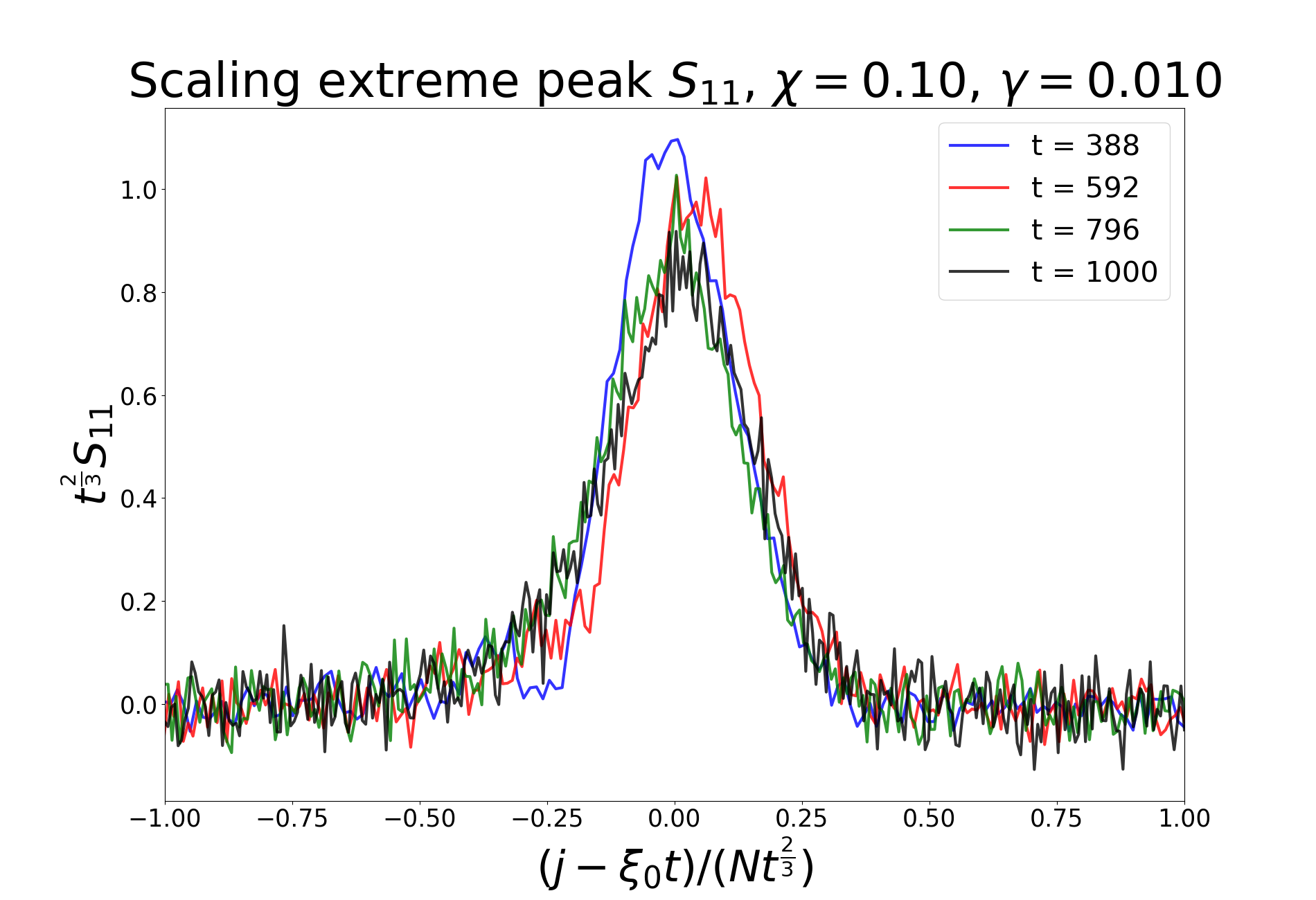}
		\caption{ Correlation function $S^{(N)}_{11}(j,t)$  for several values of times and for the  Hamiltonian \eqref{HN} with $\kappa_s$ as in  Example~\ref{example1},  $\chi=0.01$ and $\gamma=0.001$ in the top figure  and  $\chi=0.1$ and $\gamma=0.01$ in the lower figure.
		On the right  top  figure, the scaling of the fastest peak according to Airy parametrix (see Theorem~\ref{th:theorem_slow}  and Figure~\ref{Figure1})  and according to $t^{-2/3}$ in the lower figure. The speed $\xi_0$  of the fastest peak  is determined numerically.
		One can see that the central peak  has a low decay  in the top left figure, while  in the left bottom figure it is destroyed by  the  relatively  stronger nonlinearity.
		}
\label{fig_ex1_non}
\end{figure}
In  the weakly nonlinear case, the fastest  peaks of the correlation functions scale numerically according to the Airy parametrices  (cf.~Theorem~\ref{th:theorem_slow}) as can be deduced from 
the    top pictures in  Figures~\ref{fig_ex1_non}, \ref{fig_ex2_non}
while for stronger nonlinearity the fastest  peaks  seem to scale like $t^{\frac{2}{3}}$ in equation \eqref{TW}, see  bottom figures in   Figures~\ref{fig_ex1_non}, \ref{fig_ex2_non}.
The non generic peaks that are present in the linear cases and scale like $t^{1/4}$ have a fast decay  in the case of strong nonlinearity.
However for  weak nonlinearities, the central peak in the top left  Figure~\ref{fig_ex1_non}, still scales in time  like $t^{-\frac{1}{4}}$. Indeed performing a regression analysis of the log-log plot one can see 
a scaling like $t^{-0.267}$ that is slightly faster then $t^{-\frac{1}{4}}$  (see Figure~\ref{fig_ex3_non}).

  \begin{figure}[ht]
	\centering
	%\figuretitle{Example~\ref{example2}}
			\includegraphics[scale=0.16]{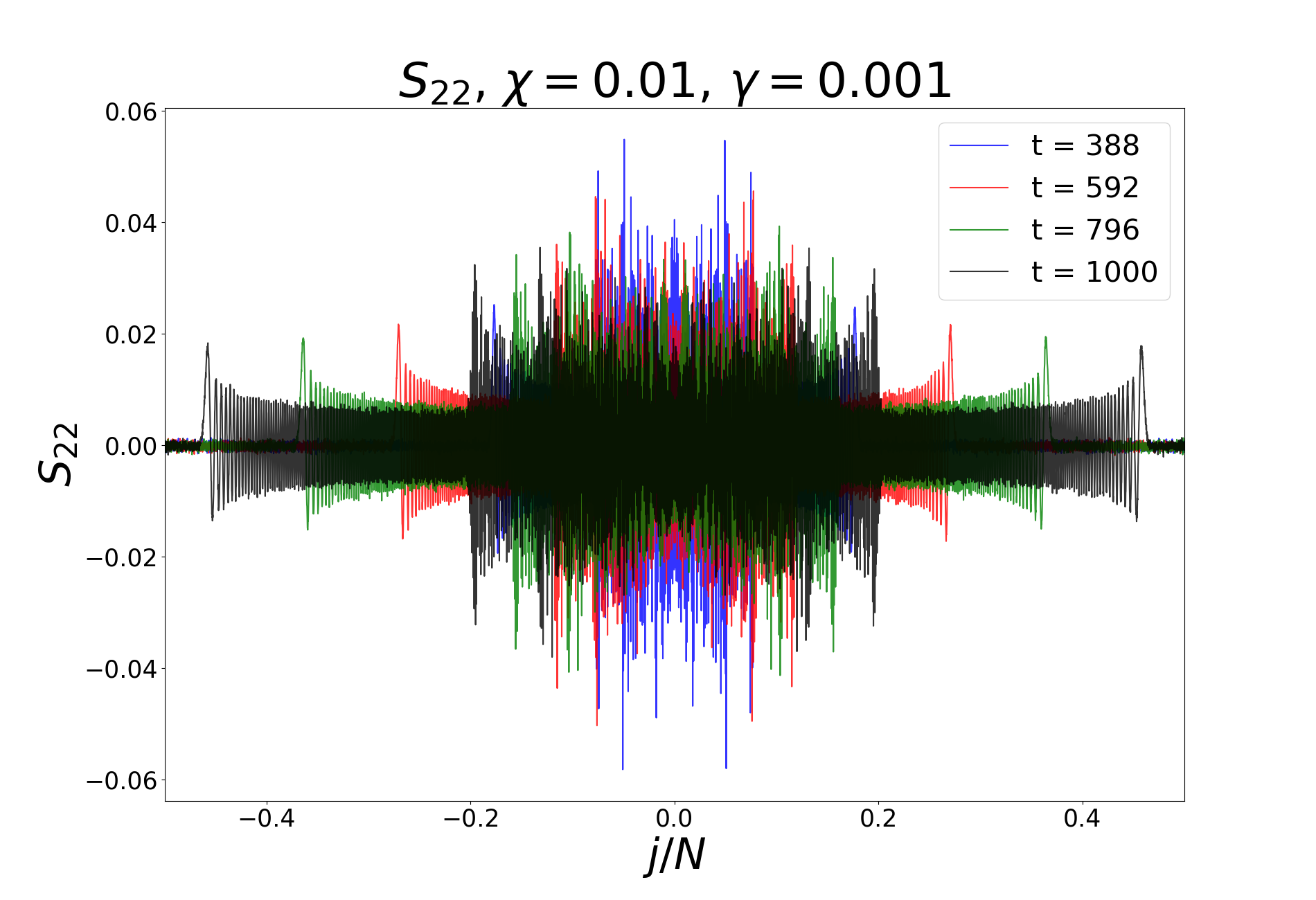}
		\includegraphics[scale=0.16]{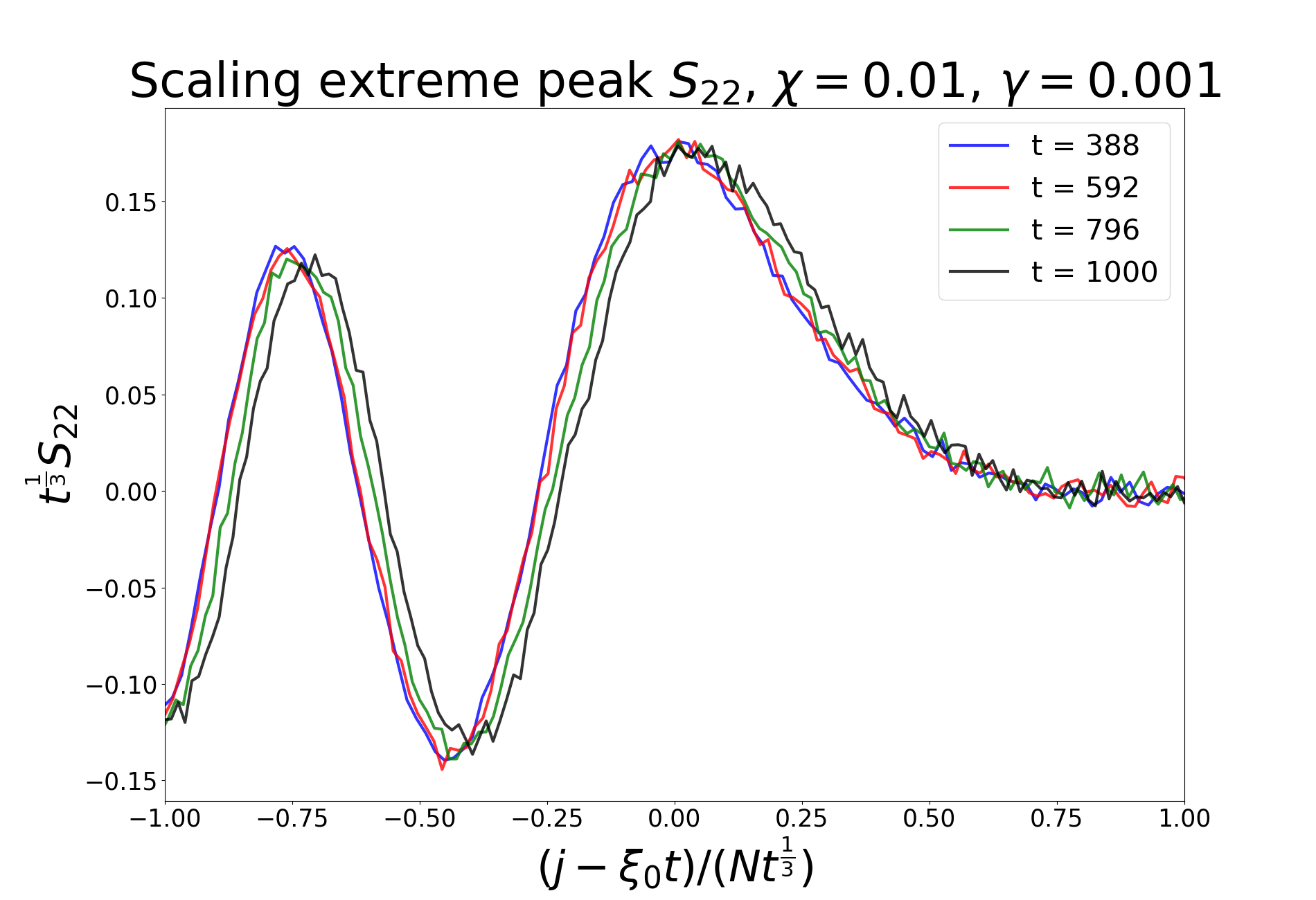}
		\includegraphics[scale=0.16]{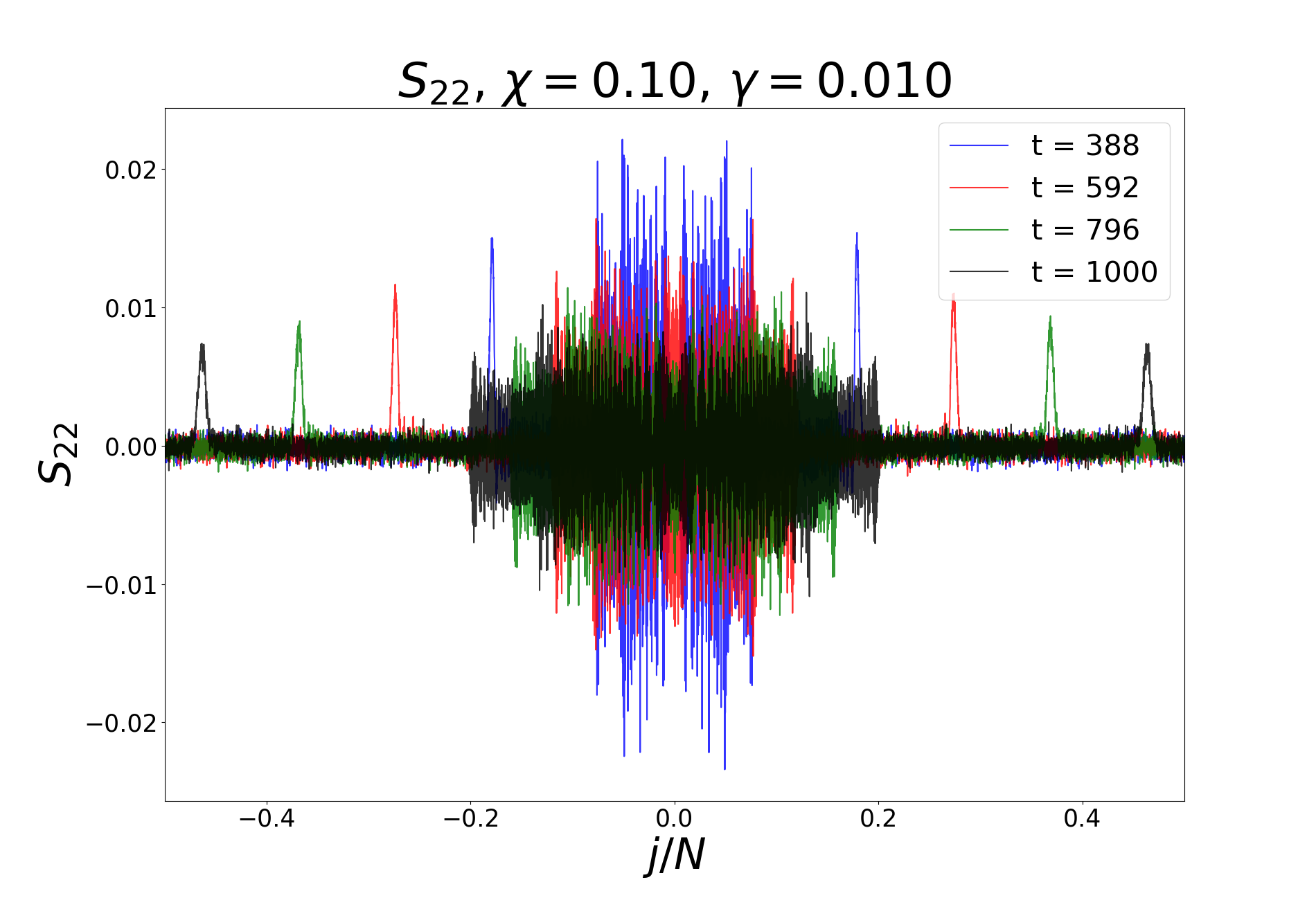}
		\includegraphics[scale=0.16]{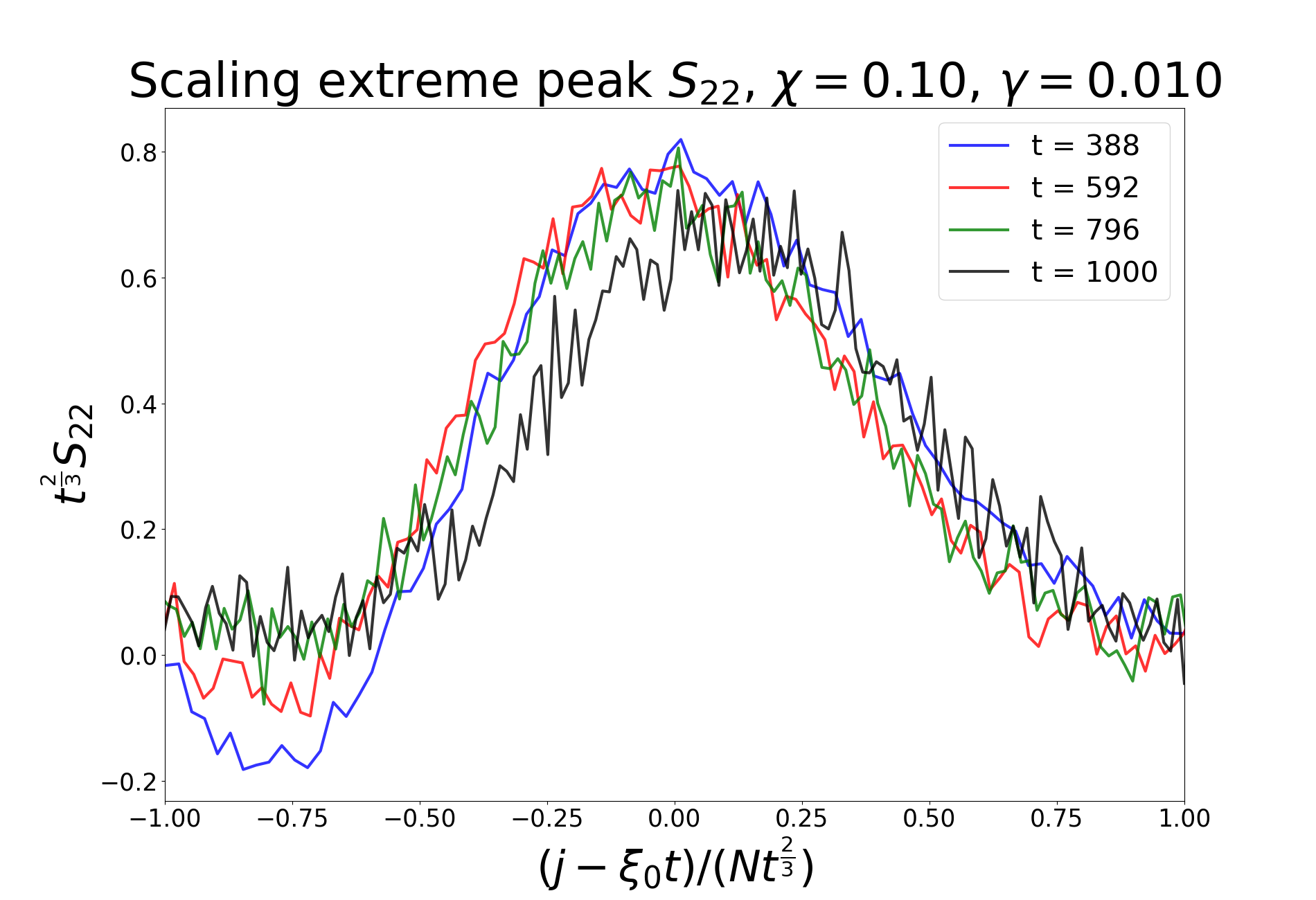}
		\caption{ Correlation function $S^{(N)}_{22}(j,t)$  for several values of times and for the  Hamiltonian \eqref{HN} with $\kappa_s$ as in  Example~\ref{example2},  $\chi=0.01$ and $\gamma=0.001$ in the top figure  and  $\chi=0.1$ and $\gamma=0.01$ in the lower figure.
		The  right  top  figure shows  the scaling of the fastest peak compatible  with the Airy   parametrix  and according to $t^{-2/3}$ in the lower figure. The speed $\xi_0$  of the fastest peak  is determined numerically. The decay rate of the  slower moving peaks  that are scaling like $t^{-1/4}$ in the linear case  (see Figure~\ref{Figure1}), is not very clear due to their highly oscillatory behaviour.
	}
\label{fig_ex2_non}
\end{figure}
\begin{figure}[ht]
	\centering
	%\figuretitle{Example~\ref{example2}}
			\includegraphics[scale=0.16]{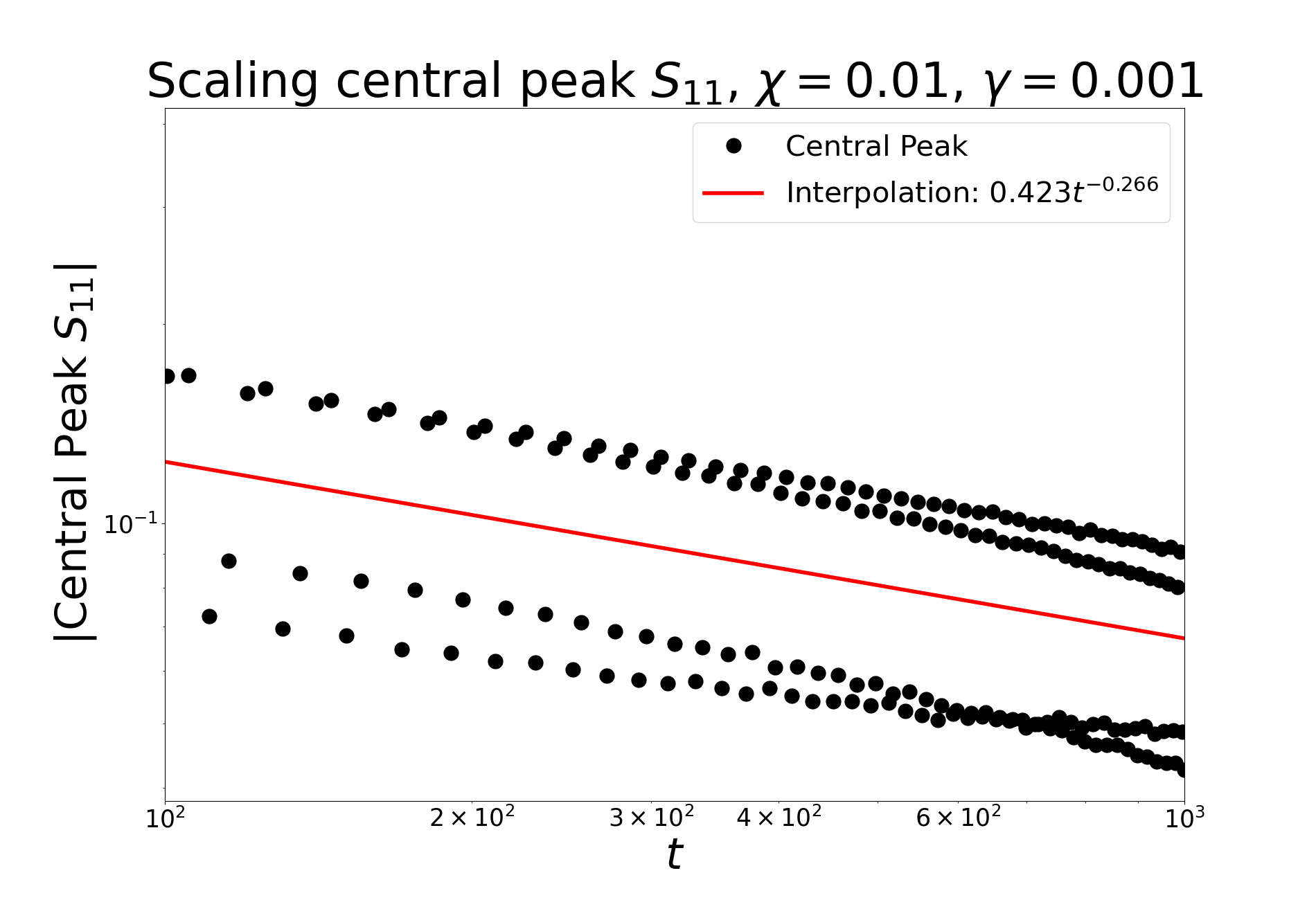}
		\includegraphics[scale=0.16]{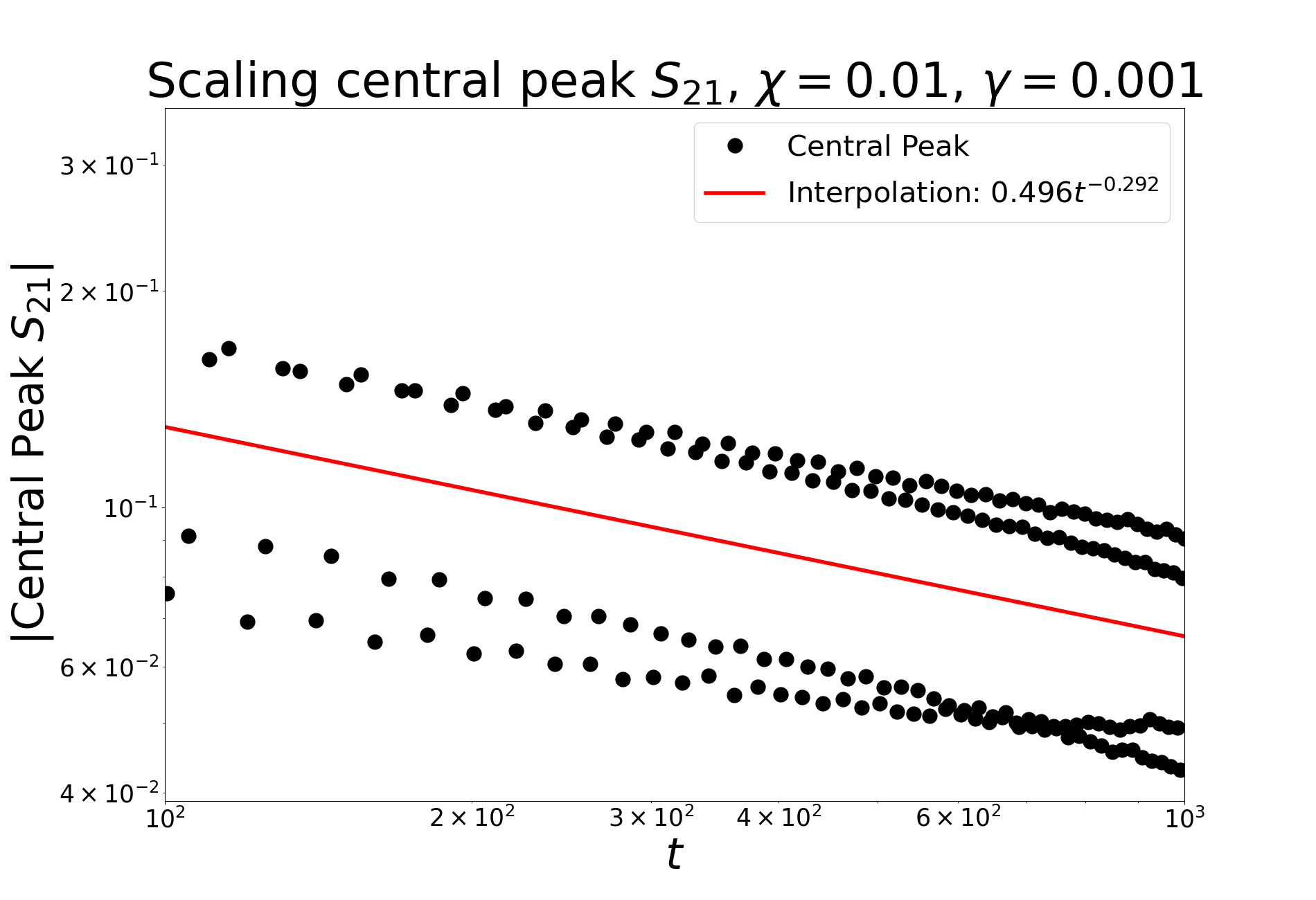}
		\caption{  Logarithmic plot  of the central peak of the example in Figure~\ref{fig_ex1_non} for $S_{11}(j,t)$ and $S_{21}(j,t)$	and several values of times. The peak is highly oscillatory and the oscillations are interpolated by the red line  that suggests a scaling of the correlation function $S_{11}(j,t)$ and  $S_{21}(j,t)$ near $j\sim0$ compatible with  $t^{-\frac{1}{4}}$. }
\label{fig_ex3_non}
\end{figure}

%We calculate the correlation functions as defined in \eqref{eq:correlation_functions} and study numerically their behaviour.
%We consider  Example~\ref{example1} and Example~\ref{example2}  with different strength of nonlinearity  namely
%\[
%\mbox{  $m=2, \, \kappa_1 = 1,$   $\kappa_2 = \frac{1}{4}$,}\;\;
%\begin{cases} 
%&\mbox{$\chi=0.01$ and $\gamma=0.001$}\\
%&\mbox{$\chi=0.1$ and $\gamma=0.01$ }
%\end{cases}
%\]
%\[
%\mbox{  $m=3, \, \kappa_1 = 1,$   $\kappa_2 = \frac{1}{8}$, $\kappa_2 = \frac{7}{72}$, }\;\;
%\begin{cases} 
%&\mbox{$\chi=0.01$ and $\gamma=0.001$}\\
%&\mbox{$\chi=0.1$ and $\gamma=0.01$ }
%\end{cases}
%\]

\begin{appendices}
\section{Proof of Proposition~\ref{prop:matrix_root}\label{Appendix_A}}
\begin{proof}
In view of the notation introduced in \eqref{A}, \eqref{v_a}, and \eqref{form:T} the proof of Proposition \ref{prop:matrix_root} amounts to showing the existence of $\tau_0,\ldots,\tau_m\in\R$ satisfying $\sum_{s=0}^m \tau_s =0$ such that
\begin{eqnarray}\label{eq:polynomial_eq}
Q(z^{-1}) \, Q(z) &=& \ell (z) \qquad \mbox{for all $z\in\C\setminus\{0\}$, where}\\ \nonumber
Q(z) = \tau_0 + \tau_1 z + \ldots + \tau_m z^m &\mbox{and}&\ell (z) = - \kappa_{m} z^{-m} - \ldots - \kappa_{1} z^{-1} + a_0 - \kappa_{1} z - \ldots - \kappa_{m} z^{m}\,.
\end{eqnarray}
The existence of the $\tau_j$'s is a consequence of the Fej\'er-Riesz lemma. For the convenience of the reader we present a proof following the presentation in~\cite[pg.~117~f]{book_Riesz_Nagy}. Denote by $P$ the polynomial of degree $2m$ given by $P(z) := z^m \ell (z)$. Observe that for all $x\in\R$ we have
$$\ell (e^{ix}) = a_0 - 2 \sum_{j=1}^m \kappa_j \cos (j x) \; \geq a_0 -2   \sum_{j=1}^m \kappa_j \; = \;  0.$$
By the positivity of $\kappa_1$ equality holds in the inequality above iff $\cos(x) = 1$.
This implies that $P$ has no zeros on the unit circle $|z|=1$ except for $z=1$. We denote by $\eta_k$, $1 \leq k \leq r_<$, the zeros of $P$ that lie within the unit disc $|\eta_k|<1$ and by $\xi_k$, $1 \leq k \leq r_>$, the zeros of $P$ with $|\xi_k|>1$, recorded repeatedly according to their multiplicities, so that
\begin{equation}\label{eq:polynomial_factorization}
P(z) = - \kappa_m (z-1)^{r_0} \prod_{k=1}^{r_<} (z - \eta_k) \prod_{k=1}^{r_>} (z - \xi_k)\,.
\end{equation}
Using the uniqueness of such a factorization for any polynomial together with the relation $z^{2m}P(z^{-1})=P(z)$ one obtains that $r_< = r_>$ and that the zeros can be listed in such a way that $\eta_k=\xi_k^{-1}$ for all $1 \leq k \leq r_<$. Moreover, we learn that $r_0$ is even with $1\leq \varrho := r_0/2 = m- r_<$. Now it follows from formula \eqref{eq:polynomial_factorization} that
\begin{equation*}\label{eq:factorization_ell}
l(z) \; = \; z^{-m} P(z) \; = \; c \, (z^{-1}-1)^{\varrho} (z-1)^{\varrho}\prod_{k=1}^{r_<} (z^{-1} - \xi_k) \prod_{k=1}^{r_<} (z - \xi_k)\qquad \mbox{with } \; c := -\kappa_m (-1)^{\varrho} \prod_{k=1}^{r_<} (-\xi_k^{-1}) \neq 0 \,.
\end{equation*}
Choosing $d \in \C$ with $d^2=c$ we see that $Q(z) := d (z-1)^{\varrho} \prod_{k=1}^{r_<} (z - \xi_k)$ satisfies \eqref{eq:polynomial_eq}. Next we show that the coefficients of the polynomial $Q$ are real. To this end observe that $P$ has real coefficients and therefore all non-real zeros of $P$ come in complex conjugate pairs with equal multiplicities. Therefore the polynomial $d^{-1}Q(z) =  \sum_{j=0}^{m} s_j z^j$ has only real coefficients $s_j$. Relation \eqref{eq:polynomial_eq} implies $a_0 = d^2 \sum_{j=0}^{m} s_j^2$. Consequently, $d^2$ is the quotient of two positive numbers and $d$ must be real. Thus we have $\tau_j = d s_j \in \R$ for all $0 \leq j \leq m$.
We complete the proof by arguing that $\sum_{s=0}^m \tau_s =0$ and $(\sum_{s=1}^m s \tau_s)^2 =\sum_{s=1}^m s^2 \kappa_s$ hold true.
This can be deduced from \eqref{eq:polynomial_eq} via $Q(1)^2=\ell(1)=0$ and $-2Q'(1)^2=\ell''(1)=-\sum_{s=1}^m 2s^2 \kappa_s$.
\end{proof}

%\section{Proof of Lemma~\ref{lemma_limit}\label{Appendix_B}}
%\section{Proof of Theorem~\ref{Theorem_Gaussian}\label{Appendix_C}}
%\begin{proof}
%\end{proof}
\section{Pearcey integral\label{Appendix_D}}
The general Pearcey integral is defined as 
\begin{equation}\label{pea}
\bar P(b,a):=\int_{-\infty}^\infty e^{i(t^4 +bt^2+at)}dt,\quad  0\leq \arg b\leq \pi,\;\; a\in\mathbb{R}.
\end{equation}
This integral decribes  cusp singularities in physical phenomena, like  the semiclassical limit of the   linear  Schr\"odinger equation.
The integral \eqref{pea}, after a rotation of the integration path through an angle of $\pi/8$ that removes the rapidly oscillatory term $e^{it^4}$, can  be written in the form $\bar P(b,a)=2e^{i\pi/8}P(be^{-i\pi/4},ae^{i\pi/8})$, with
\begin{equation}\label{pearcy}
P(b,a):=\int_0^\infty e^{-t^4 -bt^2}\cos(at)dt.
\end{equation}
We are interested in the case $b=0$. The corresponding integral is absolutely convergent for all complex values of $a$ and represents the analytic continuation of the Pearcey integral. For the Pearcey integral ${\mathcal P}_{+}(a)$ defined in \eqref{Pearcey} we obtain
\[
{\mathcal P}_{+}(a) = 2e^{i\pi/8}P(0,ae^{i\pi/8})\,.
\]
Note that the integral ${\mathcal P}_{-}(a)$, also defined in \eqref{Pearcey}, can be related to the function $P$ by rotating the integration path by an angle of $-\pi/8$. This gives ${\mathcal P_{-}}(a) = 2e^{-i\pi/8}P(0,ae^{-i\pi/8})$. From this we learn that ${\mathcal P_{-}}(\bar a)=\overline{\mathcal P_{+}(a)}$. On the reals we therefore have
\[
{\mathcal P}_{-}(a)=\overline{{\mathcal P}_{+}(a)}\,,\quad a\in\mathbb{R}\,.
\]
\section{Integrals of motion correlation functions}\label{Appendix_E}
Here, for completeness, we report the limiting correlation functions for the integral of motions as defined in \eqref{CIM}. Using the notation $f(k)=|\omega(k)|$ introduced in Lemma~\ref{fandtheta} they are:
\begin{equation}
\begin{split}
	S_{k+3,n+3}(j,t)= &\frac{1}{2\beta^2}\int_{0}^{1}\int_{0}^{1}\cos\left(f(x)t\right)\cos\left(f(y)t\right)\cos\left(2\pi x(j-n)\right)\cos\left(2\pi y(j+k)\right) \\ 
	&+ \cos\left(f(x)t\right)\cos\left(f(y)t\right)\cos\left(2\pi xj\right)\cos\left(2\pi y(j+k-n)\right) \\
	& + \sin\left(f(x)t\right)\sin\left(f(y)t\right)\cos\left(2\pi x(j-n)\right)\cos\left(2\pi y(j+k)\right)\cos(\theta(x))\cos(\theta(y)) \\
	& + \sin\left(f(x)t\right)\sin\left(f(y)t\right)\sin\left(2\pi x(j-n)\right)\sin\left(2\pi y(j+k)\right)\sin(\theta(x))\sin(\theta(y)) \di x \di y
	\end{split}
\end{equation}
for $k,n\leq \frac{N-1}{2}$,
\begin{equation}
\begin{split}
S_{n+3,k+3}(j,t) =& 
\frac{1}{2\beta^2} \int_0^1 \int_0^1 f(x)f(y)\sin\left(f(x)t\right)\sin\left(f(y)t\right)\sin\left(2\pi xj\right)\sin\left(2\pi yj\right)\sin\left(2\pi xn\right)\sin\left(2\pi yk\right) \\ 
& + f^2(x)\cos(f(x)t)\cos(f(y)t)\cos\left(2\pi xj\right)\cos\left(2\pi yj\right)\sin\left(2\pi yn\right)\sin\left(2\pi yk\right) \di x \di y
\end{split}
\end{equation}
for $k,n>\frac{N-1}{2}$ and 
\begin{equation}
	\begin{split}
		S_{n+3,k+3}(j,t) &=
		\frac{1}{2\beta^2} \int_0^1\int_0^1 \cos\left(2\pi xj - \theta(x)\right)\cos\left(2\pi yj\right) \sin\left(2\pi yk\right)\sin\left(2\pi yn\right)\sin((f(x)+ f(y))t ) \\ 
		&+ \cos\left(2\pi xj - \theta(x)\right)\sin\left(2\pi yj\right) \sin\left(2\pi yk\right)\cos\left(2\pi yn\right)\sin((f(x) -f(y))t ) \di x \di y
	\end{split}
\end{equation}
for $k>\frac{N-1}{2}, n \leq \frac{N-1}{2}$.

%\todo{Th:Do we want/need the relation to $_0F_2$?}
%In particular for $b=0$ one obtains the integral absolutely convergent for all complex values of $a$ and represents the analytic continuation of the Pearcey integral ( see e.g. \cite{Kaminski},\cite{Pagola}).
%Therefore the Pearcey integral ${\mathcal P}_{\pm}(a)$ defined in \eqref{Pearcey} take the form

%\[
%P(0,a)\:=\frac{1}{4}\Gamma\left(\frac{1}{4}\right){}_0F_2\left(\frac{2}{4},\frac{3}{4};\left(\frac{a}{4}\right)^4\right)-\frac{a^2}{8}\Gamma\left(\frac{3}{4}\right){}_0F_2\left(\frac{5}{4},\frac{6}{4};\left(\frac{a}{4}\right)^4\right).
%\]
%This integral is absolutely convergent for all complex values of $a$ and represents the analytic continuation of the Pearcey integral ( see e.g. \cite{Kaminski},\cite{Pagola}).
%Therefore the Pearcey integral ${\mathcal P}_{\pm}(a)$ defined in \eqref{Pearcey} take the form
%\[
%{\mathcal P}_{+}(a)=\frac{e^{i\pi/8}}{2}\Gamma\left(\frac{1}{4}\right){}_0F_2\left(\frac{2}{4},\frac{3}{4};i\left(\frac{a}{4}\right)^4\right)-e^{\pi i\frac{3}{8}}\frac{a^2}{4}\Gamma\left(\frac{3}{4}\right){}_0F_2\left(\frac{5}{4},\frac{6}{4};i\left(\frac{a}{4}\right)^4\right),
%\]
%and 
%\[
%{\mathcal P}_{-}(a)=\overline{{\mathcal P}_{+}(-a)}=\overline{{\mathcal P}_{+}(a)},\quad a\in\mathbb{R}.
%\]

\section{Numerical Computation}
The  numerical computations have been implemented with  \texttt{Python} software, all codes are available on GitHub \cite{FPUT_repo}.
Fig. \ref{Figure1}--\ref{fig_ex2} are the result of the numerical evaluation via the standard routine \texttt{numpy.trapz} of  the integrals in  \eqref{eq:expCpp}--\eqref{S33} for various values of $j$ and $t$ and then we just added the Airy function \eqref{Airy0} and the Pearcey integral \eqref{Pearcey0}.

To obtain Fig. \ref{fig_ex1_non} we proceed in the following way. First we sampled a random initial data according to the Gibbs measure defined by the corresponding harmonic part of \eqref{HN}, namely the Hamiltonian of Example \ref{example1} with $m=2$. We let these data evolve according to the Hamilton equations of \eqref{HN} and compute the values of the correlations function. Then we repeated this procedure  $4\times 10^6$ times and we averaged the values of the correlations functions. On the left panel we plot the correlation functions,  instead on the right one we focus on the extreme peak and we guess a proper scaling depending on the size of the perturbation. Fig. \ref{fig_ex2_non} is made in a similar way, where now  the nonlinear potential has the same harmonic part as Example \ref{example2}.

In Fig. \ref{fig_ex3_non} we focus our attention on the central peak of the chain with potential as is Fig. \ref{fig_ex1_non}. We follow the same procedure as before and plot in logarithmic scale the average scaling of the highest peak in the center of the chain. We decide to plot the  average height of this peak since it is highly  oscillatory and it is difficult to precisely track the oscillations. 
\end{appendices}

\vskip 0.8cm
{\noindent \bf Acknowledgments.}
This manuscript was initiated during the research in pairs that took place in May 2019 at the Centre International des Rescontres math\'ematiques (CIRM), Luminy, France  during the chair Morlet semester "Integrability and randomness in mathematical physics".
 The authors  thank CIRM for the generous support, excellent work environment, and kind hospitality.
  K.M. was supported in part by the National Science Foundation under grant DMS-1733967.
T.G. and G.M. acknowledge support from the European Union's H2020 research and innovation program under the Marie Sk\l owdoska--Curie grant No. 778010 {\em  IPaDEGAN}.\\  
 We thank  Manuela Girotti for related initial calculations in the case of nearest-neighbor interactions. 
 %which helped in particular with Lemma \ref{lemma_limit}.  
 We also thank Giuseppe Pitton for sharing his codes with us, and for useful discussions.

%\vskip 2cm
%\noindent{\bf Acknowledgements.}
%T.G. and G.M. acknowledges  the support of the H2020-MSCA-RISE-2017 PROJECT No. 778010 IPADEGAN. K.M. was supported in part by the National Science Foundation under grant DMS-1733967.  
%This work was started while visiting CIRM, Luminy, France as research in pairs.  We  thank CIRM for excellent working conditions and generous support.
%We wish to thank  Manuela Girotti  for useful feedback  and Giuseppe Pitton  for sharing his codes with us.
%\newpage

\vskip 0.7cm
	\bibliographystyle{siam}
	\bibliography{BIB_Longrange.bib}
%\begin{thebibliography}{10}
%
%
%
%\bibitem{gray2006toeplitz}
%R. Gray.
%\newblock Toeplitz and circulant matrices: A review.
%\newblock {\em Foundations and Trends in Communications and
%  Information Theory}, 2(3):155--239, 2006.
%  
%  \bibitem{Mazur}Peter Mazur and Elliott Montroll,  Poincar\'e Cycles, Ergodicity, and Irreversibility in Assemblies of Coupled Harmonic Oscillators
%\newblock {\em J. Math. Phys.} 1,  (1960), p. 70-84. doi: 10.1063/1.1703637
%\bibitem{KuduDahar} Aritra Kundu and Abhishek Dhar, Equilibrium dynamical correlations in the Toda chain and other integrable models, 
% \newblock {\em Physical Review E}, {\bf 4}, 062130 (2016).
%
%
%
%\end{thebibliography}
%
%

\end{document}